\documentclass[3p]{elsarticle}

\usepackage[british]{babel}
\usepackage{amsmath,amssymb}
\usepackage{amsthm}
\usepackage{booktabs}
\usepackage{bm}
\usepackage{graphicx}
\usepackage{pgfplots}
\pgfplotsset{compat=1.17}
\usepackage[most]{tcolorbox}
\usepackage[section]{placeins}

\usepackage{hyperref}

\theoremstyle{plain}
\newtheorem{theorem}{Theorem}
\newtheorem{proposition}{Proposition}
\newtheorem{lemma}{Lemma}
\newtheorem{corollary}{Corollary}
\theoremstyle{remark}
\newtheorem{remark}{Remark}

\newcommand{\bfe}{\bm{e}}
\newcommand{\bfu}{\bm{u}}
\newcommand{\Id}{\mathsf{I}}
\newcommand{\Lam}{\bm{\Lambda}}
\newcommand{\Tm}{\mathsf{T}}
\newcommand{\Am}{\mathsf{A}}
\newcommand{\Pm}{\mathsf{P}}
\newcommand{\Cx}{C_x}
\newcommand{\Cy}{C_y}
\newcommand{\Cz}{C_z}
\newcommand{\avg}[1]{\left\langle #1 \right\rangle}

\journal{Journal of Computational Physics}

\begin{document}

\begin{frontmatter}

\title{When does the moment basis matter in lattice Boltzmann methods?}

\author{Alessandro De Rosis}
\ead{alessandro.derosis@manchester.ac.uk}
\affiliation{organization={Department of Mechanical and Aerospace Engineering, The University of Manchester},
             addressline={Oxford Road},
             city={Manchester},
             postcode={M13 9PL},
             country={United Kingdom}}

\begin{abstract}
Whether orthogonalising the moment basis changes a central-moment lattice Boltzmann
scheme is decided by the geometry of the velocity set. The basis enters the collision
only through the conjugate $\Tm^{-1}\Lam\Tm$, so two bases related by a constant
matrix $\Am$ define exactly the same scheme if and only if
$[\Lam,\Am]_{ij}=\Am_{ij}(\lambda_i-\lambda_j)$ vanishes, and on physical states only
its columns for the non-conserved moments matter. We prove that on lattices with velocities in
$\{0,\pm1\}^d$ and symmetric weights orthogonalisation reaches the shear moments only
through the face and body diagonals: whatever the orthogonalisation, the shear rate is
free on D2Q9 and D3Q15, whereas on D3Q19 and D3Q27 the scheme changes whenever the
fourth-order moments do not relax at the shear rate. On rectangular lattices, where the
symmetry is broken, an independent bulk rate yields an orthogonal counterpart that
does not share the transport coefficients of the published scheme, which is itself
orthogonal in the inner product of the rest equilibrium up to the one relaxation entry
that consistency requires. Where the
formulations differ but share their hydrodynamics, a basis orthogonal in the inner
product of the rest equilibrium, which guarantees linear stability at rest, was never
less robust in the least stable flow direction by more than a few per cent, often much
more robust, and the more accurate against a Taylor--Green DNS. Computed through the
non-orthogonal transforms with a conjugated relaxation matrix, it costs at most $2\%$
more on D3Q27.
\end{abstract}


\end{frontmatter}

\section{Introduction}

Multiple-relaxation-time (MRT) lattice Boltzmann
schemes perform the
collision on moments obtained from the populations by an invertible linear map,
relaxing each towards its equilibrium at its own rate~\cite{10.1098/rsta.2001.0955, PhysRevE.102.023306}. Such a scheme is fixed by two
choices: the basis in which the moments are taken and the diagonal matrix of rates
that acts on them. Central moments, taken in the frame moving with the
fluid~\cite{geier2006,premnath2009,premnath2011,derosis2016,derosis2017}, and
cumulants~\cite{geier2015}
have multiplied the first choice, and within each family one may or may not
orthogonalise, usually by Gram--Schmidt against the lattice weights. Comparisons
between formulations routinely attribute differences in accuracy, robustness and
stability to the basis~\cite{coreixas2020,chavez2018}. In particular,
orthogonalisation is said to introduce a spurious coupling between moments of
different order that narrows the stability range, and this is one of the two
arguments, with the aspect-ratio dependence of orthogonal transforms, for
non-orthogonal constructions on rectangular and cuboid
lattices~\cite{yahia2021,yahia2021cuboid}.

What has been established rigorously is narrower. Schemes relaxing in a moving frame
reproduce the equivalent equations of the MRT method to second and third
order~\cite{dubois2015}, an asymptotic statement about truncated macroscopic
behaviour. The common formalism~\cite{coreixas2020,coreixas2019} compares
post-collision populations of raw, Hermite, central and cumulant formulations case
by case. Li \& Shan~\cite{lishan2021} constrained the second choice, showing that
independent rates may be assigned per irreducible representation of the group of rotations of three-dimensional space, $SO(3)$, and no
more finely. Bellotti \textit{et al.}~\cite{bellotti2022} gave every lattice
Boltzmann scheme an exact finite difference form. Dubois \textit{et al.}~\cite{dubois2026} found the second-order equivalent equations of D2Q9
insensitive to the equilibrium adopted. What is missing is a criterion for when two
moment formulations are \emph{the same scheme}, identical as maps on populations,
and, where they are not, a comparison that isolates the basis from everything else.

This paper supplies both. The observation it rests on is elementary: in population
space the basis enters only through the conjugate $\Tm^{-1}\Lam\Tm$, a form that is
standard~\cite{lallemand2000,coreixas2019}. Its consequences for the choice of
basis, however, have not been drawn. The contributions are:
\begin{enumerate}
\item an exact criterion for equivalence, $[\Lam,\Am]=0$, stated as an elementary
  conjugation lemma, with a form for physical states, a coupling-graph form, a bound
  on the difference, and a projection test that does not depend on how the
  orthogonalisation is carried out (Section~\ref{sec:gauge});
\item a theorem that decides from the geometry of the velocity set alone whether
  orthogonalisation reaches the shear moments: never on D2Q9 and D3Q15, always on
  D3Q19 and D3Q27, for any Gram--Schmidt order and any symmetric weights
  (Theorem~\ref{thm:geom}), with the couplings in closed form on rectangular
  lattices, where the symmetry is broken (Section~\ref{sec:when});
\item a guarantee of linear stability at rest for bases orthogonal in the inner
  product of the rest equilibrium, inherited by every basis gauge equivalent to one
  (Proposition~\ref{prop:rest});
\item a comparison of the formulations where they differ, by linear analysis,
  nonlinear simulation and comparison with direct numerical simulation on four
  lattices, including the published rectangular scheme itself, which is shown to be an
  orthogonal scheme with a single off-diagonal relaxation entry
  (Section~\ref{sec:stability}, Table~\ref{tab:summary});
\item an implementation result: any formulation can be computed through the
  transforms of any basis in its gauge class, so that the orthogonal scheme need not
  cost more than the non-orthogonal one (Corollary~\ref{cor:impl},
  Section~\ref{sec:cost}).
\end{enumerate}

The rest of the paper is organised as follows. Section~\ref{sec:gauge} derives the criterion, Section~\ref{sec:when} applies it,
Section~\ref{sec:stability} compares the formulations, and Section~\ref{sec:discussion}
concludes with recommendations. The appendices examine the choice of inner product,
collect the convergence studies, document the verification and a lid-driven cavity on the
rectangular lattice, and give the details of the wavevector search on D3Q19 and of the
cost of the collision kernels.

\section{Gauge equivalence in moment space}
\label{sec:gauge}
We first establish the general criterion for when a change of moment basis is a pure representation change and when it alters the collision operator. We then restrict it to physical states, recast it in a form that does not depend on how the orthogonal basis is constructed, show that any formulation can be computed through any basis in its gauge class, and identify the bases whose rest state is guaranteed to be stable.

\subsection{Setting and notation}

Let $\{\bfe_i\}_{i=0}^{q-1}$ be the discrete velocity set of a $q$-velocity
lattice with weights $w_i$, and let $f=(f_0,\dots,f_{q-1})^{\!\top}$ denote the
populations. A moment-space scheme is specified by an invertible transformation
matrix $\Tm$, whose rows are the basis polynomials evaluated on the velocity set,
together with a diagonal relaxation matrix
$\Lam=\mathrm{diag}(\lambda_1,\dots,\lambda_q)$. Writing $k=\Tm f$ for the
moments, $k^{\mathrm{eq}}=\Tm f^{\mathrm{eq}}$ for their equilibria and
$k^{F}=\Tm \mathcal{F}$ for the moments of a forcing term $\mathcal{F}$, the
collision step reads
\begin{equation}
  k^{\star}=(\Id-\Lam)\,k+\Lam\,k^{\mathrm{eq}}
            +\left(\Id-\tfrac{1}{2}\Lam\right)k^{F},
  \label{eq:collision-moments}
\end{equation}
and the post-collision populations are $f^{\star}=\Tm^{-1}k^{\star}$.

For central moments the basis polynomials are evaluated at the peculiar velocity
$\bm{C}_i=\bfe_i-\bfu$, with $\Cx=e_{ix}-u_x$ and so on. The first $1+d$ basis
elements, in $d$ dimensions, are the conserved moments $1,\Cx,\Cy(,\Cz)$. We take
$\lambda=1$ on them, so that $(\Id-\tfrac{1}{2}\Lam)k^F$ reproduces the standard
half-force correction.

Substituting $k=\Tm f$ into \eqref{eq:collision-moments} and acting with
$\Tm^{-1}$ gives the population-space form
\begin{equation}
  f^{\star}=f+\Pm\left(f^{\mathrm{eq}}-f\right)
              +\left(\Id-\tfrac{1}{2}\Pm\right)\mathcal{F},
  \qquad
  \boxed{\;\Pm=\Tm^{-1}\Lam\,\Tm\;}
  \label{eq:collision-populations}
\end{equation}
The basis enters only through the conjugate $\Pm$, and everything below follows
from that.

\subsection{The equivalence criterion}

\begin{lemma}[Conjugation]
\label{thm:gauge}
Let $\Tm_1$ and $\Tm_2$ be two invertible moment bases related by a constant
(velocity-independent) matrix $\Am=\Tm_2\Tm_1^{-1}$, and let the two schemes use
the same relaxation matrix $\Lam$, the same equilibrium $f^{\mathrm{eq}}$ and the
same forcing $\mathcal{F}$ in population space. Then the two schemes produce
identical post-collision populations for every $f$, every $\bfu$ and every
$\mathcal{F}$ if and only if
\begin{equation}
  [\Lam,\Am]=\Lam\Am-\Am\Lam=0 .
  \label{eq:criterion}
\end{equation}
\end{lemma}

\begin{proof}
From $\Tm_2=\Am\Tm_1$,
\begin{equation*}
  \Pm_2=\Tm_2^{-1}\Lam\,\Tm_2=\Tm_1^{-1}\Am^{-1}\Lam\,\Am\,\Tm_1 ,
\end{equation*}
so $\Pm_2=\Pm_1$ if and only if $\Am^{-1}\Lam\Am=\Lam$, which is
\eqref{eq:criterion}. If $\Pm_1=\Pm_2$ the two right-hand sides of
\eqref{eq:collision-populations} coincide. Conversely, taking $f=f^{\mathrm{eq}}$
leaves $f^{\star}=f^{\mathrm{eq}}+(\Id-\tfrac12\Pm)\mathcal{F}$, and equality for
every $\mathcal{F}$ forces $\Pm_1=\Pm_2$.
\end{proof}

``The same forcing'' means the same population-space vector $\mathcal{F}$. Its
moment-space representations $\Tm_1\mathcal{F}$ and $\Tm_2\mathcal{F}$ differ, but
they are two descriptions of one forcing, not two forcings; the same holds for the
equilibrium.

Condition \eqref{eq:criterion} places $\Am$ in the centraliser of $\Lam$ in the group
$GL(q)$ of invertible $q \times q$ matrices; equivalently, $\Am$ is block diagonal with respect to the eigenspaces of
$\Lam$. We call such a change of basis a \emph{gauge} transformation in the
representation-theoretic sense: moment coordinates related by $\Am$ with
$[\Lam,\Am]=0$ describe one population-space collision map, much as
gauge-related potentials describe one field. The freedom is global, since $\Am$ is
constant, rather than local. For diagonal $\Lam$ the condition is entirely
explicit.

\begin{lemma}[Factorisation]
\label{lem:factor}
For diagonal $\Lam$ and any matrix $\Am$,
\begin{equation}
  [\Lam,\Am]_{ij}=\Am_{ij}\left(\lambda_i-\lambda_j\right).
  \label{eq:factorisation}
\end{equation}
\end{lemma}

Each discrepancy between two formulations is thus a product of a \emph{geometric}
factor, the off-diagonal entry $\Am_{ij}$ generated by the change of basis, and a
\emph{kinetic} one, the rate contrast $\lambda_i-\lambda_j$. Lemma~\ref{thm:gauge} treats $f$, $\bfu$ and $\mathcal{F}$ as independent,
which is what makes its converse hold. In a physical scheme $\bfu$ is the
velocity of $f$, and not every entry of the commutator can then act.

\begin{corollary}[Physical states]
\label{cor:physical}
Let $\rho$ and $\bfu$ be the density and velocity of $f$, and let there be no
forcing. Then the two schemes coincide for every $f$ if and only if $[\Lam,\Am]_{ij}=0$ for every $i$ and every
non-conserved $j$. A forcing term, whose mass moment vanishes but whose momentum
moment does not, activates the momentum columns as well.
\end{corollary}

\begin{proof}
The difference of the two post-collision states is
$(\Pm_2-\Pm_1)(f^{\mathrm{eq}}-f)=\Tm_1^{-1}\Am^{-1}[\Lam,\Am]\,v$ with
$v=\Tm_1(f^{\mathrm{eq}}-f)$. The conserved components of $v$ vanish, because
$f^{\mathrm{eq}}$ and $f$ share density and velocity, while the non-conserved
components range freely with $f$. The difference therefore vanishes for every $f$
exactly when the non-conserved columns of $[\Lam,\Am]$ do.
\end{proof}

The same identity measures how far apart two formulations are. With $\Pi_{\mathrm{nc}}$
the projector onto the non-conserved moments,
\begin{equation}
  \bigl\|f^{\star}_2-f^{\star}_1\bigr\|\le
  \bigl\|\Tm_1^{-1}\Am^{-1}\bigr\|\,\bigl\|[\Lam,\Am]\,\Pi_{\mathrm{nc}}\bigr\|\,
  \bigl\|\Tm_1\bigl(f^{\mathrm{eq}}-f\bigr)\bigr\|
  \label{eq:distance}
\end{equation}
on physical states without forcing, so that the size of the non-conserved columns of
the commutator, multiplied by the non-equilibrium, bounds the difference at every
step.

The criterion also admits a combinatorial form, and one in which all bases
collapse together.

\begin{corollary}[Coupling graph]
\label{cor:graph}
Define the undirected \emph{coupling graph} $\mathcal{G}(\Am)$ on the $q$ moments,
with an edge $\{i,j\}$ whenever $\Am_{ij}\neq 0$ or $\Am_{ji}\neq 0$. Then
$[\Lam,\Am]=0$ if and only if $\Lam$ is constant on every connected component of
$\mathcal{G}(\Am)$. The condition of Corollary~\ref{cor:physical} is the same
statement with the conserved vertices deleted.
\end{corollary}

\begin{corollary}[Single relaxation time]
\label{cor:bgk}
If $\Lam=\omega\Id$ then $[\Lam,\Am]=0$ for every $\Am$. If only the
non-conserved moments share a common rate, the schemes still coincide on physical
states without forcing, for every $\Am$ that leaves the conserved moments
unchanged, as any orthogonalisation that begins with them does.
\end{corollary}

\begin{corollary}[Implementation]
\label{cor:impl}
If $\Tm_2=\Am\Tm_1$ with $\Am$ constant, the scheme with basis $\Tm_2$ and
relaxation matrix $\Lam$ is computed exactly by the transforms of $\Tm_1$ with the
constant relaxation matrix $\Am^{-1}\Lam\Am$ in place of $\Lam$.
\end{corollary}

\begin{proof}
$\Pm_2=\Tm_2^{-1}\Lam\Tm_2=\Tm_1^{-1}\left(\Am^{-1}\Lam\Am\right)\Tm_1$.
\end{proof}

A formulation can therefore be run through the cheapest transform in its gauge class;
only the relaxation step changes, and it stays sparse when $\Am$ is
(Section~\ref{sec:cost}).

Two remarks delimit the theorem, the first its hypothesis and the second its
strength.

\begin{remark}[Central versus raw moments]
\label{rem:central}
Lemma~\ref{thm:gauge} requires $\Am$ to be constant. The shift from raw to
central moments is $\Tm_{\mathrm{CM}}=\mathsf{N}(\bfu)\,\mathsf{M}$ with
$\mathsf{N}$ depending on the fluid velocity, so it is not a gauge transformation
and the theorem does not apply to it. Equivalently, the central-moment collision
can be written as a raw-moment one that relaxes towards a generalised equilibrium
depending on the local non-equilibrium~\cite{asinari2008}: the velocity dependence
can be moved into the equilibrium, but not removed by a constant change of basis.
This is the precise sense in which central moments are substantive while
orthogonalisation, as shown below, often is not.
\end{remark}

\begin{remark}[Strength of the statement]
\label{rem:strength}
Lemma~\ref{thm:gauge} and Corollary~\ref{cor:physical} are identities between
maps, holding at all orders in Mach number. This is strictly stronger than
equality of equivalent partial differential equations truncated at some order.
Every inequivalence reported in Section~\ref{sec:when} involves at least one entry in
a non-conserved column, and so holds for physical states without forcing.
\end{remark}

\begin{remark}[Boundaries]
\label{rem:boundaries}
The collision is local, so Lemma~\ref{thm:gauge} and its corollaries hold node by
node in bounded domains as well. Boundary conditions that act on populations, such
as bounce-back, do not involve the basis at all; boundary schemes written in moment
space are covered only if they use the same conjugate $\Pm$. The tests of
Section~\ref{sec:stability} are periodic except the lid-driven cavity of
Section~\ref{sec:rectstab}, whose half-way bounce-back walls act on populations.
\end{remark}

\paragraph{Three levels of equivalence}
The results distinguish three senses in which two formulations can coincide.
(i)~\emph{Operator equivalence}, $\Pm_1=\Pm_2$: the two collision maps agree for
every $f$, $\bfu$ and $\mathcal{F}$ (Lemma~\ref{thm:gauge}). (ii)~\emph{Equivalence
on physical states}: the maps agree whenever $\bfu$ is the velocity of $f$ and no
force acts, which requires only the non-conserved columns of $[\Lam,\Am]$ to vanish
(Corollary~\ref{cor:physical}); this is the sense that matters in a simulation.
(iii)~\emph{Numerical equivalence}: formulations equivalent in either sense still
differ at round-off in floating point, because $\Tm_1^{-1}$ and $\Tm_2^{-1}$ are
evaluated differently and are differently conditioned (Section~\ref{sec:cost} and
\ref{sec:verification}).
Each result below is stated at the strongest level that holds. D2Q9 with a single
shear rate is equivalent at level~(i); D2Q9 with the fourth-order moment relaxed at
the bulk rate, and D3Q27 with a two-rate central matrix, only at level~(ii); every
inequivalence reported fails already at level~(ii).

\subsection{A basis-independent obstruction}

Corollary~\ref{cor:graph} depends on $\Am$, hence on which orthogonal basis is
chosen, and Gram--Schmidt is order dependent. The following criterion removes
that dependence.

Equip $\mathbb{R}^{q}$ with the lattice-weight inner product
$\avg{a,b}=\sum_i w_i\,a_i b_i$. Let $\mathcal{D}$ be the deviatoric second-order
subspace, spanned by the traceless symmetric second-order monomials, which relax
at the shear rate $\omega$, and let $P_{\mathcal{D}}$ be the orthogonal projector
onto it. On the square and cubic lattices $\mathcal{D}$ is orthogonal to the
conserved moments, by parity and by the symmetry of the weights under exchange of
axes.

\begin{proposition}[Projection criterion]
\label{prop:projection}
Let an orthogonal basis be related to the monomial basis by a constant invertible
$\Am$ that leaves the conserved moments unchanged, each element relaxing at the
rate of the monomial it replaces, and let $\mathcal{D}$ be orthogonal to the
conserved moments. If some monomial $m$ that
relaxes at a rate other than $\omega$ has $P_{\mathcal{D}}\,m\neq0$, the
orthogonal and non-orthogonal schemes differ on physical states without forcing,
whichever orthogonal basis is used. Conversely, if every monomial outside
$\mathcal{D}$ is orthogonal to $\mathcal{D}$, Gram--Schmidt in any order leaves
the deviatoric block uncoupled, and $\omega$ may be chosen freely.
\end{proposition}

\begin{proof}
Suppose the schemes coincide on physical states. By Corollary~\ref{cor:physical}
and Lemma~\ref{lem:factor}, $\Am_{ij}=0$ whenever $j$ is non-conserved and
$\lambda_j\neq\lambda_i$, so each non-conserved element of the orthogonal basis
lies in $V_\lambda+\mathcal{C}$, where $V_\lambda$ is spanned by the non-conserved
monomials of its rate $\lambda$ and $\mathcal{C}$ by the conserved ones. Being
orthogonal to the conserved elements, which span $\mathcal{C}$, the elements of
rate $\lambda$ span $P_{\mathcal{C}^\perp}V_\lambda$, and these subspaces are
mutually orthogonal for distinct rates. Since $\mathcal{D}\perp\mathcal{C}$, every
$d\in\mathcal{D}$ lies in $P_{\mathcal{C}^\perp}V_\omega$, and for $m\in V_\mu$
with $\mu\neq\omega$,
$\avg{d,m}=\avg{d,P_{\mathcal{C}^\perp}m}=0$. Hence $P_{\mathcal{D}}m=0$, a
contradiction. For the converse, proceed along the Gram--Schmidt order: if every
earlier element lies either in $\mathcal{D}$ or in the span of the monomials
outside it, a new monomial has zero projection onto the elements of the other
class, so its orthogonalised form stays within its own class, and $\Am$ has no
entry connecting $\mathcal{D}$ to the rest.
\end{proof}

\paragraph{Scope}
Lemma~\ref{thm:gauge} and its corollaries hold for any two invertible bases
related by a constant matrix and any diagonal relaxation matrix.
Proposition~\ref{prop:projection} is narrower. It concerns an orthogonal basis
obtained from a monomial basis by a constant $\Am$ that leaves the conserved
moments unchanged, with each element inheriting the rate of the monomial it
replaces, on a lattice whose deviatoric subspace is orthogonal to the conserved
moments. Within that class it is independent of how the orthogonalisation is
carried out, and it turns the question into one inner product per monomial:

\begin{tcolorbox}[colback=black!3,colframe=black!55,boxrule=0.5pt,arc=1pt,
  left=4pt,right=4pt,top=3pt,bottom=3pt,title={Practical criterion},
  fonttitle=\bfseries\small,coltitle=black,colbacktitle=black!10]
Given a lattice, a monomial central-moment basis and a relaxation matrix:
\begin{enumerate}
\item identify the deviatoric subspace $\mathcal{D}$, which relaxes at the shear
  rate $\omega$, and check that it is orthogonal to the conserved moments;
\item for every higher-order monomial $m$ relaxing at a rate other than $\omega$,
  compute $\|P_{\mathcal{D}}m\|^2$;
\item if any of these is nonzero, every orthogonalisation of the basis changes the
  scheme; if all vanish, the shear rate is free, and any remaining conditions are
  read directly from the non-conserved columns of $[\Lam,\Am]$.
\end{enumerate}
\end{tcolorbox}

\subsection{Stability at rest}
\label{sec:reststab-theory}

Lemma~\ref{thm:gauge} decides when two bases give the same scheme; it does not
say which of two different schemes to prefer. One property singles out a class of
bases. Let $w_i=f^{\mathrm{eq}}_i(1,\mathbf{0})$ be the populations of the rest
equilibrium, $\mathsf W=\mathrm{diag}(w_i)$, and $\|g\|_w^2=\sum_i|g_i|^2/w_i$ the
associated norm.

\begin{proposition}[Stability at rest]
\label{prop:rest}
Let the rows of $\Tm(\mathbf 0)$, including the conserved ones, be mutually orthogonal in the
inner product $\sum_i w_i a_i b_i$, and let every non-conserved rate lie in $(0,2]$.
Then the amplification matrix of the scheme linearised about the rest state satisfies
$\|\mathsf G(\bm k)\|_w\le1$ for every wavevector $\bm k$. The same holds for every
basis related to this one by a constant $\Am$ with $[\Lam,\Am]_{ij}=0$ for all
non-conserved $j$.
\end{proposition}

\begin{proof}
At rest the velocity dependence of $\Tm$ contributes only through terms multiplying
$f^{\mathrm{eq}}-f$, which vanish there, so the linearised collision is
$\mathsf J=\Id-\Pm(\Id-\Pi)$, with $\Pm=\Tm(\mathbf 0)^{-1}\Lam\Tm(\mathbf 0)$ and
$\Pi$ the derivative of the equilibrium, $\Pi=\mathsf W(\bm1\bm1^{\top}+c_s^{-2}
\bm{e}\bm{e}^{\top})$, the $w$-orthogonal projector onto the conserved moments. With
$\Tm\mathsf W\Tm^{\top}=\mathsf D$ diagonal, $\mathsf W^{-1}\Pm=\Tm^{\top}\mathsf
D^{-1}\Lam\Tm$ is symmetric, so $\Pm$ is self-adjoint in the $w$-inner product, and
the conserved moments, mapped by $\mathsf W$, are among its eigenvectors. Hence
$\mathsf J$ is self-adjoint, with eigenvalue $1$ on the range of $\Pi$ and
$1-\lambda_j\in[-1,1)$ on its complement, and $\|\mathsf J\|_w\le1$. The streaming
factor $\mathsf E(\bm k)$ is diagonal with unimodular entries, so
$\|\mathsf E(\bm k)\|_w=1$. For the second statement, $f^{\mathrm{eq}}-f$ has
vanishing conserved moments in the linearisation too, and the argument of
Corollary~\ref{cor:physical} gives the same $\mathsf J$.
\end{proof}

This is the stability structure of Refs.~\cite{banda2006,junk2009} specialised to
moment bases. On D2Q9 and D3Q27 with $c_s^2=1/3$ the rest equilibrium has the lattice
weights, so the orthogonal bases of this paper satisfy the hypothesis, whereas the
non-orthogonal ones do so only where they are gauge equivalent to them. Nothing
guarantees the rest state otherwise. With isotropic rates drawn at random from
$(0.05,2)$ for the bulk, shear, third- and fourth-order moments of D2Q9, the
non-orthogonal central-moment scheme is linearly unstable at rest in $69$ of $200$
draws, the orthogonal one in none; in the setting of Section~\ref{sec:stability},
with $\lambda_3=\lambda_4=1$, the non-orthogonal rest state loses stability once
$\omega_b$ exceeds a threshold between $1.75$ and $1.98$ that grows with $\omega$
(Fig.~\ref{fig:stability}). On a rectangular lattice with a free sound speed the rest
equilibrium no longer has the lattice weights, and which inner product a basis is
orthogonal in then matters (Section~\ref{sec:ipsummary} and \ref{sec:ip}).

\section{When does orthogonalisation matter?}
\label{sec:when}

The question is not which moments Gram--Schmidt mixes, but whether the resulting
coupling graph joins moments that carry different relaxation rates.
Figure~\ref{fig:graphs} shows the graphs for the square, rectangular and cubic lattices; the rest of
this section reads the consequences off them.

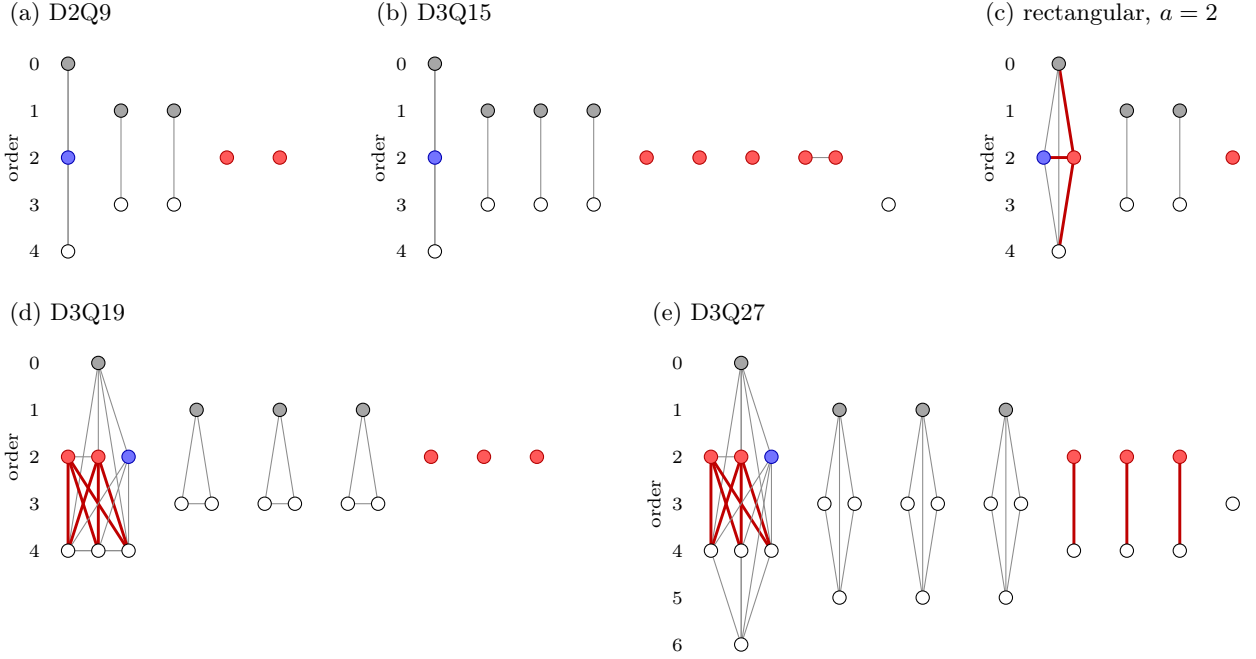
\begin{figure}[t]
\centering
\begin{tikzpicture}[baseline=(current bounding box.north),cons/.style={circle,draw,fill=black!35,inner sep=0pt,minimum size=5pt},dev/.style={circle,draw=red!70!black,fill=red!65,inner sep=0pt,minimum size=5pt},tr/.style={circle,draw=blue!70!black,fill=blue!55,inner sep=0pt,minimum size=5pt},hi/.style={circle,draw,fill=white,inner sep=0pt,minimum size=5pt}]
\draw[black!45] (0.000,0.000)--(0.000,-1.240);
\draw[black!45] (0.000,0.000)--(0.000,-2.480);
\draw[black!45] (0.700,-0.620)--(0.700,-1.860);
\draw[black!45] (1.400,-0.620)--(1.400,-1.860);
\draw[black!45] (0.000,-1.240)--(0.000,-2.480);
\node[cons] at (0.000,0.000) {};
\node[tr] at (0.000,-1.240) {};
\node[hi] at (0.000,-2.480) {};
\node[cons] at (0.700,-0.620) {};
\node[hi] at (0.700,-1.860) {};
\node[cons] at (1.400,-0.620) {};
\node[hi] at (1.400,-1.860) {};
\node[dev] at (2.100,-1.240) {};
\node[dev] at (2.800,-1.240) {};
\node[font=\scriptsize,anchor=east] at (-0.25,0.000) {0};
\node[font=\scriptsize,anchor=east] at (-0.25,-0.620) {1};
\node[font=\scriptsize,anchor=east] at (-0.25,-1.240) {2};
\node[font=\scriptsize,anchor=east] at (-0.25,-1.860) {3};
\node[font=\scriptsize,anchor=east] at (-0.25,-2.480) {4};
\node[font=\scriptsize,rotate=90] at (-0.75,-1.240) {order};
\node[font=\small,anchor=west] at (-0.9,0.650) {(a) D2Q9};
\end{tikzpicture}\hfill
\begin{tikzpicture}[baseline=(current bounding box.north),cons/.style={circle,draw,fill=black!35,inner sep=0pt,minimum size=5pt},dev/.style={circle,draw=red!70!black,fill=red!65,inner sep=0pt,minimum size=5pt},tr/.style={circle,draw=blue!70!black,fill=blue!55,inner sep=0pt,minimum size=5pt},hi/.style={circle,draw,fill=white,inner sep=0pt,minimum size=5pt}]
\draw[black!45] (0.000,0.000)--(0.000,-1.240);
\draw[black!45] (0.000,0.000)--(0.000,-2.480);
\draw[black!45] (0.700,-0.620)--(0.700,-1.860);
\draw[black!45] (1.400,-0.620)--(1.400,-1.860);
\draw[black!45] (2.100,-0.620)--(2.100,-1.860);
\draw[black!45] (4.900,-1.240)--(5.300,-1.240);
\draw[black!45] (0.000,-1.240)--(0.000,-2.480);
\node[cons] at (0.000,0.000) {};
\node[tr] at (0.000,-1.240) {};
\node[hi] at (0.000,-2.480) {};
\node[cons] at (0.700,-0.620) {};
\node[hi] at (0.700,-1.860) {};
\node[cons] at (1.400,-0.620) {};
\node[hi] at (1.400,-1.860) {};
\node[cons] at (2.100,-0.620) {};
\node[hi] at (2.100,-1.860) {};
\node[dev] at (2.800,-1.240) {};
\node[dev] at (3.500,-1.240) {};
\node[dev] at (4.200,-1.240) {};
\node[dev] at (4.900,-1.240) {};
\node[dev] at (5.300,-1.240) {};
\node[hi] at (6.000,-1.860) {};
\node[font=\scriptsize,anchor=east] at (-0.25,0.000) {0};
\node[font=\scriptsize,anchor=east] at (-0.25,-0.620) {1};
\node[font=\scriptsize,anchor=east] at (-0.25,-1.240) {2};
\node[font=\scriptsize,anchor=east] at (-0.25,-1.860) {3};
\node[font=\scriptsize,anchor=east] at (-0.25,-2.480) {4};
\node[font=\scriptsize,rotate=90] at (-0.75,-1.240) {order};
\node[font=\small,anchor=west] at (-0.9,0.650) {(b) D3Q15};
\end{tikzpicture}\hfill
\begin{tikzpicture}[baseline=(current bounding box.north),cons/.style={circle,draw,fill=black!35,inner sep=0pt,minimum size=5pt},dev/.style={circle,draw=red!70!black,fill=red!65,inner sep=0pt,minimum size=5pt},tr/.style={circle,draw=blue!70!black,fill=blue!55,inner sep=0pt,minimum size=5pt},hi/.style={circle,draw,fill=white,inner sep=0pt,minimum size=5pt}]
\draw[black!45] (0.200,0.000)--(0.000,-1.240);
\draw[red!75!black,line width=1.1pt] (0.200,0.000)--(0.400,-1.240);
\draw[black!45] (0.200,0.000)--(0.200,-2.480);
\draw[black!45] (1.100,-0.620)--(1.100,-1.860);
\draw[black!45] (1.800,-0.620)--(1.800,-1.860);
\draw[red!75!black,line width=1.1pt] (0.000,-1.240)--(0.400,-1.240);
\draw[black!45] (0.000,-1.240)--(0.200,-2.480);
\draw[red!75!black,line width=1.1pt] (0.400,-1.240)--(0.200,-2.480);
\node[cons] at (0.200,0.000) {};
\node[tr] at (0.000,-1.240) {};
\node[dev] at (0.400,-1.240) {};
\node[hi] at (0.200,-2.480) {};
\node[cons] at (1.100,-0.620) {};
\node[hi] at (1.100,-1.860) {};
\node[cons] at (1.800,-0.620) {};
\node[hi] at (1.800,-1.860) {};
\node[dev] at (2.500,-1.240) {};
\node[font=\scriptsize,anchor=east] at (-0.25,0.000) {0};
\node[font=\scriptsize,anchor=east] at (-0.25,-0.620) {1};
\node[font=\scriptsize,anchor=east] at (-0.25,-1.240) {2};
\node[font=\scriptsize,anchor=east] at (-0.25,-1.860) {3};
\node[font=\scriptsize,anchor=east] at (-0.25,-2.480) {4};
\node[font=\scriptsize,rotate=90] at (-0.75,-1.240) {order};
\node[font=\small,anchor=west] at (-0.9,0.650) {(c) rectangular, $a=2$};
\end{tikzpicture}

\vspace{0.9em}
\begin{tikzpicture}[baseline=(current bounding box.north),cons/.style={circle,draw,fill=black!35,inner sep=0pt,minimum size=5pt},dev/.style={circle,draw=red!70!black,fill=red!65,inner sep=0pt,minimum size=5pt},tr/.style={circle,draw=blue!70!black,fill=blue!55,inner sep=0pt,minimum size=5pt},hi/.style={circle,draw,fill=white,inner sep=0pt,minimum size=5pt}]
\draw[black!45] (0.400,0.000)--(0.800,-1.240);
\draw[black!45] (0.400,0.000)--(0.000,-2.480);
\draw[black!45] (0.400,0.000)--(0.400,-2.480);
\draw[black!45] (0.400,0.000)--(0.800,-2.480);
\draw[black!45] (1.700,-0.620)--(1.500,-1.860);
\draw[black!45] (1.700,-0.620)--(1.900,-1.860);
\draw[black!45] (2.800,-0.620)--(2.600,-1.860);
\draw[black!45] (2.800,-0.620)--(3.000,-1.860);
\draw[black!45] (3.900,-0.620)--(3.700,-1.860);
\draw[black!45] (3.900,-0.620)--(4.100,-1.860);
\draw[black!45] (0.000,-1.240)--(0.400,-1.240);
\draw[red!75!black,line width=1.1pt] (0.000,-1.240)--(0.000,-2.480);
\draw[red!75!black,line width=1.1pt] (0.000,-1.240)--(0.400,-2.480);
\draw[red!75!black,line width=1.1pt] (0.000,-1.240)--(0.800,-2.480);
\draw[red!75!black,line width=1.1pt] (0.400,-1.240)--(0.000,-2.480);
\draw[red!75!black,line width=1.1pt] (0.400,-1.240)--(0.400,-2.480);
\draw[red!75!black,line width=1.1pt] (0.400,-1.240)--(0.800,-2.480);
\draw[black!45] (0.800,-1.240)--(0.000,-2.480);
\draw[black!45] (0.800,-1.240)--(0.400,-2.480);
\draw[black!45] (0.800,-1.240)--(0.800,-2.480);
\draw[black!45] (1.500,-1.860)--(1.900,-1.860);
\draw[black!45] (2.600,-1.860)--(3.000,-1.860);
\draw[black!45] (3.700,-1.860)--(4.100,-1.860);
\draw[black!45] (0.000,-2.480)--(0.400,-2.480);
\draw[black!45] (0.000,-2.480)--(0.800,-2.480);
\draw[black!45] (0.400,-2.480)--(0.800,-2.480);
\node[cons] at (0.400,0.000) {};
\node[dev] at (0.000,-1.240) {};
\node[dev] at (0.400,-1.240) {};
\node[tr] at (0.800,-1.240) {};
\node[hi] at (0.000,-2.480) {};
\node[hi] at (0.400,-2.480) {};
\node[hi] at (0.800,-2.480) {};
\node[cons] at (1.700,-0.620) {};
\node[hi] at (1.500,-1.860) {};
\node[hi] at (1.900,-1.860) {};
\node[cons] at (2.800,-0.620) {};
\node[hi] at (2.600,-1.860) {};
\node[hi] at (3.000,-1.860) {};
\node[cons] at (3.900,-0.620) {};
\node[hi] at (3.700,-1.860) {};
\node[hi] at (4.100,-1.860) {};
\node[dev] at (4.800,-1.240) {};
\node[dev] at (5.500,-1.240) {};
\node[dev] at (6.200,-1.240) {};
\node[font=\scriptsize,anchor=east] at (-0.25,0.000) {0};
\node[font=\scriptsize,anchor=east] at (-0.25,-0.620) {1};
\node[font=\scriptsize,anchor=east] at (-0.25,-1.240) {2};
\node[font=\scriptsize,anchor=east] at (-0.25,-1.860) {3};
\node[font=\scriptsize,anchor=east] at (-0.25,-2.480) {4};
\node[font=\scriptsize,rotate=90] at (-0.75,-1.240) {order};
\node[font=\small,anchor=west] at (-0.9,0.650) {(d) D3Q19};
\end{tikzpicture}\hfill
\begin{tikzpicture}[baseline=(current bounding box.north),cons/.style={circle,draw,fill=black!35,inner sep=0pt,minimum size=5pt},dev/.style={circle,draw=red!70!black,fill=red!65,inner sep=0pt,minimum size=5pt},tr/.style={circle,draw=blue!70!black,fill=blue!55,inner sep=0pt,minimum size=5pt},hi/.style={circle,draw,fill=white,inner sep=0pt,minimum size=5pt}]
\draw[black!45] (0.400,0.000)--(0.800,-1.240);
\draw[black!45] (0.400,0.000)--(0.000,-2.480);
\draw[black!45] (0.400,0.000)--(0.400,-2.480);
\draw[black!45] (0.400,0.000)--(0.800,-2.480);
\draw[black!45] (0.400,0.000)--(0.400,-3.720);
\draw[black!45] (1.700,-0.620)--(1.500,-1.860);
\draw[black!45] (1.700,-0.620)--(1.900,-1.860);
\draw[black!45] (1.700,-0.620)--(1.700,-3.100);
\draw[black!45] (2.800,-0.620)--(2.600,-1.860);
\draw[black!45] (2.800,-0.620)--(3.000,-1.860);
\draw[black!45] (2.800,-0.620)--(2.800,-3.100);
\draw[black!45] (3.900,-0.620)--(3.700,-1.860);
\draw[black!45] (3.900,-0.620)--(4.100,-1.860);
\draw[black!45] (3.900,-0.620)--(3.900,-3.100);
\draw[red!75!black,line width=1.1pt] (4.800,-1.240)--(4.800,-2.480);
\draw[red!75!black,line width=1.1pt] (5.500,-1.240)--(5.500,-2.480);
\draw[red!75!black,line width=1.1pt] (6.200,-1.240)--(6.200,-2.480);
\draw[black!45] (0.000,-1.240)--(0.400,-1.240);
\draw[red!75!black,line width=1.1pt] (0.000,-1.240)--(0.000,-2.480);
\draw[red!75!black,line width=1.1pt] (0.000,-1.240)--(0.400,-2.480);
\draw[red!75!black,line width=1.1pt] (0.000,-1.240)--(0.800,-2.480);
\draw[red!75!black,line width=1.1pt] (0.400,-1.240)--(0.000,-2.480);
\draw[red!75!black,line width=1.1pt] (0.400,-1.240)--(0.400,-2.480);
\draw[red!75!black,line width=1.1pt] (0.400,-1.240)--(0.800,-2.480);
\draw[black!45] (0.800,-1.240)--(0.000,-2.480);
\draw[black!45] (0.800,-1.240)--(0.400,-2.480);
\draw[black!45] (0.800,-1.240)--(0.800,-2.480);
\draw[black!45] (0.800,-1.240)--(0.400,-3.720);
\draw[black!45] (1.500,-1.860)--(1.700,-3.100);
\draw[black!45] (1.900,-1.860)--(1.700,-3.100);
\draw[black!45] (2.600,-1.860)--(2.800,-3.100);
\draw[black!45] (3.000,-1.860)--(2.800,-3.100);
\draw[black!45] (3.700,-1.860)--(3.900,-3.100);
\draw[black!45] (4.100,-1.860)--(3.900,-3.100);
\draw[black!45] (0.000,-2.480)--(0.400,-3.720);
\draw[black!45] (0.400,-2.480)--(0.400,-3.720);
\draw[black!45] (0.800,-2.480)--(0.400,-3.720);
\node[cons] at (0.400,0.000) {};
\node[dev] at (0.000,-1.240) {};
\node[dev] at (0.400,-1.240) {};
\node[tr] at (0.800,-1.240) {};
\node[hi] at (0.000,-2.480) {};
\node[hi] at (0.400,-2.480) {};
\node[hi] at (0.800,-2.480) {};
\node[hi] at (0.400,-3.720) {};
\node[cons] at (1.700,-0.620) {};
\node[hi] at (1.500,-1.860) {};
\node[hi] at (1.900,-1.860) {};
\node[hi] at (1.700,-3.100) {};
\node[cons] at (2.800,-0.620) {};
\node[hi] at (2.600,-1.860) {};
\node[hi] at (3.000,-1.860) {};
\node[hi] at (2.800,-3.100) {};
\node[cons] at (3.900,-0.620) {};
\node[hi] at (3.700,-1.860) {};
\node[hi] at (4.100,-1.860) {};
\node[hi] at (3.900,-3.100) {};
\node[dev] at (4.800,-1.240) {};
\node[hi] at (4.800,-2.480) {};
\node[dev] at (5.500,-1.240) {};
\node[hi] at (5.500,-2.480) {};
\node[dev] at (6.200,-1.240) {};
\node[hi] at (6.200,-2.480) {};
\node[hi] at (6.900,-1.860) {};
\node[font=\scriptsize,anchor=east] at (-0.25,0.000) {0};
\node[font=\scriptsize,anchor=east] at (-0.25,-0.620) {1};
\node[font=\scriptsize,anchor=east] at (-0.25,-1.240) {2};
\node[font=\scriptsize,anchor=east] at (-0.25,-1.860) {3};
\node[font=\scriptsize,anchor=east] at (-0.25,-2.480) {4};
\node[font=\scriptsize,anchor=east] at (-0.25,-3.100) {5};
\node[font=\scriptsize,anchor=east] at (-0.25,-3.720) {6};
\node[font=\scriptsize,rotate=90] at (-0.75,-1.860) {order};
\node[font=\small,anchor=west] at (-0.9,0.650) {(e) D3Q27};
\end{tikzpicture}
\caption{Coupling graphs $\mathcal{G}(\Am)$ of the order-graded orthogonalisation, drawn from the computed transformation matrices with the standard weights: one vertex per basis moment, placed by polynomial order, and one edge per nonzero off-diagonal entry of $\Am$. Vertices are conserved (grey), deviatoric (red), trace (blue) or higher-order (white); edges joining the deviatoric block to any other moment are drawn thick and red. Components are separated horizontally; the isolated vertex at order three on D3Q15 and D3Q27 is $\Cx\Cy\Cz$. (a,\,b) On D2Q9 and D3Q15 no such edge exists, and the shear rate is free. (c) On the rectangular D2Q9 lattice with $a=2$ the broken symmetry between the axes joins the normal-stress difference to the density, the trace and the fourth-order moment, even with a single shear rate (Section~\ref{sec:rect-theory}). (d,\,e) On D3Q19 the two normal-stress differences are joined to the three fourth-order moments $C_\alpha^2C_\beta^2$ through the face diagonals, and on D3Q27 the off-diagonal shear moments are in addition joined to $C_\alpha C_\beta C_\gamma^2$ through the body diagonals, the two routes of Theorem~\ref{thm:geom}.}
\label{fig:graphs}
\end{figure}

\subsection{A lattice-geometric criterion}
\label{sec:geom}

Proposition~\ref{prop:projection} reduces the question to one inner product per
monomial. On the standard lattices these inner products can be evaluated once and for
all, and the answer depends only on which velocity shells the lattice contains.

\begin{theorem}[Lattice geometry]
\label{thm:geom}
Let the velocities lie in $\{0,\pm1\}^d$, $d=2$ or $3$, with positive weights that are
invariant under permutations and reflections of the axes. Let the basis consist of the
conserved moments, the trace, the deviatoric second-order monomials, relaxed at
$\omega$, and further monomials $\Cx^{p}\Cy^{q}\Cz^{r}$ with exponents at most two,
linearly independent on the lattice, and let it be orthogonalised in the inner product
of the weights by Gram--Schmidt in any order that begins with the conserved moments.
Then a basis monomial outside the deviatoric block has a nonzero projection onto it if
and only if it is
\begin{enumerate}
\item[(a)] $C_\alpha^2C_\beta^2$ with $\alpha\neq\beta$, in three dimensions, on a lattice
  with velocities that have exactly two nonzero components; or
\item[(b)] $C_\alpha C_\beta C_\gamma^2$ with $\alpha,\beta,\gamma$ distinct, on a lattice
  with velocities that have three nonzero components.
\end{enumerate}
Consequently the shear rate is free on D2Q9 and D3Q15, whereas on D3Q19 and D3Q27
the orthogonal and non-orthogonal schemes differ on physical states without forcing
whenever a monomial of type (a) or (b) relaxes at a rate other than $\omega$.
\end{theorem}

\begin{proof}
At rest the central moments are raw moments, and with components in $\{0,\pm1\}$ the
square $C_\alpha^2$ is the indicator of a nonzero $\alpha$-component. Write $W(S)$ for
the total weight of the velocities whose components on the set of axes $S$ are all
nonzero; by the symmetry of the weights $W(S)$ depends only on $|S|$. The deviatoric
block is spanned by $C_\alpha C_\beta$ and $C_\alpha^2-C_\gamma^2$, and it is orthogonal
to the conserved moments by parity and symmetry. By reflection symmetry
$\avg{m,C_\alpha C_\beta}$ vanishes unless the exponents of $\alpha$ and $\beta$ in $m$
are odd and all others even; with exponents at most two this leaves
$m=C_\alpha C_\beta C_\gamma^2$, for which $\avg{m,C_\alpha C_\beta}=W(\{\alpha,\beta,\gamma\})$,
the weight of the velocities with three nonzero components. Likewise
$\avg{m,C_\alpha^2-C_\gamma^2}$ vanishes unless all exponents are even, $m=\prod_{\delta
\in S}C_\delta^2$, and then equals $W(S\cup\{\alpha\})-W(S\cup\{\gamma\})$, which is
nonzero only if exactly one of $\alpha,\gamma$ lies in $S$. The trace is symmetric and
contributes nothing; for $S=\{\alpha,\beta\}$ with $\beta\ne\gamma$ the difference is
the weight of the velocities whose $\alpha$- and $\beta$-components are nonzero and whose
$\gamma$-component is zero, which exists only in three dimensions and is positive only
if the lattice has velocities with exactly two nonzero components; for $|S|=0$ or
$3$ the difference vanishes, and $|S|=1$ gives the second-order monomials themselves,
which are not among the further monomials. Every other overlap vanishes, and Proposition~\ref{prop:projection}
turns the two cases into the stated consequences. On D2Q9 case (a) cannot arise, and on
D3Q15, which has no velocities with exactly two nonzero components, $C_\alpha C_\beta
C_\gamma^2$ coincides with $C_\alpha C_\beta$ and is not an independent basis element.
\end{proof}

The theorem locates the difference between two and three dimensions in the geometry
of the velocity set. Orthogonalisation reaches the shear moments only through the face
diagonals, which weight $\Cx^2\Cy^2$ differently from $\Cx^2\Cz^2$, and through the
body diagonals, which link $\Cx\Cy\Cz^2$ to $\Cx\Cy$; D2Q9 has no third axis, D3Q15 has
body diagonals but no independent monomial of type (b), and D3Q19 and D3Q27 have face
diagonals. The weights enter only through their symmetry, so the conclusion holds for
every such inner product, including the rest-equilibrium and unweighted ones used in
\ref{sec:ip}. On rectangular lattices the symmetry under exchange of axes is
broken, and even the trace acquires a projection (Section~\ref{sec:rect-theory}).

For D3Q15 we checked the prediction directly (Fig.~\ref{fig:graphs}b), with the
fifteen-moment basis
\begin{equation*}
  \Bigl\{1,\ C_\alpha,\ \textstyle\sum_\alpha C_\alpha^2,\ \text{deviatoric},\
  C_\alpha\sum_{\beta\ne\alpha}C_\beta^2,\ \Cx\Cy\Cz,\ \sum_{\alpha<\beta}C_\alpha^2C_\beta^2\Bigr\}:
\end{equation*}
its rest-frame Gram--Schmidt matrix has no entry touching the deviatoric block, for the standard and
for random symmetric weights, the surviving commutator entries link only the
density, the trace and the fourth-order moment, and at a random state the
post-collision populations of the two formulations agree exactly for a single shear
rate, differ with an independent bulk rate, and agree again when the fourth-order
moment relaxes at the bulk rate. D3Q15 therefore behaves as D2Q9.

\subsection{The square D2Q9 lattice}

Take the non-orthogonal central-moment basis of Ref.~\cite{derosis2016},
\begin{equation}
  \begin{aligned}
    \{\,&1,\ \Cx,\ \Cy,\ \Cx^2+\Cy^2,\ \Cx^2-\Cy^2,\\
        &\Cx\Cy,\ \Cx^2\Cy,\ \Cx\Cy^2,\ \Cx^2\Cy^2\,\},
  \end{aligned}
  \label{eq:basis-d2q9}
\end{equation}
together with the orthogonal basis obtained from it by rest-frame Gram--Schmidt,
the Hermite-type set $\{1,\Cx,\Cy,\Cx^2+\Cy^2-2c_s^2,\Cx^2-\Cy^2,\Cx\Cy,
(\Cx^2-c_s^2)\Cy,(\Cy^2-c_s^2)\Cx,(\Cx^2-c_s^2)(\Cy^2-c_s^2)\}$. The
transformation $\Am$ between them is constant, lower triangular with unit
diagonal, and has five nonzero off-diagonal entries:
\begin{equation}
  \begin{aligned}
    &\Am_{41}=-\tfrac{2}{3},\quad
     \Am_{73}=-\tfrac{1}{3},\quad
     \Am_{82}=-\tfrac{1}{3},\\
    &\Am_{91}=\tfrac{1}{9},\quad
     \Am_{94}=-\tfrac{1}{3}.
  \end{aligned}
  \label{eq:A-d2q9}
\end{equation}
Its coupling graph, Fig.~\ref{fig:graphs}(a), has components
\begin{equation}
  \begin{aligned}
    &\{1,\ \Cx^2{+}\Cy^2,\ \Cx^2\Cy^2\},\quad \{\Cx,\ \Cx\Cy^2\},\\
    &\{\Cy,\ \Cx^2\Cy\},\quad \{\Cx^2{-}\Cy^2\},\quad \{\Cx\Cy\}.
  \end{aligned}
  \label{eq:components-d2q9}
\end{equation}
As Theorem~\ref{thm:geom} requires, the two deviatoric moments are isolated vertices, and they are precisely the
moments carrying the shear rate $\omega$. By Corollary~\ref{cor:graph} that rate
is unconstrained, which proves the following.

\begin{corollary}
\label{cor:d2q9}
On the square D2Q9 lattice with
\begin{equation*}
  \Lam=\mathrm{diag}(1,1,1,1,\omega,\omega,1,1,1),
\end{equation*}
the orthogonal and non-orthogonal central-moment schemes are identical, for arbitrary $\omega$,
arbitrary $\bfu$, arbitrary populations and arbitrary forcing.
\end{corollary}

Proposition~\ref{prop:projection} gives the same result from the other side:
$\|P_{\mathcal{D}}\,\Cx^2\Cy^2\|^2=0$ on D2Q9, the fourth-order moment being
exactly orthogonal to the deviatoric block because $\Cx^2\Cy^2$ is symmetric
under $x\leftrightarrow y$ while $\Cx^2-\Cy^2$ is antisymmetric, and even in each
component while $\Cx\Cy$ is odd.

\paragraph{Tuning the bulk viscosity}
Give the trace moment its own rate,
\begin{equation*}
  \Lam=\mathrm{diag}(1,1,1,\omega_b,\omega,\omega,1,1,1),
\end{equation*}
as is standard when the bulk viscosity is tuned independently~\cite{dellar2001}. Then \eqref{eq:A-d2q9} and
Lemma~\ref{lem:factor} give
\begin{equation}
  [\Lam,\Am]_{41}=-\tfrac{2}{3}\left(\omega_b-1\right),
  \qquad
  [\Lam,\Am]_{94}=\tfrac{1}{3}\left(\omega_b-1\right),
  \label{eq:comm-d2q9-bulk}
\end{equation}
and the schemes part company. The population-level difference
$f^{\star}_{\mathrm{orth}}-f^{\star}_{\mathrm{non}}$ is exactly proportional to
$(\omega_b-1)$ and is already nonzero at rest, $\bfu=0$ and $\mathcal{F}=0$,
where it reduces to a fixed linear combination of the populations; it is not an
$O(\mathrm{Ma})$ effect. Of the two entries, $(4,1)$ lies in the density column
and is inactive on physical states, whereas $(9,4)$ connects the trace to the
fourth-order moment and is not; its effect is made explicit in
Eq.~\eqref{eq:q9}. Deleting the conserved vertices
from \eqref{eq:components-d2q9} leaves $\{\Cx^2{+}\Cy^2,\Cx^2\Cy^2\}$ as the only
component with more than one vertex, so the operative condition on physical states
is $\omega_b=\lambda_4$, the rate of $\Cx^2\Cy^2$. The equivalence is broken by an
independent bulk viscosity only because the fourth-order moment is conventionally
held at unity: relaxing it at $\omega_b$ as well restores the equivalence while
leaving the bulk viscosity free, as Section~\ref{sec:d2stab} confirms numerically.

\paragraph{The single term}
In the orthogonal coordinates, in which the orthogonal scheme is diagonal, the
non-orthogonal scheme differs from it on physical states by a single term, in the
relaxation of the Hermite fourth-order moment $q_9=(\Cx^2-c_s^2)(\Cy^2-c_s^2)$:
\begin{equation}
  q_9^{\star}-q_9=\lambda_4\,\delta q_9+\frac{\lambda_4-\omega_b}{3}\,\delta q_4 ,
  \label{eq:q9}
\end{equation}
where $\delta q=\Tm_{\mathrm{orth}}(f^{\mathrm{eq}}-f)$ and $q_4$ is the trace. The
coefficient of the second term is $-[\Lam,\Am]_{94}$, and every other component of the
post-collision state is the same in the two schemes; we checked
Eq.~\eqref{eq:q9} on random physical states to round-off. The raw moment
$\Cx^2\Cy^2=q_9+\tfrac13(\Cx^2+\Cy^2)-\tfrac19$ contains one third of the trace, so
relaxing it at $\lambda_4$ relaxes that part of the trace at the wrong rate: with
$\lambda_4=1$ the non-orthogonal scheme feeds $(1-\omega_b)/3$ of the trace
non-equilibrium, up to a third of a barely damped mode as $\omega_b\to2$, into the
fourth order at every step. Relaxing at one common rate every non-conserved moment
that orthogonalisation mixes, here $\lambda_4=\omega_b$, removes the term and makes
the two formulations one scheme on physical states; we refer to this choice as the
recipe.

\subsection{Rectangular lattices}
\label{sec:rect-theory}

Lattice Boltzmann schemes on rectangular grids have been built with
multiple-relaxation-time collisions~\cite{bouzidi2001,zhou2012,zong2016,peng2016}, with
additional velocities and a BGK collision~\cite{hegele2013}, and with central
moments~\cite{yahia2021,yahia2021cuboid}; on the nine-velocity stencil, isotropic
transport coefficients take additional terms, which Refs.~\cite{zong2016}
and~\cite{peng2016} supply through an extra degree of freedom and through stress
components in the equilibrium moments, respectively. Stretch the lattice in $y$ by the grid aspect ratio $a=\Delta y/\Delta x$ and keep
the standard D2Q9 weights, so that $c_{sx}^2=1/3$ and $c_{sy}^2=a^2/3$. The
rest-frame orthogonalisation acquires three couplings absent at $a=1$:
\begin{align}
  \Am_{51}&=\frac{2a^{2}\left(1-a^{2}\right)}{3\left(a^{4}+1\right)},
  &
  \Am_{54}&=\frac{a^{4}-1}{a^{4}+1},
  &
  \Am_{95}&=\frac{1-a^{2}}{6},
  \label{eq:A-rect}
\end{align}
each vanishing identically at $a=1$. The mechanism is transparent: on an
anisotropic lattice $\avg{\Cx^2-\Cy^2,1}=c_{sx}^2-c_{sy}^2\neq0$, so the
normal-stress-difference moment cannot be orthogonalised without drawing in the
density and the trace. Rectangular formulations that retune the weights change
these coefficients but not the mechanism, which requires only $c_{sx}\neq c_{sy}$.

For any inner product built from product weights, of which the lattice weights,
the rest-equilibrium weights and the unweighted inner product are examples, the
offending coefficient is
\begin{equation}
  \Am_{54}=-\frac{\mathrm{Var}(\Cx^2)-\mathrm{Var}(\Cy^2)}{\mathrm{Var}(\Cx^2)+\mathrm{Var}(\Cy^2)},
  \label{eq:A54-var}
\end{equation}
with the variances taken under the weights; it reduces to~\eqref{eq:A-rect} for the
lattice weights and vanishes only if the weights are tuned to equalise the two
variances, which the anisotropy of the lattice otherwise prevents. The fourth-order
coupling is simpler still: for product weights
$\Am_{95}=\tfrac12[\mathrm{E}(\Cx^2)-\mathrm{E}(\Cy^2)]$, with the means taken under the
weights, which reduces to~\eqref{eq:A-rect} for the lattice weights and vanishes
whenever the weights have isotropic second moments, as those of the rest equilibrium
do.

The equivalence consequently fails on a rectangular lattice \emph{even for the
single-$\omega$ relaxation matrix of Corollary~\ref{cor:d2q9}}:
\begin{align}
  [\Lam,\Am]_{51}&=-\frac{2a^{2}\left(a^{2}-1\right)\left(\omega-1\right)}
                         {3\left(a^{4}+1\right)},\nonumber\\
  [\Lam,\Am]_{54}&=\frac{\left(a^{4}-1\right)\left(\omega-1\right)}{a^{4}+1},\\
  [\Lam,\Am]_{95}&=\frac{\left(a^{2}-1\right)\left(\omega-1\right)}{6},\nonumber
\end{align}
the second and third of which lie in non-conserved columns and act on physical
states. For the full rectangular relaxation matrix
\begin{equation*}
  \Lam=\mathrm{diag}(1,1,1,\omega_3,\omega_4,\omega_5,\omega_6,\omega_6,\omega_7),
\end{equation*}
eight entries survive, among them
\begin{equation}
  [\Lam,\Am]_{54}=-\frac{a^{4}-1}{a^{4}+1}\left(\omega_3-\omega_4\right),
  \label{eq:rect-key}
\end{equation}
a purely geometric factor multiplying the bulk--shear rate contrast. This term is the coupling described as `spurious' in the rectangular- and cuboid-lattice literature~\cite{yahia2021,yahia2021cuboid}, now
in closed form, and it is not perturbative: the prefactor in \eqref{eq:rect-key}
is $0.670$ at $a=1.5$, $0.882$ at $a=2$ and $0.976$ at $a=3$.

\paragraph{An equilibrium-independent statement}
By Lemma~\ref{thm:gauge} the two schemes differ as maps precisely when
$\Pm_{\mathrm{orth}}\neq\Pm_{\mathrm{non}}$, and the size of the difference at a
given state is governed by $(\Pm_{\mathrm{orth}}-\Pm_{\mathrm{non}})
(f^{\mathrm{eq}}-f)$. There is thus no need to commit to a particular rectangular
equilibrium: the rank of $\Delta\Pm=\Pm_{\mathrm{orth}}-\Pm_{\mathrm{non}}$
bounds the set of states on which the two can agree. Evaluated at
$\bfu=(0.13,-0.07)$,
\begin{center}
\begin{tabular}{lcc}
\toprule
 & $\mathrm{rank}\,\Delta\Pm$ & $\max\left|\Delta\Pm\right|$\\
\midrule
square, $\omega_b=1$ (gauge)      & $0$ & $0$      \\
square, $\omega_b=1.6$            & $2$ & $0.747$  \\
rectangular $a=2$, single $\omega$& $2$ & $0.707$  \\
rectangular $a=2$, full MRT       & $4$ & $0.627$  \\
\bottomrule
\end{tabular}
\end{center}
The two schemes agree only on a subspace of codimension two, or four in the full
MRT case, for every choice of equilibrium and forcing, and the discrepancy is
$O(1)$ rather than perturbative.

One rate is conspicuously absent from these expressions. The off-diagonal moment
$\Cx\Cy$ and its rate $\omega_5$ never appear in the commutator on either
two-dimensional lattice: in two dimensions orthogonalisation does not touch the
off-diagonal shear moment, and the whole effect is a normal-stress and trace
phenomenon. This ceases to hold in three dimensions.

\subsection{The D3Q27 lattice}

We use the non-orthogonal central-moment basis of Ref.~\cite{derosis2017},
ordered by polynomial degree, with the second-order block split into the
five-dimensional deviatoric part
$\{\Cx\Cy,\Cx\Cz,\Cy\Cz,\Cx^2-\Cy^2,\Cx^2-\Cz^2\}$ and the trace
$\Cx^2+\Cy^2+\Cz^2$, matching the decomposition of a rank-two tensor into
irreducible representations of $SO(3)$.

The coupling graph of the order-graded orthogonalisation,
Fig.~\ref{fig:graphs}(e), has components
\begin{equation}
\begin{aligned}
  &\{1,\ \Cx^2{-}\Cy^2,\ \Cx^2{-}\Cz^2,\ \Cx^2{+}\Cy^2{+}\Cz^2,\\
  &\qquad \Cx^2\Cy^2,\ \Cx^2\Cz^2,\ \Cy^2\Cz^2,\ \Cx^2\Cy^2\Cz^2\},\\[2pt]
  &\{\Cx,\ \Cx\Cy^2,\ \Cx\Cz^2,\ \Cx\Cy^2\Cz^2\},\\
  &\{\Cy,\ \Cy\Cx^2,\ \Cy\Cz^2,\ \Cx^2\Cy\Cz^2\},\\
  &\{\Cz,\ \Cz\Cx^2,\ \Cz\Cy^2,\ \Cx^2\Cy^2\Cz\},\\[2pt]
  &\{\Cx\Cy,\ \Cx\Cy\Cz^2\},\quad \{\Cx\Cz,\ \Cx\Cy^2\Cz\},\\
  &\{\Cy\Cz,\ \Cx^2\Cy\Cz\},\quad \{\Cx\Cy\Cz\}.
\end{aligned}
\label{eq:components-d3q27}
\end{equation}
In contrast to \eqref{eq:components-d2q9}, the deviatoric moments are no longer
isolated. Both $\Cx^2-\Cy^2$ and $\Cx^2-\Cz^2$ share a component with the trace,
three fourth-order moments and the sixth-order moment, and each off-diagonal shear
moment is chained to a fourth-order moment. These connections survive the
deletion of the conserved vertices, so on physical states they tie every even
non-conserved rate to the shear rate. The odd moments fall into three components,
one per Cartesian direction, each tying two third-order rates to a fifth-order
one, and $\Cx\Cy\Cz$ is isolated.

\begin{corollary}
\label{cor:d3q27}
On D3Q27, with the order-graded orthogonal basis, the orthogonal and
non-orthogonal central-moment schemes coincide on physical states without forcing
if and only if all even non-conserved moments relax at a single rate and, in each
Cartesian direction, the odd third- and fifth-order moments share a rate, the
moment $\Cx\Cy\Cz$ being free. For rate assignments that respect the irreducible
representations of $SO(3)$ this means one rate for the even non-conserved central
moments and one for the odd, a relaxation matrix of two-relaxation-time form in
central moments. Off the physical manifold, or with forcing, the conserved
vertices re-enter and every rate must equal the conserved one.
\end{corollary}

\begin{remark}[Relation to the two-relaxation-time scheme]
\label{rem:trt}
The two-relaxation-time scheme~\cite{ginzburg2008} relaxes the parts of the
populations that are symmetric and antisymmetric under $\bfe_i\to-\bfe_i$ at two
rates, which is the same as relaxing the even and the odd \emph{raw} moments at
those rates. The condition of Corollary~\ref{cor:d3q27} concerns \emph{central}
moments, and since the shift $\Tm_{\mathrm{CM}}=\mathsf{N}(\bfu)\,\mathsf{M}$
mixes raw moments of opposite parity, the two schemes coincide only at $\bfu=0$
or when the two rates are equal. On D2Q9 with rates $1.7$ and $1.2$, for
instance, the classical scheme and its central-moment counterpart agree to
$10^{-16}$ at rest but differ by $8\times10^{-4}$ at $\bfu=(0.1,-0.05)$. The
central form is therefore an analogue of the two-relaxation-time scheme rather
than an instance of it.
\end{remark}

Every choice in which the fourth-order moments do not share the shear rate
therefore gives distinct schemes, including the common one that relaxes them at
unity, and by Proposition~\ref{prop:projection} this conclusion does not depend on
the orthogonalisation. Table~\ref{tab:d3q27} lists the leading discrepancies;
they involve $\omega$, the shear rate, which is never unity in any useful
simulation.

\begin{table}[t]
\centering
\caption{Leading entries of $[\Lam,\Am]$ on D3Q27 for
$\Lam$ with shear rate $\omega$, bulk rate $\omega_b$ and all higher moments at
unity (left), and with an independent fourth-order rate $\omega_4$ (right). The
last entry lies in the density column and is inactive on physical states without
forcing; the others are not.}
\label{tab:d3q27}
\renewcommand{\arraystretch}{1.15}
\begin{tabular}{llcc}
\toprule
row & column & higher moments at $1$ & independent $\omega_4$ \\
\midrule
$\Cx^2\Cy^2$          & $\Cx^2-\Cy^2$          & $-\left(\omega-1\right)/9$  & $-\left(\omega-\omega_4\right)/9$ \\
$\Cx^2\Cy^2$          & $\Cx^2-\Cz^2$          & $2\left(\omega-1\right)/9$  & $2\left(\omega-\omega_4\right)/9$ \\
$\Cx^2\Cz^2$          & $\Cx^2-\Cy^2$          & $2\left(\omega-1\right)/9$  & $2\left(\omega-\omega_4\right)/9$ \\
$\Cx^2\Cy\Cz$         & $\Cy\Cz$               & $\left(\omega-1\right)/3$   & $\left(\omega-\omega_4\right)/3$ \\
$\Cx^2{+}\Cy^2{+}\Cz^2$ & $1$                  & $1-\omega_b$                & $1-\omega_b$ \\
\bottomrule
\end{tabular}
\end{table}

\paragraph{Why three dimensions differ}
Table~\ref{tab:projections} gives the projections behind Theorem~\ref{thm:geom}. On
D3Q27 both routes are open, $\Cx^2\Cy^2$ through the face diagonals and
$\Cx^2\Cy\Cz$ through the body diagonals; on D3Q19, the most widely used
three-dimensional velocity set, only the first, with
$\|P_{\mathcal{D}}\,\Cx^2\Cy^2\|^2=1/27$, more than twice the D3Q27 value
(Section~\ref{sec:d19stab}); on D3Q15 neither.

\begin{table}[t]
\centering
\caption{Squared norm of the projection of higher-order monomials onto the
deviatoric second-order subspace $\mathcal{D}$, after removing the density and
trace directions, with the standard weights. The vanishing D2Q9 and D3Q15 entries
are the origin of the free shear rate, the nonvanishing D3Q19 and D3Q27 entries of the
obstruction (Theorem~\ref{thm:geom}). A dash marks a monomial that is not an
independent basis element: in two dimensions there is no $\Cz$; on D3Q19 every
velocity has a zero component, so any monomial containing $\Cx\Cy\Cz$ vanishes
identically; and on D3Q15 $\Cx^2\Cy\Cz$ coincides with the deviatoric monomial
$\Cy\Cz$.}
\label{tab:projections}
\begin{tabular}{lccc}
\toprule
 & $\Cx^2\Cy^2$ & $\Cx^2\Cy\Cz$ & $\Cx^2\Cy^2\Cz^2$ \\
\midrule
D2Q9   & $0$     & --- & --- \\
D3Q19  & $1/27$  & --- & --- \\
D3Q27  & $4/243$ & $1/81$ & $0$ \\
D3Q15  & $0$     & --- & $0$ \\
\bottomrule
\end{tabular}
\end{table}

The two-dimensional equivalence is therefore a consequence of the velocity set,
not a general property of central moments: in two dimensions, and on D3Q15,
orthogonalisation never reaches the shear moments, whereas on D3Q19 and D3Q27 it ties
them to the fourth order.

\section{Which formulation to prefer}
\label{sec:stability}

Where the criterion says that the formulations differ, we compare them by linear
stability analysis about uniform flow, by nonlinear simulation and, in three
dimensions, against a direct numerical simulation. Table~\ref{tab:summary}
summarises the outcome; the rest of this section explains it.

Linearising the collision--streaming update, as in
Refs.~\cite{lallemand2000,wissocq2019}, about a uniform base state
$f^{(0)}=f^{\mathrm{eq}}(1,\bfu_0)$ gives the amplification matrix
\begin{equation}
  \mathsf{G}(\bm{k})=\mathsf{E}(\bm{k})\,\mathsf{J},
  \quad
  \mathsf{E}(\bm{k})=\mathrm{diag}\!\left(e^{-\mathrm{i}\,\bm{k}\cdot\bfe_i}\right),
  \quad
  \mathsf{J}_{ij}=\left.\frac{\partial f^{\star}_i}{\partial f_j}\right|_{f^{(0)}}\!\!,
\end{equation}
in which $\mathsf J$ retains the dependence of the equilibrium and of $\Tm(\bfu)$ on
the flow velocity and is evaluated by complex-step differentiation. The scheme is
stable at $\bfu_0$ when $\max_{\bm k}\rho(\mathsf G)\le1+10^{-8}$ over a uniform
grid of wavevectors, refined locally on D3Q19 (Section~\ref{sec:d19stab}), and $u_{\max}$
is the end of the stable segment of base
velocities that starts from rest. The nonlinear solvers compute the orthogonal
formulation as the non-orthogonal one with $\Lam$ replaced by $\Am^{-1}\Lam\Am$
(Corollary~\ref{cor:impl}), so that no difference between the two can be an
artefact of implementation. The convergence studies are collected in
\ref{sec:convergence}, and the verification, including validation against analytic decay
rates, in \ref{sec:verification}.

\begin{table}[tbp]
\centering
\caption{Where the two formulations differ: summary of the comparisons of
Section~\ref{sec:stability}. Margin: $100\,(u_{\max}^{\mathrm{orth}}/u_{\max}^{\mathrm{non}}-1)$
in the least stable of the flow directions examined unless stated. The recipe relaxes
every moment that orthogonalisation mixes at one rate, which makes the two
formulations one scheme.}
\label{tab:summary}
\small
\begin{tabular}{@{}p{0.2\textwidth}p{0.29\textwidth}p{0.46\textwidth}@{}}
\toprule
lattice and matrix & test & result\\
\midrule
D2Q9, $\omega_b\neq1$, $\lambda_4=1$ & linear, $\omega=\omega_b=1.9$ & $+78\%$; limits within $1\%$ for $\omega_b\lesssim1.6$, $\omega\lesssim1.84$; non-orthogonal rest state unstable above $\omega_b=1.75$--$1.98$; reversal $\le3.6\%$ at $\omega=1.98$\\
 & advected shear layer & $+54\%$ on $48^2$, rising to $+70\%$ with resolution and run length\\
 & double shear layer & $+42\%$ and $+41\%$ at $\omega_b=1.9$, $\mathrm{Re}=10^4$ and $10^6$; recipe $10\%$ short of orthogonal at $10^6$\\
rectangular, $a=0.5$ & published matrix, $\omega_b=1$ & orthogonal counterpart inconsistent: shear viscosity at $45^\circ$ $5.7$ times too large\\
 & bulk rate $=$ shear rate & up to $+61\%$ along the diagonal; up to $+160\%$ at $a=0.8$\\
 & published matrix, entry~\eqref{eq:rect-key} restored & consistent; rest state lost above $\omega=1.22$ with the lattice weights; the published scheme itself in the rest-equilibrium inner product\\
 & lid-driven cavity, $100\times200$ & largest stable Re: published 6590 ($6733$ in Ref.~\cite{yahia2021}), orthogonal counterpart 1547, entry~\eqref{eq:rect-key} restored 662; bulk rate $=$ shear rate: 5276 against 2323\\
D3Q27, standard & linear, $\omega=1.9$, $1.925$, $1.95$ & $+25\%$, $+45\%$, $+73\%$; axis reversal $\le1.8\%$\\
 & advected Taylor--Green, $\omega=1.95$ & $+21\%$ on $24^3$, $+41\%$ on $48^3$, $+53\%$ at three times the run length\\
 & Taylor--Green, $32^3$, $\mathrm{Re}\le25600$ & both survive; recipe diverges for $\mathrm{Re}\ge1600$\\
 & Taylor--Green, $\mathrm{Re}=1600$, vs DNS & largest dissipation error $17\%$ against $41\%$ on $64^3$\\
D3Q19, standard & linear, $\omega=1.9$, $1.95$ & least stable (body diagonal) equal within $0.5\%$; axis $+7\%$, $+10\%$; face $+6\%$, $+5\%$\\
 & advected Taylor--Green, $\omega=1.9$, $1.95$ & $+3\%$, $+10\%$ on $24^3$; $+1\%$, $+3\%$ on $48^3$; $-1\%$, $-1\%$ at three times the run length\\
 & Taylor--Green, $32^3$, $\mathrm{Re}\le25600$ & both survive; recipe diverges for $\mathrm{Re}\ge400$\\
 & Taylor--Green, $\mathrm{Re}=1600$, vs DNS & largest dissipation error $14\%$ against $41\%$ on $64^3$\\
all & inner product, \ref{sec:ip} & advantage kept near the rest-equilibrium inner product, lost far from it\\
all & cost, Corollary~\ref{cor:impl} & $+0\%$ (D2Q9, one shear rate), $+8$--$10\%$ (D2Q9, $\omega_b\neq1$), $\le2\%$ (D3Q27)\\
\bottomrule
\end{tabular}
\end{table}

\subsection{D2Q9 with a tuned bulk viscosity}
\label{sec:d2stab}\label{sec:nonlinear}\label{sec:lambda4}

With $\Lam=\mathrm{diag}(1,1,1,\omega_b,\omega,\omega,1,1,1)$ the formulations differ
through the single term~\eqref{eq:q9}. Figure~\ref{fig:stability} shows $u_{\max}$
at $\omega=1.9$ against $\omega_b$ and maps the margin over the plane of the two
rates. On the gauge line $\omega_b=1$ the limits agree to all digits, as
Corollary~\ref{cor:d2q9} requires, and over most of the plane,
$\omega_b\lesssim1.6$ with $\omega\lesssim1.84$, they agree to within $1\%$ although
the schemes differ. Where they separate the orthogonal formulation leads, by $76\%$
along the axis and $139\%$ along the diagonal at $\omega=\omega_b=1.9$, and by $78\%$
in the least stable direction, which is the axis for the orthogonal formulation and
the diagonal for the non-orthogonal one. Its limit is
almost independent of $\omega_b$, whereas the non-orthogonal one collapses above
$\omega_b\simeq1.7$ and, beyond a threshold rising from $\omega_b=1.75$ at
$\omega=1$ to $1.92$ at $\omega=1.9$, loses the rest state altogether (white line);
by Proposition~\ref{prop:rest} the orthogonal rest state is stable for every
$\omega_b\le2$. A low bulk viscosity is therefore out of reach of the non-orthogonal
scheme with $\lambda_4=1$ at moderate shear rates. The only reversal is at the edge
of the plane, at $\omega=1.98$, where the non-orthogonal formulation is ahead by at
most $3.6\%$.

\begin{figure*}[tbp]
\centering
\begin{tikzpicture}
\begin{axis}[width=0.42\textwidth,height=5.2cm,title={$\theta=0$},
  xlabel={$\omega_b$}, ylabel={$u_{\max}$}, grid=major,
  xmin=0.3,xmax=2.0, ymin=0.2, ymax=0.68]
\addplot[blue,thick] coordinates {
 (0.4,0.4149)(0.8,0.4153)(1.0,0.4154)(1.2,0.4156)(1.4,0.4157)(1.6,0.4159)(1.7,0.4159)(1.8,0.4160)(1.9,0.4162)};
\addplot[red,thick,dashed] coordinates {
 (0.4,0.4055)(0.8,0.4146)(1.0,0.4154)(1.2,0.4157)(1.4,0.4160)(1.6,0.4160)(1.7,0.3930)(1.8,0.2986)(1.9,0.2367)};
\addplot[blue,only marks,mark=*,mark size=1.6pt] coordinates {
 (0.4,0.4688)(1.0,0.4424)(1.4,0.4336)(1.7,0.4277)(1.8,0.4248)(1.9,0.4248)};
\addplot[red,only marks,mark=square,mark size=1.6pt] coordinates {
 (0.4,0.4512)(1.0,0.4424)(1.4,0.4336)(1.7,0.4248)(1.8,0.3486)(1.9,0.2754)};
\addplot[black,densely dotted,thick,mark=triangle,mark size=1.5pt] coordinates {
 (0.4,0.4134)(1.0,0.4154)(1.4,0.4159)(1.7,0.4170)(1.8,0.4176)(1.9,0.4135)};
\end{axis}
\end{tikzpicture}\hfill
\begin{tikzpicture}
\begin{axis}[width=0.42\textwidth,height=5.2cm,title={$\theta=\pi/4$},
  xlabel={$\omega_b$}, grid=major,
  xmin=0.3,xmax=2.0, ymin=0.2, ymax=0.68]
\addplot[blue,thick] coordinates {
 (0.4,0.5748)(0.8,0.5748)(1.0,0.5748)(1.2,0.5745)(1.4,0.5722)(1.6,0.5704)(1.7,0.5696)(1.8,0.5690)(1.9,0.5591)};
\addplot[red,thick,dashed] coordinates {
 (0.4,0.5388)(0.8,0.5748)(1.0,0.5748)(1.2,0.5748)(1.4,0.5748)(1.6,0.5594)(1.7,0.5150)(1.8,0.4557)(1.9,0.2341)};
\addplot[blue,only marks,mark=*,mark size=1.6pt] coordinates {
 (0.4,0.6270)(1.0,0.6123)(1.4,0.6064)(1.7,0.6006)(1.8,0.6006)(1.9,0.5947)};
\addplot[red,only marks,mark=square,mark size=1.6pt] coordinates {
 (0.4,0.6152)(1.0,0.6123)(1.4,0.6064)(1.7,0.5713)(1.8,0.5215)(1.9,0.4658)};
\addplot[black,densely dotted,thick,mark=triangle,mark size=1.5pt] coordinates {
 (0.4,0.5748)(1.0,0.5748)(1.4,0.5693)(1.7,0.5650)(1.8,0.5638)(1.9,0.5529)};
\end{axis}
\end{tikzpicture}

\vspace{0.6em}
\begin{tikzpicture}
\begin{axis}[width=0.42\textwidth,height=5.2cm,title={$\theta=0$, margin (\%)},
  xlabel={$\omega_b$}, ylabel={$\omega$}, view={0}{90}, xmin=0.3,xmax=1.98, ymin=1,ymax=1.98,
  colormap={gain}{rgb255=(178,24,43) rgb255=(247,247,247) rgb255=(146,197,222) rgb255=(33,102,172) rgb255=(5,48,97)},
  point meta min=-50, point meta max=150, colorbar, colorbar style={width=0.18cm,
  tick label style={font=\scriptsize}, ytick={-50,0,50,100,150}, yticklabels={$-50$,$0$,$50$,$100$,$\geq150$}},
  tick label style={font=\scriptsize}, axis on top]
\addplot3[surf,shader=interp,mesh/cols=15] coordinates {(0.3,1.0,0.00)(0.42,1.0,0.00)(0.54,1.0,0.00)(0.66,1.0,0.00)(0.78,1.0,0.00)(0.9,1.0,0.00)(1.02,1.0,0.00)(1.14,1.0,0.00)(1.26,1.0,0.00)(1.38,1.0,0.00)(1.5,1.0,0.00)(1.62,1.0,0.00)(1.74,1.0,31.52)(1.86,1.0,150.00)(1.98,1.0,150.00) (0.3,1.07,0.00)(0.42,1.07,0.00)(0.54,1.07,0.00)(0.66,1.07,0.00)(0.78,1.07,0.00)(0.9,1.07,0.00)(1.02,1.07,0.00)(1.14,1.07,0.00)(1.26,1.07,0.00)(1.38,1.07,0.00)(1.5,1.07,0.00)(1.62,1.07,0.00)(1.74,1.07,29.88)(1.86,1.07,150.00)(1.98,1.07,150.00) (0.3,1.14,0.00)(0.42,1.14,0.00)(0.54,1.14,0.00)(0.66,1.14,0.00)(0.78,1.14,0.00)(0.9,1.14,0.00)(1.02,1.14,0.00)(1.14,1.14,0.00)(1.26,1.14,0.00)(1.38,1.14,0.00)(1.5,1.14,0.00)(1.62,1.14,0.00)(1.74,1.14,28.40)(1.86,1.14,150.00)(1.98,1.14,150.00) (0.3,1.21,0.00)(0.42,1.21,0.00)(0.54,1.21,0.00)(0.66,1.21,0.00)(0.78,1.21,0.00)(0.9,1.21,0.00)(1.02,1.21,0.00)(1.14,1.21,0.00)(1.26,1.21,0.00)(1.38,1.21,0.00)(1.5,1.21,0.00)(1.62,1.21,0.00)(1.74,1.21,27.07)(1.86,1.21,150.00)(1.98,1.21,150.00) (0.3,1.28,0.00)(0.42,1.28,0.00)(0.54,1.28,0.00)(0.66,1.28,0.00)(0.78,1.28,0.00)(0.9,1.28,0.00)(1.02,1.28,0.00)(1.14,1.28,0.00)(1.26,1.28,0.00)(1.38,1.28,0.00)(1.5,1.28,0.00)(1.62,1.28,0.00)(1.74,1.28,26.00)(1.86,1.28,150.00)(1.98,1.28,150.00) (0.3,1.35,0.00)(0.42,1.35,0.00)(0.54,1.35,0.00)(0.66,1.35,0.00)(0.78,1.35,0.00)(0.9,1.35,0.00)(1.02,1.35,0.00)(1.14,1.35,0.00)(1.26,1.35,0.00)(1.38,1.35,0.00)(1.5,1.35,0.00)(1.62,1.35,0.00)(1.74,1.35,24.95)(1.86,1.35,150.00)(1.98,1.35,150.00) (0.3,1.42,0.00)(0.42,1.42,0.00)(0.54,1.42,0.00)(0.66,1.42,0.00)(0.78,1.42,0.00)(0.9,1.42,0.00)(1.02,1.42,0.00)(1.14,1.42,0.00)(1.26,1.42,0.00)(1.38,1.42,0.00)(1.5,1.42,0.00)(1.62,1.42,0.00)(1.74,1.42,24.15)(1.86,1.42,150.00)(1.98,1.42,150.00) (0.3,1.49,0.00)(0.42,1.49,0.00)(0.54,1.49,0.00)(0.66,1.49,0.00)(0.78,1.49,0.00)(0.9,1.49,0.00)(1.02,1.49,0.00)(1.14,1.49,0.00)(1.26,1.49,0.00)(1.38,1.49,0.00)(1.5,1.49,0.00)(1.62,1.49,0.00)(1.74,1.49,23.36)(1.86,1.49,150.00)(1.98,1.49,150.00) (0.3,1.56,0.00)(0.42,1.56,0.00)(0.54,1.56,0.00)(0.66,1.56,0.00)(0.78,1.56,0.00)(0.9,1.56,0.00)(1.02,1.56,0.00)(1.14,1.56,0.00)(1.26,1.56,0.00)(1.38,1.56,0.00)(1.5,1.56,0.00)(1.62,1.56,0.00)(1.74,1.56,22.80)(1.86,1.56,150.00)(1.98,1.56,150.00) (0.3,1.63,0.00)(0.42,1.63,0.00)(0.54,1.63,0.00)(0.66,1.63,0.00)(0.78,1.63,0.00)(0.9,1.63,0.00)(1.02,1.63,0.00)(1.14,1.63,0.00)(1.26,1.63,0.00)(1.38,1.63,0.00)(1.5,1.63,0.00)(1.62,1.63,0.00)(1.74,1.63,22.24)(1.86,1.63,150.00)(1.98,1.63,150.00) (0.3,1.7,0.00)(0.42,1.7,0.00)(0.54,1.7,0.00)(0.66,1.7,0.00)(0.78,1.7,0.00)(0.9,1.7,0.00)(1.02,1.7,0.00)(1.14,1.7,0.00)(1.26,1.7,0.00)(1.38,1.7,0.00)(1.5,1.7,0.00)(1.62,1.7,0.00)(1.74,1.7,21.91)(1.86,1.7,150.00)(1.98,1.7,150.00) (0.3,1.77,0.30)(0.42,1.77,0.22)(0.54,1.77,0.15)(0.66,1.77,0.15)(0.78,1.77,0.00)(0.9,1.77,0.00)(1.02,1.77,0.00)(1.14,1.77,0.00)(1.26,1.77,0.00)(1.38,1.77,0.00)(1.5,1.77,-0.07)(1.62,1.77,0.00)(1.74,1.77,21.40)(1.86,1.77,64.03)(1.98,1.77,150.00) (0.3,1.84,0.37)(0.42,1.84,0.30)(0.54,1.84,0.30)(0.66,1.84,0.15)(0.78,1.84,0.07)(0.9,1.84,0.00)(1.02,1.84,0.00)(1.14,1.84,0.00)(1.26,1.84,-0.07)(1.38,1.84,-0.07)(1.5,1.84,0.00)(1.62,1.84,-0.07)(1.74,1.84,20.36)(1.86,1.84,62.27)(1.98,1.84,150.00) (0.3,1.91,3.36)(0.42,1.91,5.67)(0.54,1.91,5.92)(0.66,1.91,4.83)(0.78,1.91,2.47)(0.9,1.91,0.00)(1.02,1.91,0.00)(1.14,1.91,0.00)(1.26,1.91,0.00)(1.38,1.91,-0.08)(1.5,1.91,-0.08)(1.62,1.91,-0.08)(1.74,1.91,18.98)(1.86,1.91,60.12)(1.98,1.91,150.00) (0.3,1.98,9.36)(0.42,1.98,14.00)(0.54,1.98,13.25)(0.66,1.98,9.80)(0.78,1.98,6.17)(0.9,1.98,2.67)(1.02,1.98,-0.47)(1.14,1.98,-2.69)(1.26,1.98,-3.57)(1.38,1.98,-3.50)(1.5,1.98,-2.62)(1.62,1.98,-1.01)(1.74,1.98,6.80)(1.86,1.98,45.98)(1.98,1.98,102.29)};
\addplot3[black,thick,dashed] coordinates {(1,1,200)(1,1.98,200)};
\addplot3[black,thick,densely dotted] coordinates {(0.3,1.9,200)(1.98,1.9,200)};
\addplot3[white,very thick] coordinates {(1.7532,1.0,200)(1.7577,1.07,200)(1.7627,1.14,200)(1.7684,1.21,200)(1.7747,1.28,200)(1.7821,1.35,200)(1.7906,1.42,200)(1.8005,1.49,200)(1.8122,1.56,200)(1.8262,1.63,200)(1.8433,1.7,200)(1.8649,1.77,200)(1.8924,1.84,200)(1.9294,1.91,200)(1.9811,1.98,200)};
\end{axis}
\end{tikzpicture}\hfill
\begin{tikzpicture}
\begin{axis}[width=0.42\textwidth,height=5.2cm,title={$\theta=\pi/4$, margin (\%)},
  xlabel={$\omega_b$}, view={0}{90}, xmin=0.3,xmax=1.98, ymin=1,ymax=1.98,
  colormap={gain}{rgb255=(178,24,43) rgb255=(247,247,247) rgb255=(146,197,222) rgb255=(33,102,172) rgb255=(5,48,97)},
  point meta min=-50, point meta max=150, colorbar, colorbar style={width=0.18cm,
  tick label style={font=\scriptsize}, ytick={-50,0,50,100,150}, yticklabels={$-50$,$0$,$50$,$100$,$\geq150$}},
  tick label style={font=\scriptsize}, axis on top]
\addplot3[surf,shader=interp,mesh/cols=15] coordinates {(0.3,1.0,0.00)(0.42,1.0,0.00)(0.54,1.0,0.00)(0.66,1.0,0.00)(0.78,1.0,0.00)(0.9,1.0,0.00)(1.02,1.0,0.00)(1.14,1.0,0.00)(1.26,1.0,0.00)(1.38,1.0,0.00)(1.5,1.0,0.00)(1.62,1.0,0.00)(1.74,1.0,0.79)(1.86,1.0,150.00)(1.98,1.0,150.00) (0.3,1.07,0.00)(0.42,1.07,0.00)(0.54,1.07,0.00)(0.66,1.07,0.00)(0.78,1.07,0.00)(0.9,1.07,0.00)(1.02,1.07,0.00)(1.14,1.07,0.00)(1.26,1.07,0.00)(1.38,1.07,0.00)(1.5,1.07,0.00)(1.62,1.07,0.00)(1.74,1.07,1.27)(1.86,1.07,150.00)(1.98,1.07,150.00) (0.3,1.14,0.00)(0.42,1.14,0.00)(0.54,1.14,0.00)(0.66,1.14,0.00)(0.78,1.14,0.00)(0.9,1.14,0.00)(1.02,1.14,0.00)(1.14,1.14,0.00)(1.26,1.14,0.00)(1.38,1.14,0.00)(1.5,1.14,0.00)(1.62,1.14,0.00)(1.74,1.14,1.92)(1.86,1.14,150.00)(1.98,1.14,150.00) (0.3,1.21,0.00)(0.42,1.21,0.00)(0.54,1.21,0.00)(0.66,1.21,0.00)(0.78,1.21,0.00)(0.9,1.21,0.00)(1.02,1.21,0.00)(1.14,1.21,0.00)(1.26,1.21,0.00)(1.38,1.21,0.00)(1.5,1.21,0.00)(1.62,1.21,0.00)(1.74,1.21,2.58)(1.86,1.21,150.00)(1.98,1.21,150.00) (0.3,1.28,0.00)(0.42,1.28,0.00)(0.54,1.28,0.00)(0.66,1.28,0.00)(0.78,1.28,0.00)(0.9,1.28,0.00)(1.02,1.28,0.00)(1.14,1.28,0.00)(1.26,1.28,0.00)(1.38,1.28,0.00)(1.5,1.28,0.00)(1.62,1.28,0.00)(1.74,1.28,3.13)(1.86,1.28,150.00)(1.98,1.28,150.00) (0.3,1.35,0.00)(0.42,1.35,0.00)(0.54,1.35,0.00)(0.66,1.35,0.00)(0.78,1.35,0.00)(0.9,1.35,0.00)(1.02,1.35,0.00)(1.14,1.35,0.00)(1.26,1.35,0.00)(1.38,1.35,0.00)(1.5,1.35,0.00)(1.62,1.35,0.00)(1.74,1.35,3.69)(1.86,1.35,150.00)(1.98,1.35,150.00) (0.3,1.42,0.16)(0.42,1.42,0.05)(0.54,1.42,0.00)(0.66,1.42,0.00)(0.78,1.42,0.00)(0.9,1.42,0.00)(1.02,1.42,0.00)(1.14,1.42,0.00)(1.26,1.42,0.00)(1.38,1.42,0.00)(1.5,1.42,0.00)(1.62,1.42,0.00)(1.74,1.42,4.25)(1.86,1.42,150.00)(1.98,1.42,150.00) (0.3,1.49,0.05)(0.42,1.49,0.05)(0.54,1.49,0.05)(0.66,1.49,0.10)(0.78,1.49,0.00)(0.9,1.49,0.00)(1.02,1.49,0.00)(1.14,1.49,0.00)(1.26,1.49,0.00)(1.38,1.49,0.00)(1.5,1.49,0.00)(1.62,1.49,0.42)(1.74,1.49,4.77)(1.86,1.49,150.00)(1.98,1.49,150.00) (0.3,1.56,-0.16)(0.42,1.56,0.00)(0.54,1.56,0.00)(0.66,1.56,0.05)(0.78,1.56,0.00)(0.9,1.56,0.00)(1.02,1.56,0.00)(1.14,1.56,0.00)(1.26,1.56,0.00)(1.38,1.56,0.00)(1.5,1.56,0.00)(1.62,1.56,1.38)(1.74,1.56,5.40)(1.86,1.56,150.00)(1.98,1.56,150.00) (0.3,1.63,1.23)(0.42,1.63,-0.16)(0.54,1.63,-0.05)(0.66,1.63,-0.05)(0.78,1.63,0.00)(0.9,1.63,0.00)(1.02,1.63,0.00)(1.14,1.63,0.00)(1.26,1.63,0.00)(1.38,1.63,0.00)(1.5,1.63,0.00)(1.62,1.63,2.30)(1.74,1.63,6.28)(1.86,1.63,150.00)(1.98,1.63,150.00) (0.3,1.7,5.02)(0.42,1.7,-0.16)(0.54,1.7,-0.16)(0.66,1.7,-0.05)(0.78,1.7,-0.05)(0.9,1.7,0.00)(1.02,1.7,0.00)(1.14,1.7,0.00)(1.26,1.7,-0.05)(1.38,1.7,0.00)(1.5,1.7,0.05)(1.62,1.7,3.35)(1.74,1.7,7.17)(1.86,1.7,150.00)(1.98,1.7,150.00) (0.3,1.77,10.32)(0.42,1.77,-0.11)(0.54,1.77,-0.21)(0.66,1.77,-0.16)(0.78,1.77,-0.11)(0.9,1.77,0.00)(1.02,1.77,0.00)(1.14,1.77,0.00)(1.26,1.77,0.00)(1.38,1.77,-0.05)(1.5,1.77,1.27)(1.62,1.77,4.43)(1.74,1.77,10.36)(1.86,1.77,28.43)(1.98,1.77,150.00) (0.3,1.84,16.54)(0.42,1.84,-0.21)(0.54,1.84,-0.21)(0.66,1.84,-0.21)(0.78,1.84,-0.11)(0.9,1.84,-0.05)(1.02,1.84,0.00)(1.14,1.84,0.00)(1.26,1.84,0.00)(1.38,1.84,-0.05)(1.5,1.84,1.98)(1.62,1.84,4.98)(1.74,1.84,15.32)(1.86,1.84,31.06)(1.98,1.84,150.00) (0.3,1.91,23.19)(0.42,1.91,3.64)(0.54,1.91,0.00)(0.66,1.91,0.00)(0.78,1.91,0.00)(0.9,1.91,0.00)(1.02,1.91,0.00)(1.14,1.91,0.00)(1.26,1.91,-0.27)(1.38,1.91,-0.49)(1.5,1.91,-0.71)(1.62,1.91,1.75)(1.74,1.91,15.28)(1.86,1.91,33.11)(1.98,1.91,150.00) (0.3,1.98,14.68)(0.42,1.98,1.96)(0.54,1.98,1.36)(0.66,1.98,0.97)(0.78,1.98,0.77)(0.9,1.98,0.32)(1.02,1.98,-0.06)(1.14,1.98,-0.39)(1.26,1.98,-0.64)(1.38,1.98,-0.83)(1.5,1.98,-1.03)(1.62,1.98,-1.09)(1.74,1.98,2.94)(1.86,1.98,76.61)(1.98,1.98,150.00)};
\addplot3[black,thick,dashed] coordinates {(1,1,200)(1,1.98,200)};
\addplot3[black,thick,densely dotted] coordinates {(0.3,1.9,200)(1.98,1.9,200)};
\addplot3[white,very thick] coordinates {(1.7532,1.0,200)(1.7577,1.07,200)(1.7627,1.14,200)(1.7684,1.21,200)(1.7747,1.28,200)(1.7821,1.35,200)(1.7906,1.42,200)(1.8005,1.49,200)(1.8122,1.56,200)(1.8262,1.63,200)(1.8433,1.7,200)(1.8649,1.77,200)(1.8924,1.84,200)(1.9294,1.91,200)(1.9811,1.98,200)};
\end{axis}
\end{tikzpicture}
\caption{Top: maximum stable base velocity against the bulk relaxation rate at
$\omega=1.9$, for base flow aligned with the lattice (left) and along the diagonal
(right). Lines are the von Neumann analysis, symbols the nonlinear simulation of
an advected shear layer ($48^2$, $700$ steps). Blue solid lines and filled circles
are the orthogonal formulation; red dashed lines and open squares the
non-orthogonal one. Black dotted lines with triangles are the von Neumann limit
when the fourth-order moment relaxes at $\omega_b$, for which the two formulations
coincide (Section~\ref{sec:lambda4}). Both methods agree that the two
formulations coincide at $\omega_b=1$, that the orthogonal limit is nearly
independent of $\omega_b$, and that the non-orthogonal limit collapses above
$\omega_b\simeq1.7$. The simulated separation is smaller than the linear one:
$+54.3\%$ against $+75.8\%$ at $\theta=0$, and $+27.7\%$ against $+138.8\%$ at
$\theta=\pi/4$, both at $\omega_b=1.9$; it widens towards the linear value with
resolution and run length (Table~\ref{tab:convergence}). Bottom: the margin
$100\,(u_{\max}^{\mathrm{orth}}/u_{\max}^{\mathrm{non}}-1)$ of the von Neumann limits
over the plane of the bulk and shear rates, on a $15\times15$ grid; the dashed line
$\omega_b=1$ is the gauge line, on which the two formulations are one scheme, and
the dotted line $\omega=1.9$ is the section shown above. The colour saturates where
the margin exceeds $150\%$; to the right of the white line the non-orthogonal rest
state is itself unstable.}
\label{fig:stability}
\end{figure*}
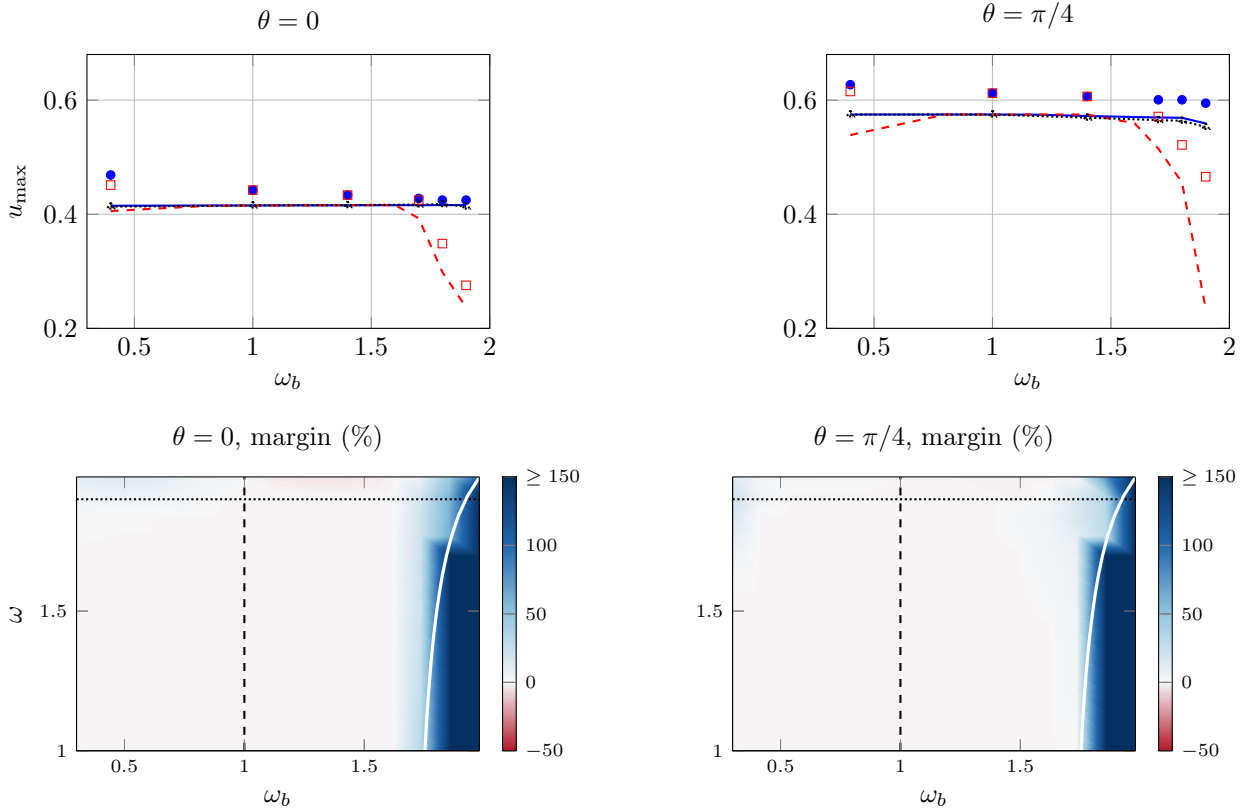

Nonlinear simulations confirm the picture. In a weakly perturbed uniform drift on
$48^2$, the analogue of the linear base state, the survival limits (symbols in
Fig.~\ref{fig:stability}) agree at $\omega_b=1$, keep the orthogonal limit nearly
independent of $\omega_b$ and show the collapse above $\omega_b\simeq1.7$; the
margin, $+54\%$ at $\omega_b=1.9$, widens towards the linear $+76\%$ with resolution
and run length, to $+70\%$ at three times the run length (Table~\ref{tab:convergence}).
In the under-resolved doubly periodic shear layer of Ref.~\cite{brown1995} ($128^2$,
$\kappa=80$, $\delta=0.05$), used for bulk-viscosity studies in
Ref.~\cite{dellar2001}, the orthogonal scheme survives to a $42\%$ higher layer
velocity at $\omega_b=1.9$ and $\mathrm{Re}=10^4$ (Table~\ref{tab:dsl}); at
$U_0\le0.2$ every scheme survives up to $\mathrm{Re}=10^6$, so the difference is a
finite-Mach effect, and away from the stability boundary it amounts to a relative
$10^{-6}$ in the populations.

\begin{table}[tbp]
\centering
\caption{Largest layer velocity $U_0$ for which the double shear layer ($128^2$,
$\kappa=80$, $\delta=0.05$, $\omega\to2$ through $\mathrm{Re}=U_0N/\nu$) survives
two convective times: orthogonal (orth.) and non-orthogonal (non.) schemes with
$\lambda_4=1$, and the non-orthogonal scheme with $\lambda_4=\omega_b$, for which
the two formulations coincide. Bisection resolution $2.3\times10^{-3}$; the
entries $\geq0.450$ reached the top of the search interval.}
\label{tab:dsl}
\small\setlength{\tabcolsep}{5pt}
\begin{tabular}{lcccccc}
\toprule
 & \multicolumn{3}{c}{$\mathrm{Re}=10^4$} & \multicolumn{3}{c}{$\mathrm{Re}=10^6$}\\
\cmidrule(lr){2-4}\cmidrule(lr){5-7}
$\omega_b$ & orth. & non. & $\lambda_4=\omega_b$ & orth. & non. & $\lambda_4=\omega_b$\\
\midrule
$1.0$ & $\geq0.450$ & $\geq0.450$ & $\geq0.450$ & $0.375$ & $0.375$ & $0.375$\\
$1.4$ & $0.396$ & $0.394$ & $0.396$ & $0.359$ & $0.354$ & $0.356$\\
$1.8$ & $0.354$ & $0.328$ & $0.354$ & $0.323$ & $0.288$ & $0.300$\\
$1.9$ & $0.347$ & $0.244$ & $0.347$ & $0.316$ & $0.225$ & $0.286$\\
\bottomrule
\end{tabular}
\end{table}

\begin{figure}[tbp]
\centering
\pgfplotsset{dslmap/.style={width=0.40\textwidth, axis equal image, enlargelimits=false,
  tick label style={font=\scriptsize}, colorbar, colorbar style={width=0.18cm,
  tick label style={font=\scriptsize,/pgf/number format/fixed}}, axis on top}}
\begin{tikzpicture}
\begin{axis}[dslmap, xmin=0,xmax=128,ymin=0,ymax=128, title={(a) $\rho-\bar\rho$, orthogonal},
  ylabel={$y$}, colormap/viridis, colorbar style={ytick={-0.3,-0.2,-0.1,0,0.1}}, point meta min=-0.3504, point meta max=0.1293]
\addplot graphics[xmin=0,xmax=128,ymin=0,ymax=128] {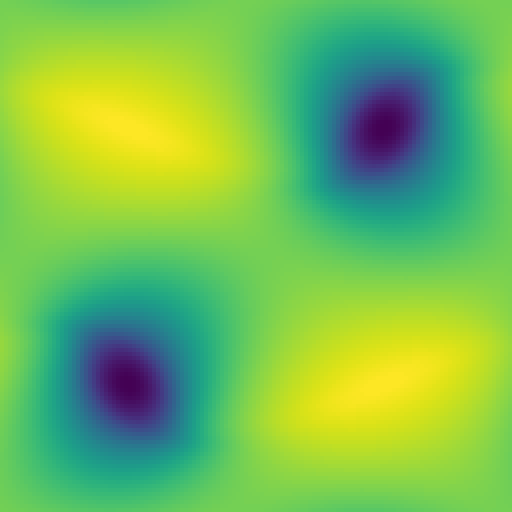};
\end{axis}
\end{tikzpicture}\hfill
\begin{tikzpicture}
\begin{axis}[dslmap, xmin=0,xmax=128,ymin=0,ymax=128, title={(b) $\rho-\bar\rho$, non-orthogonal},
  colormap/viridis, colorbar style={ytick={-0.3,-0.2,-0.1,0,0.1}}, point meta min=-0.3504, point meta max=0.1293]
\addplot graphics[xmin=0,xmax=128,ymin=0,ymax=128] {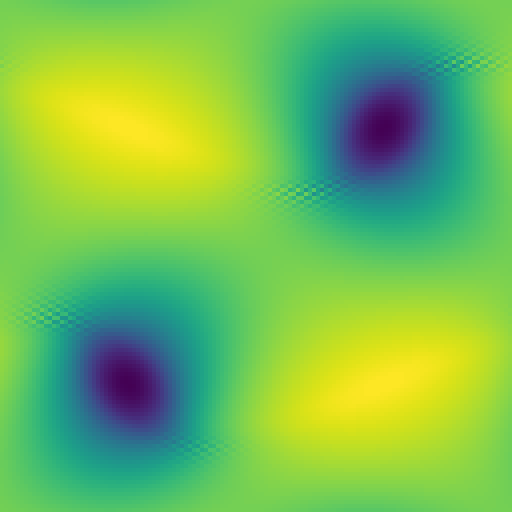};
\draw[white,thick] (axis cs:96,96) rectangle (axis cs:127.5,127.5);
\end{axis}
\end{tikzpicture}

\vspace{0.4em}
\begin{tikzpicture}
\begin{axis}[dslmap, xmin=96,xmax=128,ymin=96,ymax=128, title={(c) detail of (b)},
  xlabel={$x$}, ylabel={$y$}, colormap/viridis, colorbar style={ytick={-0.3,-0.2,-0.1,0,0.1}}, point meta min=-0.3504, point meta max=0.1293]
\addplot graphics[xmin=96,xmax=128,ymin=96,ymax=128] {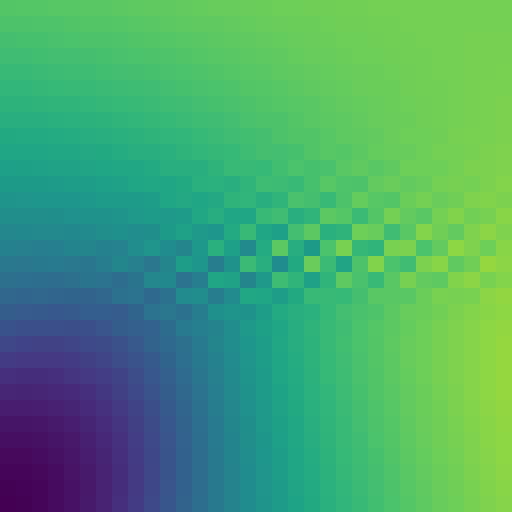};
\end{axis}
\end{tikzpicture}\hfill
\begin{tikzpicture}
\begin{axis}[dslmap, xmin=0,xmax=128,ymin=0,ymax=128, title={(d) vorticity, non-orthogonal},
  xlabel={$x$}, colormap={div}{rgb255=(33,102,172) rgb255=(247,247,247) rgb255=(178,24,43)},
  point meta min=-0.0901, point meta max=0.0901]
\addplot graphics[xmin=0,xmax=128,ymin=0,ymax=128] {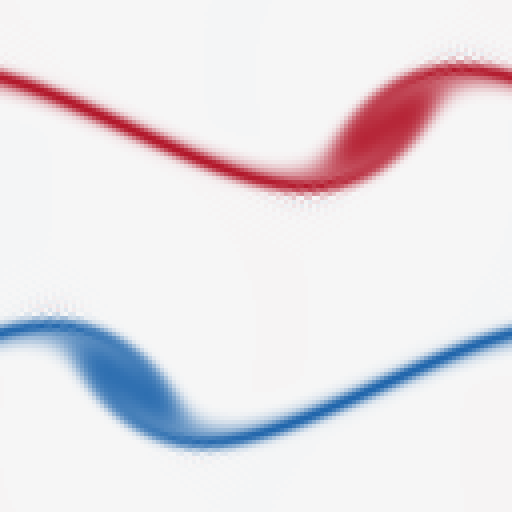};
\end{axis}
\end{tikzpicture}
\caption{Double shear layer on D2Q9 ($128^2$, $\kappa=80$, $\delta=0.05$,
$U_0=0.30$, $\omega_b=1.9$, $\mathrm{Re}=10^4$), just beyond the stability limit of
the non-orthogonal scheme, after $340$ steps ($t=0.80\,N/U_0$). (a,b) Density
fluctuation for the orthogonal and non-orthogonal schemes; (c) detail of (b), the
region marked in white; (d) vorticity of the non-orthogonal run. The non-orthogonal
run diverges $30$ steps later; the orthogonal one survives.}
\label{fig:dslmap}
\end{figure}

Figure~\ref{fig:dslmap} shows how the non-orthogonal scheme fails just beyond its
limit. Packets of grid-scale density oscillation appear on the braids between the
vortices and grow by about $10\%$ per step until the run diverges. Their dominant
wavevector, $(1.67,-2.99)$, is that of the most unstable linear mode for a uniform
stream at the layer speed, $(\pm1.72,\pi)$ in Fig.~\ref{fig:modes}, and the growing
part of the difference between the two runs is compressive, with twice as much
divergence as vorticity at the grid scale. The linear mode itself is carried by the
density and the trace, which hold $34\%$ and $36\%$ of it, and one collision raises
its fourth-order content from $3\%$ to $24\%$ (Fig.~\ref{fig:modes}b): it is the
transfer of Eq.~\eqref{eq:q9}. The collapse is thus a consequence of relaxing the
raw fourth-order moment at unity while the trace relaxes at $\omega_b$, not of
non-orthogonality as such.

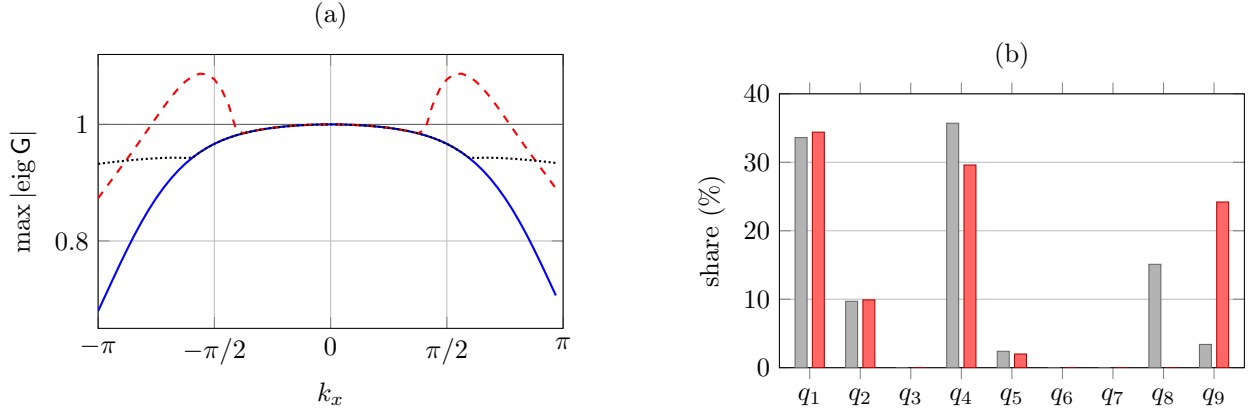
\begin{figure}[tbp]
\centering
\begin{tikzpicture}
\begin{axis}[width=0.47\textwidth,height=5.2cm,title={(a)},
  xlabel={$k_x$}, ylabel={$\max|\text{eig}\,\mathsf{G}|$}, grid=major,
  xmin=-3.1416,xmax=3.1416, ymin=0.65, ymax=1.12,
  xtick={-3.1416,-1.5708,0,1.5708,3.1416},
  xticklabels={$-\pi$,$-\pi/2$,$0$,$\pi/2$,$\pi$}]
\addplot[black!60,thin] coordinates {(-3.1416,1)(3.1416,1)};
\addplot[blue,thick] coordinates {(-3.1416,0.67971)(-3.0434,0.70649)(-2.9452,0.73342)(-2.8471,0.75997)(-2.7489,0.78567)(-2.6507,0.81013)(-2.5525,0.83308)(-2.4544,0.85429)(-2.3562,0.87365)(-2.2580,0.89112)(-2.1598,0.90672)(-2.0617,0.92052)(-1.9635,0.93261)(-1.8653,0.94312)(-1.7671,0.95220)(-1.6690,0.96000)(-1.5708,0.96667)(-1.4726,0.97235)(-1.3744,0.97718)(-1.2763,0.98127)(-1.1781,0.98475)(-1.0799,0.98769)(-0.9817,0.99018)(-0.8836,0.99229)(-0.7854,0.99407)(-0.6872,0.99557)(-0.5890,0.99680)(-0.4909,0.99781)(-0.3927,0.99862)(-0.2945,0.99923)(-0.1963,0.99966)(-0.0982,0.99992)(0.0000,1.00000)(0.0982,0.99992)(0.1963,0.99966)(0.2945,0.99923)(0.3927,0.99862)(0.4909,0.99781)(0.5890,0.99680)(0.6872,0.99557)(0.7854,0.99407)(0.8836,0.99229)(0.9817,0.99018)(1.0799,0.98769)(1.1781,0.98475)(1.2763,0.98127)(1.3744,0.97718)(1.4726,0.97235)(1.5708,0.96667)(1.6690,0.96000)(1.7671,0.95220)(1.8653,0.94312)(1.9635,0.93261)(2.0617,0.92052)(2.1598,0.90672)(2.2580,0.89112)(2.3562,0.87365)(2.4544,0.85429)(2.5525,0.83308)(2.6507,0.81013)(2.7489,0.78567)(2.8471,0.75997)(2.9452,0.73342)(3.0434,0.70649)};
\addplot[red,thick,dashed] coordinates {(-3.1416,0.87317)(-3.0434,0.89026)(-2.9452,0.90714)(-2.8471,0.92396)(-2.7489,0.94093)(-2.6507,0.95820)(-2.5525,0.97585)(-2.4544,0.99380)(-2.3562,1.01181)(-2.2580,1.02944)(-2.1598,1.04612)(-2.0617,1.06114)(-1.9635,1.07363)(-1.8653,1.08263)(-1.7671,1.08709)(-1.6690,1.08579)(-1.5708,1.07736)(-1.4726,1.06015)(-1.3744,1.03209)(-1.2763,0.99056)(-1.1781,0.98475)(-1.0799,0.98769)(-0.9817,0.99018)(-0.8836,0.99229)(-0.7854,0.99407)(-0.6872,0.99557)(-0.5890,0.99680)(-0.4909,0.99781)(-0.3927,0.99862)(-0.2945,0.99923)(-0.1963,0.99966)(-0.0982,0.99992)(0.0000,1.00000)(0.0982,0.99992)(0.1963,0.99966)(0.2945,0.99923)(0.3927,0.99862)(0.4909,0.99781)(0.5890,0.99680)(0.6872,0.99557)(0.7854,0.99407)(0.8836,0.99229)(0.9817,0.99018)(1.0799,0.98769)(1.1781,0.98475)(1.2763,0.99056)(1.3744,1.03209)(1.4726,1.06015)(1.5708,1.07736)(1.6690,1.08579)(1.7671,1.08709)(1.8653,1.08263)(1.9635,1.07363)(2.0617,1.06114)(2.1598,1.04612)(2.2580,1.02944)(2.3562,1.01181)(2.4544,0.99380)(2.5525,0.97585)(2.6507,0.95820)(2.7489,0.94093)(2.8471,0.92396)(2.9452,0.90714)(3.0434,0.89026)};
\addplot[black,thick,densely dotted] coordinates {(-3.1416,0.93220)(-3.0434,0.93384)(-2.9452,0.93534)(-2.8471,0.93673)(-2.7489,0.93799)(-2.6507,0.93915)(-2.5525,0.94020)(-2.4544,0.94113)(-2.3562,0.94191)(-2.2580,0.94251)(-2.1598,0.94286)(-2.0617,0.94287)(-1.9635,0.94241)(-1.8653,0.94312)(-1.7671,0.95220)(-1.6690,0.96000)(-1.5708,0.96667)(-1.4726,0.97235)(-1.3744,0.97718)(-1.2763,0.98127)(-1.1781,0.98475)(-1.0799,0.98769)(-0.9817,0.99018)(-0.8836,0.99229)(-0.7854,0.99407)(-0.6872,0.99557)(-0.5890,0.99680)(-0.4909,0.99781)(-0.3927,0.99862)(-0.2945,0.99923)(-0.1963,0.99966)(-0.0982,0.99992)(0.0000,1.00000)(0.0982,0.99992)(0.1963,0.99966)(0.2945,0.99923)(0.3927,0.99862)(0.4909,0.99781)(0.5890,0.99680)(0.6872,0.99557)(0.7854,0.99407)(0.8836,0.99229)(0.9817,0.99018)(1.0799,0.98769)(1.1781,0.98475)(1.2763,0.98127)(1.3744,0.97718)(1.4726,0.97235)(1.5708,0.96667)(1.6690,0.96000)(1.7671,0.95220)(1.8653,0.94312)(1.9635,0.94241)(2.0617,0.94287)(2.1598,0.94286)(2.2580,0.94251)(2.3562,0.94191)(2.4544,0.94113)(2.5525,0.94020)(2.6507,0.93915)(2.7489,0.93799)(2.8471,0.93673)(2.9452,0.93534)(3.0434,0.93384)};
\end{axis}
\end{tikzpicture}\hfill
\begin{tikzpicture}
\begin{axis}[width=0.47\textwidth,height=5.2cm,title={(b)},
  ybar, bar width=4.5pt, ymin=0, ymax=40, ylabel={share (\%)},
  xtick={1,...,9}, xticklabels={$q_1$,$q_2$,$q_3$,$q_4$,$q_5$,$q_6$,$q_7$,$q_8$,$q_9$},
  xmin=0.4, xmax=9.6, enlarge x limits=false, ymajorgrids]
\addplot[fill=black!30,draw=black!60] coordinates {(1,33.6)(2,9.7)(3,0.0)(4,35.7)(5,2.4)(6,0.0)(7,0.0)(8,15.1)(9,3.4)};
\addplot[fill=red!60,draw=red!70!black] coordinates {(1,34.4)(2,9.9)(3,0.0)(4,29.6)(5,2.0)(6,0.0)(7,0.0)(8,0.0)(9,24.2)};
\end{axis}
\end{tikzpicture}
\caption{The unstable mode of the non-orthogonal D2Q9 scheme at
$\omega=\omega_b=1.9$ and $\lambda_4=1$, with base flow $u_0=0.30$ along the axis,
between the two stability limits. (a) Spectral radius of the amplification matrix
along $k_x$ at $k_y=\pi$: orthogonal (blue solid), non-orthogonal (red dashed) and
non-orthogonal with $\lambda_4=\omega_b$ (black dotted); only the second exceeds
unity. (b) Content of the most unstable eigenvector in the Hermite central moments
$q_1,\dots,q_9$, as shares of $\sum_a|q_a|^2/\avg{q_a,q_a}$, before (grey) and
after (red) collision. The mode is carried by the density $q_1$ and the trace
$q_4$, and the collision moves its trace content into the fourth-order moment
$q_9$, the transfer of Eq.~\eqref{eq:q9}.}
\label{fig:modes}
\end{figure}

The recipe, $\lambda_4=\omega_b$, removes the term: the two formulations then have
Jacobians equal to round-off and identical limits, and the common limit (dotted in
Fig.~\ref{fig:stability}) lies within $1\%$ of the orthogonal one. It is not free.
With $\omega_b$ close to two the fourth-order moment is itself barely damped, and in
the double shear layer at $\mathrm{Re}=10^6$ and $\omega_b=1.9$ the recipe falls $10\%$
short of the orthogonal scheme with $\lambda_4=1$.

\subsection{Rectangular lattices}
\label{sec:rectstab}

We test the claim on the scheme it was made for: the rectangular central-moment
scheme of Ref.~\cite{yahia2021}, with the basis~\eqref{eq:basis-d2q9} on the stretched
lattice, equilibria matched to the Maxwellian with a free sound speed, and extended
second-order equilibria that restore isotropy and Galilean invariance through
locally computed strain rates and a central-difference density gradient. Its
orthogonal counterpart has the same population-space equilibrium and the rest-frame
Gram--Schmidt basis of Section~\ref{sec:when}: it is not a scheme proposed in the
literature, but the one that differs from the published scheme in the basis alone,
which is what the claim concerns. The linear analysis retains the
wavevector dependence of the finite-difference gradient. We take $a=0.5$ with
$c_s^2=0.1$, the setting of the stability tests of Ref.~\cite{yahia2021}, and $a=0.8$
with $c_s^2=1/3$. The rest state of the published scheme at $a=0.5$ is stable up to
$\omega=1.8705$; the largest Reynolds numbers reported for it in a lid-driven cavity
correspond to $\omega=1.853$ and $1.859$ on the two finer grids, below this limit,
and to $1.888$ on the coarsest, slightly above it, where the periodic analysis
predicts a growth of $0.7\%$ per step that the wall-bounded cavity evidently does not
sustain.

With the published relaxation matrix, bulk rate at unity and shear rates at $\omega$,
the entry~\eqref{eq:rect-key} lies in the hydrodynamic block, and the orthogonal
counterpart is not a consistent discretisation of the same equations. At
$\omega=1.6$ its shear viscosity for a wave at $45^\circ$ is $5.7$ times too large at
$a=0.5$ and $1.38$ times at $a=0.8$, whereas the non-orthogonal scheme returns the
shear viscosity and the sound attenuation to four digits in every direction. It is
also far less stable (Fig.~\ref{fig:rect}), in line with Ref.~\cite{yahia2021}. With
the bulk rate equal to the shear rate the two share their transport coefficients and
the order reverses: the orthogonal formulation is as stable or more stable in every
direction, by up to $61\%$ along the diagonal at $a=0.5$ and $160\%$ at $a=0.8$, and
the recipe, trace, deviatoric and fourth-order moments at $\omega$, gives identical
limits in the two bases and keeps the rest state stable to $\omega=1.999$.

\paragraph{What the published scheme is} In the inner product of the rest equilibrium
$\Am_{95}=0$. With the published matrix, whose bulk and higher-order rates are all unity,
a single entry of the commutator then acts on physical states: the $(5,4)$ entry, which
couples the normal-stress difference to the trace and whose lattice-weight value
is~\eqref{eq:rect-key}. The published scheme is therefore exactly the orthogonal scheme
of that inner product with one off-diagonal relaxation entry, in that position, acting on
physical states; the two give the same post-collision populations to round-off. The non-orthogonal basis does not so much avoid
orthogonality as supply the one off-diagonal entry that consistency with an independent
bulk rate requires. Dropping that entry gives the inconsistent counterpart above.
Keeping it but orthogonalising in the lattice-weight inner product, where $\Am_{95}\ne0$
couples the fourth order to the normal-stress difference as well, restores the transport
coefficients, with the shear viscosity within $3\times10^{-5}$ of its nominal value in
every direction, but costs stability: the rest state is lost above
$\omega=1.22$ at $a=0.5$ and $1.88$ at $a=0.8$, against $1.87$ and $1.95$ for the
published scheme. The claim of Refs.~\cite{yahia2021,yahia2021cuboid} is therefore right
for the scheme it was made for, and the analysis shows why: an orthogonal design that
keeps an independent bulk rate has to restore the entry~\eqref{eq:rect-key}, in its
relaxation matrix as here or through additional terms in its equilibria as in
Refs.~\cite{zong2016,peng2016}; restored in the rest-equilibrium inner product it gives
the published scheme, and in the lattice-weight one it costs stability. Their second
argument, that orthogonal transforms on a stretched lattice have unwieldy coefficients
that depend on the aspect ratio, Eq.~\eqref{eq:A-rect}, concerns the implementation
only: by Corollary~\ref{cor:impl} an orthogonal formulation can be run through the
non-orthogonal transforms.

A lid-driven cavity confirms both results in a wall-bounded nonlinear flow
(\ref{sec:cavity}). In the stability test of Ref.~\cite{yahia2021}, with its
parameters ($a=0.5$, $c_s^2=0.1$, lid velocity $0.2$, $100\times200$ nodes, $100\,000$
steps), the published scheme stays stable up to $\mathrm{Re}=6590$, against $6733$ reported
there, the orthogonal counterpart up to $\mathrm{Re}=1547$, and the orthogonal scheme with the
entry~\eqref{eq:rect-key} restored only up to $\mathrm{Re}=662$; at $\mathrm{Re}=400$ the orthogonal
counterpart departs from a square-lattice reference by $17\%$ of the lid velocity on
the centrelines, against $0.7\%$ for the published scheme. With the bulk rate equal to the
shear rate the orthogonal formulation again leads, $\mathrm{Re}=5276$ against $2323$,
although at this Mach number, $0.63$, both fall short of the published matrix, whose
bulk rate at unity keeps the bulk viscosity high.

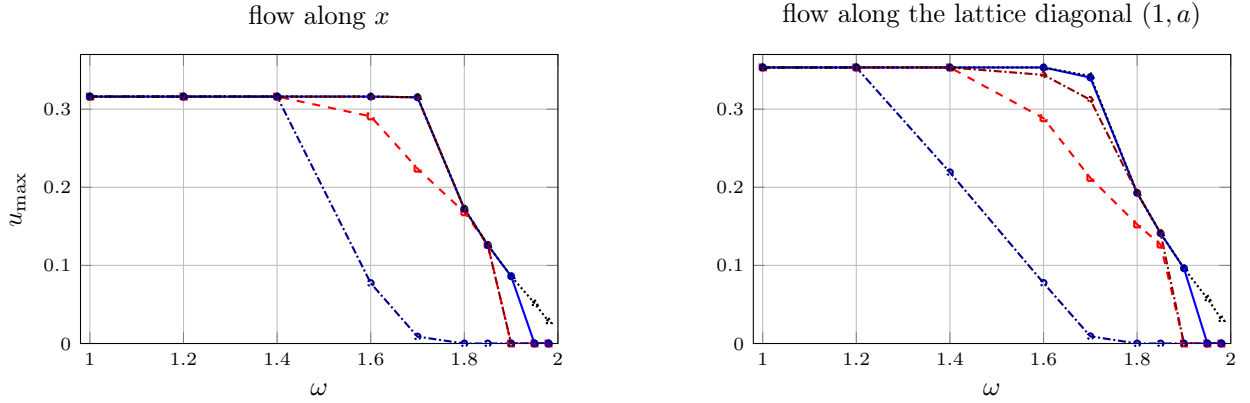
\begin{figure}[tbp]
\centering
\begin{tikzpicture}
\begin{axis}[width=0.48\textwidth,height=5.4cm,title={flow along $x$},
  xlabel={$\omega$}, ylabel={$u_{\max}$}, grid=major, xmin=0.98,xmax=2.0, ymin=0,ymax=0.37,
  tick label style={font=\scriptsize},
  yticklabel style={/pgf/number format/fixed,/pgf/number format/precision=2}]
\addplot[red,thick,dashed,mark=square,mark size=1.1pt] coordinates {(1.000,0.3161)(1.200,0.3161)(1.400,0.3161)(1.600,0.2908)(1.700,0.2238)(1.800,0.1681)(1.850,0.1262)(1.900,0.0000)(1.950,0.0000)(1.980,0.0000)};
\addplot[blue,thick,mark=*,mark size=1.1pt] coordinates {(1.000,0.3161)(1.200,0.3161)(1.400,0.3161)(1.600,0.3161)(1.700,0.3152)(1.800,0.1725)(1.850,0.1262)(1.900,0.0861)(1.950,0.0000)(1.980,0.0000)};
\addplot[red!55!black,thick,densely dashdotted,mark=diamond,mark size=1.3pt] coordinates {(1.000,0.3161)(1.200,0.3161)(1.400,0.3161)(1.600,0.3161)(1.700,0.3152)(1.800,0.1725)(1.850,0.1262)(1.900,0.0000)(1.950,0.0000)(1.980,0.0000)};
\addplot[blue!55!black,thick,densely dashdotted,mark=o,mark size=1.1pt] coordinates {(1.000,0.3161)(1.200,0.3161)(1.400,0.3161)(1.600,0.0773)(1.700,0.0088)(1.800,0.0000)(1.850,0.0000)(1.900,0.0000)(1.950,0.0000)(1.980,0.0000)};
\addplot[black,thick,densely dotted,mark=triangle,mark size=1.3pt] coordinates {(1.000,0.3161)(1.200,0.3161)(1.400,0.3161)(1.600,0.3161)(1.700,0.3152)(1.800,0.1725)(1.850,0.1262)(1.900,0.0861)(1.950,0.0516)(1.980,0.0283)};
\end{axis}
\end{tikzpicture}\hfill
\begin{tikzpicture}
\begin{axis}[width=0.48\textwidth,height=5.4cm,title={flow along the lattice diagonal $(1,a)$},
  xlabel={$\omega$}, grid=major, xmin=0.98,xmax=2.0, ymin=0,ymax=0.37,
  tick label style={font=\scriptsize},
  yticklabel style={/pgf/number format/fixed,/pgf/number format/precision=2}]
\addplot[red,thick,dashed,mark=square,mark size=1.1pt] coordinates {(1.000,0.3534)(1.200,0.3534)(1.400,0.3534)(1.600,0.2884)(1.700,0.2120)(1.800,0.1525)(1.850,0.1267)(1.900,0.0000)(1.950,0.0000)(1.980,0.0000)};
\addplot[blue,thick,mark=*,mark size=1.1pt] coordinates {(1.000,0.3534)(1.200,0.3534)(1.400,0.3534)(1.600,0.3534)(1.700,0.3405)(1.800,0.1928)(1.850,0.1412)(1.900,0.0963)(1.950,0.0000)(1.980,0.0000)};
\addplot[red!55!black,thick,densely dashdotted,mark=diamond,mark size=1.3pt] coordinates {(1.000,0.3534)(1.200,0.3534)(1.400,0.3534)(1.600,0.3439)(1.700,0.3120)(1.800,0.1928)(1.850,0.1412)(1.900,0.0000)(1.950,0.0000)(1.980,0.0000)};
\addplot[blue!55!black,thick,densely dashdotted,mark=o,mark size=1.1pt] coordinates {(1.000,0.3534)(1.200,0.3534)(1.400,0.2191)(1.600,0.0775)(1.700,0.0092)(1.800,0.0000)(1.850,0.0000)(1.900,0.0000)(1.950,0.0000)(1.980,0.0000)};
\addplot[black,thick,densely dotted,mark=triangle,mark size=1.3pt] coordinates {(1.000,0.3534)(1.200,0.3534)(1.400,0.3534)(1.600,0.3534)(1.700,0.3425)(1.800,0.1928)(1.850,0.1412)(1.900,0.0963)(1.950,0.0577)(1.980,0.0311)};
\end{axis}
\end{tikzpicture}
\caption{Linear stability limits of the rectangular central-moment scheme of
Ref.~\cite{yahia2021} and of its orthogonal counterpart at $a=0.5$, $c_s^2=0.1$
($32^2$ wavevectors, finite-difference density gradient included). Blue:
orthogonal; red: non-orthogonal. Solid and dashed lines with bulk and shear rates
equal; dash-dotted lines with the published matrix, bulk rate at unity, for which
the orthogonal counterpart is not hydrodynamically consistent (see text). Black
dotted: the recipe, trace, deviatoric and fourth-order moments at $\omega$, for
which the two formulations coincide. A value of zero means that the rest state is
itself unstable.}
\label{fig:rect}
\end{figure}

\subsection{The D3Q27 lattice}
\label{sec:d3stab}

With the standard matrix, deviatoric moments at $\omega$ and the trace and every
higher moment at unity, the formulations differ for every $\omega\ne1$
(Corollary~\ref{cor:d3q27}). Their linear limits (Fig.~\ref{fig:d3stab}; $16^3$
wavevectors, flow along the axis and the face and body diagonals) coincide for
$\omega\le1.7$. Beyond, the non-orthogonal formulation is marginally ahead along the
axis, by at most $1.8\%$, and collapses along both diagonals; because its limiting
direction moves to the face diagonal, the orthogonal formulation leads in the least
stable direction by $25\%$, $45\%$ and $73\%$ at $\omega=1.9$, $1.925$ and $1.95$. An
advected Taylor--Green vortex~\cite{brachet1983} reproduces the ordering (symbols),
with margins that grow with resolution and run length towards the linear ones, at
$\omega=1.95$ from $21\%$ on $24^3$ to $41\%$ on $48^3$ and $53\%$ at three times the
run length (Table~\ref{tab:d3conv}).

\begin{figure}[tbp]
\centering
\begin{tikzpicture}
\begin{axis}[width=0.355\textwidth,height=5cm,title={axis $(1,0,0)$},
  xlabel={$\omega$}, ylabel={$u_{\max}$}, grid=major, xmin=0.95,xmax=1.98, ymin=0.39,ymax=0.445,
  xtick={1,1.2,1.4,1.6,1.8}, tick label style={font=\scriptsize},
  yticklabel style={/pgf/number format/fixed,/pgf/number format/precision=2}]
\addplot[blue,thick,mark=*,mark size=1.1pt] coordinates {(1.000,0.4226)(1.400,0.4226)(1.600,0.4226)(1.700,0.4226)(1.750,0.4226)(1.800,0.4212)(1.850,0.4189)(1.900,0.4152)(1.925,0.4128)(1.950,0.3966)};
\addplot[red,thick,dashed,mark=square,mark size=1.1pt] coordinates {(1.000,0.4226)(1.400,0.4226)(1.600,0.4226)(1.700,0.4226)(1.750,0.4226)(1.800,0.4226)(1.850,0.4221)(1.900,0.4198)(1.925,0.4156)(1.950,0.4036)};
\addplot[black,thick,densely dotted] coordinates {(1.000,0.4226)(1.400,0.4226)(1.600,0.4226)(1.700,0.4226)(1.750,0.4226)(1.800,0.4226)(1.850,0.4203)(1.900,0.4156)(1.925,0.4110)(1.950,0.4063)};
\addplot[blue,only marks,mark=*,mark size=1.7pt] coordinates {(1.9,0.4359)(1.95,0.4324)};
\addplot[red,only marks,mark=square,mark size=1.7pt] coordinates {(1.9,0.4359)(1.95,0.4289)};
\addplot[black,only marks,mark=triangle,mark size=1.9pt] coordinates {(1.9,0.4254)(1.95,0.4219)};
\end{axis}
\end{tikzpicture}\hfill
\begin{tikzpicture}
\begin{axis}[width=0.355\textwidth,height=5cm,title={face diagonal $(1,1,0)$},
  xlabel={$\omega$}, grid=major, xmin=0.95,xmax=1.98, ymin=0.2,ymax=0.65,
  xtick={1,1.2,1.4,1.6,1.8}, tick label style={font=\scriptsize},
  yticklabel style={/pgf/number format/fixed,/pgf/number format/precision=2}]
\addplot[blue,thick,mark=*,mark size=1.1pt] coordinates {(1.000,0.5975)(1.400,0.5975)(1.600,0.5975)(1.700,0.5975)(1.750,0.5975)(1.800,0.5956)(1.850,0.5933)(1.900,0.5743)(1.925,0.5580)(1.950,0.5367)};
\addplot[red,thick,dashed,mark=square,mark size=1.1pt] coordinates {(1.000,0.5975)(1.400,0.5975)(1.600,0.5975)(1.700,0.5975)(1.750,0.5771)(1.800,0.4996)(1.850,0.4198)(1.900,0.3331)(1.925,0.2844)(1.950,0.2292)};
\addplot[black,thick,densely dotted] coordinates {(1.000,0.5975)(1.400,0.5975)(1.600,0.5975)(1.700,0.5975)(1.750,0.5965)(1.800,0.5877)(1.850,0.5789)(1.900,0.5617)(1.925,0.5460)(1.950,0.5242)};
\addplot[blue,only marks,mark=*,mark size=1.7pt] coordinates {(1.9,0.6152)(1.95,0.6152)};
\addplot[red,only marks,mark=square,mark size=1.7pt] coordinates {(1.9,0.4395)(1.95,0.3586)};
\addplot[black,only marks,mark=triangle,mark size=1.9pt] coordinates {(1.9,0.5871)(1.95,0.5801)};
\end{axis}
\end{tikzpicture}\hfill
\begin{tikzpicture}
\begin{axis}[width=0.355\textwidth,height=5cm,title={body diagonal $(1,1,1)$},
  xlabel={$\omega$}, grid=major, xmin=0.95,xmax=1.98, ymin=0.25,ymax=0.79,
  xtick={1,1.2,1.4,1.6,1.8}, tick label style={font=\scriptsize},
  yticklabel style={/pgf/number format/fixed,/pgf/number format/precision=2}]
\addplot[blue,thick,mark=*,mark size=1.1pt] coordinates {(1.000,0.7320)(1.400,0.7320)(1.600,0.7320)(1.700,0.7320)(1.750,0.7320)(1.800,0.7306)(1.850,0.7278)(1.900,0.7032)(1.925,0.6837)(1.950,0.6573)};
\addplot[red,thick,dashed,mark=square,mark size=1.1pt] coordinates {(1.000,0.7320)(1.400,0.7320)(1.600,0.7320)(1.700,0.7320)(1.750,0.7069)(1.800,0.6123)(1.850,0.5144)(1.900,0.4082)(1.925,0.3484)(1.950,0.2806)};
\addplot[black,thick,densely dotted] coordinates {(1.000,0.7320)(1.400,0.7320)(1.600,0.7320)(1.700,0.7320)(1.750,0.7306)(1.800,0.7195)(1.850,0.7093)(1.900,0.6879)(1.925,0.6689)(1.950,0.6425)};
\addplot[blue,only marks,mark=*,mark size=1.7pt] coordinates {(1.9,0.7559)(1.95,0.7523)};
\addplot[red,only marks,mark=square,mark size=1.7pt] coordinates {(1.9,0.5344)(1.95,0.4395)};
\addplot[black,only marks,mark=triangle,mark size=1.9pt] coordinates {(1.9,0.7207)(1.95,0.7102)};
\end{axis}
\end{tikzpicture}
\caption{Linear stability limits on D3Q27 for the standard
multiple-relaxation-time matrix, with shear rate $\omega$ and the trace and every
higher moment at unity ($16^3$ wavevectors): orthogonal (blue solid), non-orthogonal
(red dashed), and the common limit of the two when every even non-conserved moment
relaxes at $\omega$ (black dotted). Symbols: nonlinear simulations of an
advected Taylor--Green vortex ($24^3$, $600$ steps) with the orthogonal (filled
circles) and non-orthogonal (open squares) formulations and the recipe
(triangles). Note the different vertical scales. The two
formulations coincide for $\omega\le1.7$ in every direction. Beyond, the
non-orthogonal one is marginally the more stable along the axis, by at most
$1.8\%$, and collapses along both diagonals, where at $\omega=1.95$ the orthogonal
limit exceeds it by $134\%$.}
\label{fig:d3stab}
\end{figure}

The recipe here relaxes every even non-conserved moment at the shear rate, which
ties the bulk viscosity to it. It matches the orthogonal linear limit to $0.1\%$, but
in an under-resolved Taylor--Green vortex ($32^3$, $U_0=0.1$) it diverges at every
$\mathrm{Re}\ge1600$, whereas both standard formulations survive up to
$\mathrm{Re}=25600$: holding the fourth- and sixth-order moments at unity supplies
damping that under-resolved flow needs and the uniform-state analysis cannot see.

The formulations also differ in accuracy (Fig.~\ref{fig:tgv}). At
$\mathrm{Re}=1600$ and $\mathrm{Ma}=0.1$, against the $512^3$ direct numerical
simulation of Ref.~\cite{dairay2017}, the orthogonal formulation is the closer on
every lattice: the largest departure of the dissipation rate, relative to its peak,
is $55\%$, $33\%$ and $17\%$ on $32^3$, $48^3$ and $64^3$, against $70\%$, $62\%$ and
$41\%$, and its dissipation peak reaches $t=8.5\,t_c$ on $48^3$, against $8.98\,t_c$
in the reference, while the non-orthogonal one is still at $6.7\,t_c$. At
$\mathrm{Re}=400$ the histories converge to one curve, with the same ranking at
coarse resolution. The ranking matches the effective viscosity at finite
wavenumber, $2.1\%$ above $c_s^2(1/\omega-\tfrac12)$ in the non-orthogonal basis and
$0.5\%$ in the orthogonal one at $\omega=1.9$ on $32^3$, and the mechanism is again
that of Eq.~\eqref{eq:q9}: the raw fourth-order moments contain second-order parts,
which relaxing them at unity damps at the wrong rate.

\begin{figure}[tbp]
\centering
\begin{tikzpicture}
\begin{axis}[width=0.48\textwidth,height=5.2cm,title={(a) $\mathrm{Re}=400$, $32^3$},
  xlabel={$t/t_c$}, ylabel={$E/E_0$}, grid=major, xmin=0,xmax=20, ymin=0,ymax=1.02,
  tick label style={font=\scriptsize}]
\addplot[black!45,line width=2.2pt] coordinates {(0.00,1.0000)(0.25,0.9943)(0.50,0.9904)(0.75,0.9867)(1.00,0.9828)(1.25,0.9788)(1.50,0.9744)(1.75,0.9697)(2.00,0.9646)(2.25,0.9590)(2.50,0.9527)(2.75,0.9459)(3.00,0.9382)(3.25,0.9298)(3.50,0.9205)(3.75,0.9102)(4.00,0.8987)(4.25,0.8859)(4.50,0.8716)(4.75,0.8556)(5.00,0.8381)(5.25,0.8189)(5.50,0.7987)(5.75,0.7776)(6.00,0.7561)(6.25,0.7344)(6.50,0.7125)(6.75,0.6905)(7.00,0.6686)(7.25,0.6467)(7.50,0.6250)(7.75,0.6034)(8.00,0.5821)(8.25,0.5609)(8.50,0.5396)(8.75,0.5181)(9.00,0.4964)(9.25,0.4751)(9.50,0.4545)(9.75,0.4345)(10.00,0.4155)(10.25,0.3974)(10.50,0.3803)(10.75,0.3639)(11.00,0.3482)(11.25,0.3332)(11.50,0.3188)(11.75,0.3052)(12.00,0.2923)(12.25,0.2803)(12.50,0.2692)(12.75,0.2590)(13.00,0.2496)(13.25,0.2408)(13.50,0.2327)(13.75,0.2250)(14.00,0.2177)(14.25,0.2109)(14.50,0.2044)(14.75,0.1982)(15.00,0.1924)(15.25,0.1869)(15.50,0.1816)(15.75,0.1765)(16.00,0.1717)(16.25,0.1671)(16.50,0.1627)(16.75,0.1585)(17.00,0.1544)(17.25,0.1505)(17.50,0.1468)(17.75,0.1432)(18.00,0.1398)(18.25,0.1364)(18.50,0.1332)(18.75,0.1302)(19.00,0.1272)(19.25,0.1243)(19.50,0.1216)(19.75,0.1189)(20.00,0.1170)};
\addplot[blue,thick] coordinates {(0.00,1.0000)(0.25,0.9774)(0.50,0.9681)(0.75,0.9700)(1.00,0.9665)(1.25,0.9609)(1.50,0.9563)(1.75,0.9513)(2.00,0.9456)(2.25,0.9394)(2.50,0.9323)(2.75,0.9244)(3.00,0.9158)(3.25,0.9061)(3.50,0.8953)(3.75,0.8834)(4.00,0.8704)(4.25,0.8561)(4.50,0.8404)(4.75,0.8231)(5.00,0.8041)(5.25,0.7834)(5.50,0.7612)(5.75,0.7375)(6.00,0.7129)(6.25,0.6876)(6.50,0.6619)(6.75,0.6361)(7.00,0.6105)(7.25,0.5853)(7.50,0.5606)(7.75,0.5366)(8.00,0.5134)(8.25,0.4913)(8.50,0.4703)(8.75,0.4504)(9.00,0.4316)(9.25,0.4139)(9.50,0.3972)(9.75,0.3813)(10.00,0.3664)(10.25,0.3522)(10.50,0.3388)(10.75,0.3261)(11.00,0.3140)(11.25,0.3026)(11.50,0.2917)(11.75,0.2815)(12.00,0.2718)(12.25,0.2628)(12.50,0.2543)(12.75,0.2463)(13.00,0.2388)(13.25,0.2317)(13.50,0.2251)(13.75,0.2189)(14.00,0.2131)(14.25,0.2076)(14.50,0.2024)(14.75,0.1975)(15.00,0.1929)(15.25,0.1885)(15.50,0.1843)(15.75,0.1804)(16.00,0.1767)(16.25,0.1733)(16.50,0.1700)(16.75,0.1669)(17.00,0.1640)(17.25,0.1613)(17.50,0.1587)(17.75,0.1563)(18.00,0.1540)(18.25,0.1519)(18.50,0.1499)(18.75,0.1480)(19.00,0.1462)(19.25,0.1445)(19.50,0.1429)(19.75,0.1414)(20.00,0.1404)};
\addplot[red,thick,dashed] coordinates {(0.00,1.0000)(0.25,0.9768)(0.50,0.9664)(0.75,0.9688)(1.00,0.9650)(1.25,0.9584)(1.50,0.9533)(1.75,0.9478)(2.00,0.9415)(2.25,0.9344)(2.50,0.9264)(2.75,0.9175)(3.00,0.9077)(3.25,0.8967)(3.50,0.8846)(3.75,0.8713)(4.00,0.8568)(4.25,0.8409)(4.50,0.8236)(4.75,0.8046)(5.00,0.7838)(5.25,0.7612)(5.50,0.7371)(5.75,0.7115)(6.00,0.6849)(6.25,0.6575)(6.50,0.6299)(6.75,0.6022)(7.00,0.5748)(7.25,0.5480)(7.50,0.5219)(7.75,0.4968)(8.00,0.4727)(8.25,0.4498)(8.50,0.4281)(8.75,0.4077)(9.00,0.3885)(9.25,0.3705)(9.50,0.3537)(9.75,0.3380)(10.00,0.3234)(10.25,0.3097)(10.50,0.2970)(10.75,0.2851)(11.00,0.2739)(11.25,0.2635)(11.50,0.2537)(11.75,0.2446)(12.00,0.2360)(12.25,0.2279)(12.50,0.2203)(12.75,0.2132)(13.00,0.2065)(13.25,0.2002)(13.50,0.1943)(13.75,0.1888)(14.00,0.1835)(14.25,0.1786)(14.50,0.1739)(14.75,0.1695)(15.00,0.1653)(15.25,0.1614)(15.50,0.1576)(15.75,0.1540)(16.00,0.1506)(16.25,0.1473)(16.50,0.1441)(16.75,0.1411)(17.00,0.1382)(17.25,0.1354)(17.50,0.1327)(17.75,0.1301)(18.00,0.1276)(18.25,0.1251)(18.50,0.1228)(18.75,0.1205)(19.00,0.1183)(19.25,0.1162)(19.50,0.1142)(19.75,0.1122)(20.00,0.1108)};
\addplot[black,thick,densely dotted] coordinates {(0.00,1.0000)(0.25,0.9773)(0.50,0.9675)(0.75,0.9704)(1.00,0.9672)(1.25,0.9606)(1.50,0.9561)(1.75,0.9518)(2.00,0.9461)(2.25,0.9396)(2.50,0.9331)(2.75,0.9259)(3.00,0.9177)(3.25,0.9086)(3.50,0.8992)(3.75,0.8886)(4.00,0.8768)(4.25,0.8641)(4.50,0.8502)(4.75,0.8342)(5.00,0.8161)(5.25,0.7965)(5.50,0.7750)(5.75,0.7519)(6.00,0.7281)(6.25,0.7040)(6.50,0.6797)(6.75,0.6554)(7.00,0.6312)(7.25,0.6073)(7.50,0.5838)(7.75,0.5609)(8.00,0.5395)(8.25,0.5193)(8.50,0.5003)(8.75,0.4826)(9.00,0.4664)(9.25,0.4507)(9.50,0.4356)(9.75,0.4213)(10.00,0.4077)(10.25,0.3950)(10.50,0.3832)(10.75,0.3726)(11.00,0.3626)(11.25,0.3531)(11.50,0.3435)(11.75,0.3336)(12.00,0.3234)(12.25,0.3130)(12.50,0.3022)(12.75,0.2913)(13.00,0.2804)(13.25,0.2696)(13.50,0.2588)(13.75,0.2483)(14.00,0.2380)(14.25,0.2281)(14.50,0.2189)(14.75,0.2101)(15.00,0.2019)(15.25,0.1944)(15.50,0.1875)(15.75,0.1811)(16.00,0.1751)(16.25,0.1697)(16.50,0.1646)(16.75,0.1599)(17.00,0.1554)(17.25,0.1513)(17.50,0.1474)(17.75,0.1437)(18.00,0.1402)(18.25,0.1369)(18.50,0.1338)(18.75,0.1308)(19.00,0.1280)(19.25,0.1253)(19.50,0.1228)(19.75,0.1204)(20.00,0.1187)};
\end{axis}
\end{tikzpicture}\hfill
\begin{tikzpicture}
\begin{axis}[width=0.48\textwidth,height=5.2cm,title={(b) $\mathrm{Re}=400$, departure from $96^3$},
  xlabel={$N$}, ylabel={$\max_t|E-E_{\mathrm{ref}}|/E_0$}, grid=major, xmode=log, ymode=log,
  xtick={32,48,64}, xticklabels={$32^3$,$48^3$,$64^3$},
  tick label style={font=\scriptsize}]
\addplot[blue,thick,mark=*,mark size=1.6pt] coordinates {(32,0.06966)(48,0.02886)(64,0.00983)};
\addplot[red,thick,dashed,mark=square,mark size=1.6pt] coordinates {(32,0.11152)(48,0.06543)(64,0.03485)};
\addplot[black,thick,densely dotted,mark=triangle,mark size=1.8pt] coordinates {(32,0.04284)(48,0.01582)(64,0.00450)};
\end{axis}
\end{tikzpicture}

\vspace{0.6em}
\begin{tikzpicture}
\begin{axis}[width=0.48\textwidth,height=5.2cm,title={(c) $\mathrm{Re}=1600$, $64^3$},
  xlabel={$t/t_c$}, ylabel={$\varepsilon\,t_c/U_0^2$}, grid=major, xmin=0,xmax=15, ymin=0,ymax=0.016,
  scaled y ticks=false,
  tick label style={font=\scriptsize}, yticklabel style={/pgf/number format/fixed,/pgf/number format/precision=3}]
\addplot[black!45,line width=2.2pt] coordinates {(0.00,0.00047)(0.10,0.00047)(0.20,0.00047)(0.30,0.00047)(0.40,0.00048)(0.50,0.00048)(0.60,0.00049)(0.70,0.00049)(0.80,0.00050)(0.90,0.00051)(1.00,0.00052)(1.10,0.00053)(1.20,0.00054)(1.30,0.00056)(1.40,0.00057)(1.50,0.00059)(1.60,0.00061)(1.70,0.00063)(1.80,0.00065)(1.90,0.00068)(2.00,0.00071)(2.10,0.00074)(2.20,0.00077)(2.30,0.00080)(2.40,0.00084)(2.50,0.00088)(2.60,0.00092)(2.70,0.00097)(2.80,0.00102)(2.90,0.00107)(3.00,0.00113)(3.10,0.00119)(3.20,0.00125)(3.30,0.00132)(3.40,0.00139)(3.50,0.00148)(3.60,0.00157)(3.70,0.00167)(3.80,0.00179)(3.90,0.00192)(4.00,0.00207)(4.10,0.00223)(4.20,0.00242)(4.30,0.00262)(4.40,0.00284)(4.50,0.00306)(4.60,0.00329)(4.70,0.00351)(4.80,0.00373)(4.90,0.00393)(5.00,0.00413)(5.10,0.00431)(5.20,0.00447)(5.30,0.00462)(5.40,0.00475)(5.50,0.00487)(5.60,0.00499)(5.70,0.00511)(5.80,0.00523)(5.90,0.00537)(6.00,0.00553)(6.10,0.00571)(6.20,0.00591)(6.30,0.00613)(6.40,0.00637)(6.50,0.00663)(6.60,0.00688)(6.70,0.00708)(6.80,0.00717)(6.90,0.00725)(7.00,0.00736)(7.10,0.00754)(7.20,0.00778)(7.30,0.00807)(7.40,0.00840)(7.50,0.00873)(7.60,0.00903)(7.70,0.00932)(7.80,0.00963)(7.90,0.00999)(8.00,0.01037)(8.10,0.01074)(8.20,0.01106)(8.30,0.01134)(8.40,0.01160)(8.50,0.01192)(8.60,0.01232)(8.70,0.01263)(8.80,0.01278)(8.90,0.01284)(9.00,0.01285)(9.10,0.01276)(9.20,0.01254)(9.30,0.01221)(9.40,0.01190)(9.50,0.01173)(9.60,0.01163)(9.70,0.01155)(9.80,0.01146)(9.90,0.01137)(10.00,0.01127)(10.10,0.01122)(10.20,0.01118)(10.30,0.01113)(10.40,0.01103)(10.50,0.01089)(10.60,0.01073)(10.70,0.01057)(10.80,0.01042)(10.90,0.01026)(11.00,0.01009)(11.10,0.00992)(11.20,0.00974)(11.30,0.00956)(11.40,0.00938)(11.50,0.00921)(11.60,0.00904)(11.70,0.00888)(11.80,0.00872)(11.90,0.00855)(12.00,0.00836)(12.10,0.00816)(12.20,0.00793)(12.30,0.00769)(12.40,0.00743)(12.50,0.00718)(12.60,0.00695)(12.70,0.00676)(12.80,0.00659)(12.90,0.00645)(13.00,0.00634)(13.10,0.00624)(13.20,0.00615)(13.30,0.00606)(13.40,0.00597)(13.50,0.00588)(13.60,0.00579)(13.70,0.00570)(13.80,0.00561)(13.90,0.00552)(14.00,0.00543)(14.10,0.00533)(14.20,0.00522)(14.30,0.00511)(14.40,0.00499)(14.50,0.00488)(14.60,0.00477)(14.70,0.00465)(14.80,0.00455)(14.90,0.00444)};
\addplot[blue,thick] coordinates {(0.50,0.00057)(0.60,0.00080)(0.70,0.00041)(0.80,0.00032)(0.90,0.00059)(1.00,0.00063)(1.10,0.00049)(1.20,0.00051)(1.30,0.00062)(1.40,0.00063)(1.50,0.00061)(1.60,0.00064)(1.70,0.00070)(1.80,0.00072)(1.90,0.00074)(2.00,0.00078)(2.10,0.00082)(2.20,0.00086)(2.30,0.00090)(2.40,0.00095)(2.50,0.00100)(2.60,0.00105)(2.70,0.00112)(2.80,0.00118)(2.90,0.00124)(3.00,0.00132)(3.10,0.00140)(3.20,0.00149)(3.30,0.00159)(3.40,0.00172)(3.50,0.00185)(3.60,0.00200)(3.70,0.00217)(3.80,0.00236)(3.90,0.00257)(4.00,0.00280)(4.10,0.00305)(4.20,0.00333)(4.30,0.00364)(4.40,0.00398)(4.50,0.00434)(4.60,0.00472)(4.70,0.00509)(4.80,0.00545)(4.90,0.00578)(5.00,0.00608)(5.10,0.00635)(5.20,0.00657)(5.30,0.00676)(5.40,0.00689)(5.50,0.00698)(5.60,0.00703)(5.70,0.00705)(5.80,0.00706)(5.90,0.00707)(6.00,0.00710)(6.10,0.00717)(6.20,0.00729)(6.30,0.00744)(6.40,0.00764)(6.50,0.00787)(6.60,0.00813)(6.70,0.00842)(6.80,0.00871)(6.90,0.00899)(7.00,0.00923)(7.10,0.00941)(7.20,0.00953)(7.30,0.00962)(7.40,0.00969)(7.50,0.00978)(7.60,0.00992)(7.70,0.01013)(7.80,0.01043)(7.90,0.01081)(8.00,0.01128)(8.10,0.01175)(8.20,0.01211)(8.30,0.01229)(8.40,0.01227)(8.50,0.01212)(8.60,0.01193)(8.70,0.01175)(8.80,0.01158)(8.90,0.01144)(9.00,0.01133)(9.10,0.01124)(9.20,0.01115)(9.30,0.01105)(9.40,0.01091)(9.50,0.01075)(9.60,0.01058)(9.70,0.01041)(9.80,0.01025)(9.90,0.01011)(10.00,0.01000)(10.10,0.00990)(10.20,0.00983)(10.30,0.00977)(10.40,0.00971)(10.50,0.00964)(10.60,0.00954)(10.70,0.00939)(10.80,0.00920)(10.90,0.00898)(11.00,0.00876)(11.10,0.00856)(11.20,0.00838)(11.30,0.00822)(11.40,0.00808)(11.50,0.00795)(11.60,0.00782)(11.70,0.00770)(11.80,0.00757)(11.90,0.00743)(12.00,0.00727)(12.10,0.00710)(12.20,0.00691)(12.30,0.00671)(12.40,0.00650)(12.50,0.00629)(12.60,0.00609)(12.70,0.00590)(12.80,0.00572)(12.90,0.00555)(13.00,0.00539)(13.10,0.00525)(13.20,0.00513)(13.30,0.00504)(13.40,0.00496)(13.50,0.00490)(13.60,0.00485)(13.70,0.00482)(13.80,0.00478)(13.90,0.00474)(14.00,0.00470)(14.10,0.00466)(14.20,0.00461)(14.30,0.00456)(14.40,0.00451)(14.50,0.00445)(14.60,0.00440)(14.70,0.00435)(14.80,0.00431)(14.90,0.00426)};
\addplot[red,thick,dashed] coordinates {(0.50,0.00063)(0.60,0.00087)(0.70,0.00046)(0.80,0.00036)(0.90,0.00066)(1.00,0.00070)(1.10,0.00054)(1.20,0.00056)(1.30,0.00070)(1.40,0.00071)(1.50,0.00067)(1.60,0.00072)(1.70,0.00079)(1.80,0.00082)(1.90,0.00084)(2.00,0.00089)(2.10,0.00096)(2.20,0.00100)(2.30,0.00106)(2.40,0.00113)(2.50,0.00119)(2.60,0.00126)(2.70,0.00135)(2.80,0.00143)(2.90,0.00152)(3.00,0.00162)(3.10,0.00173)(3.20,0.00184)(3.30,0.00197)(3.40,0.00212)(3.50,0.00229)(3.60,0.00247)(3.70,0.00268)(3.80,0.00291)(3.90,0.00317)(4.00,0.00346)(4.10,0.00379)(4.20,0.00415)(4.30,0.00455)(4.40,0.00499)(4.50,0.00546)(4.60,0.00596)(4.70,0.00647)(4.80,0.00698)(4.90,0.00748)(5.00,0.00797)(5.10,0.00843)(5.20,0.00887)(5.30,0.00928)(5.40,0.00965)(5.50,0.00996)(5.60,0.01022)(5.70,0.01042)(5.80,0.01057)(5.90,0.01069)(6.00,0.01080)(6.10,0.01091)(6.20,0.01104)(6.30,0.01117)(6.40,0.01131)(6.50,0.01144)(6.60,0.01156)(6.70,0.01165)(6.80,0.01171)(6.90,0.01174)(7.00,0.01174)(7.10,0.01171)(7.20,0.01166)(7.30,0.01161)(7.40,0.01159)(7.50,0.01159)(7.60,0.01164)(7.70,0.01173)(7.80,0.01186)(7.90,0.01203)(8.00,0.01222)(8.10,0.01242)(8.20,0.01261)(8.30,0.01278)(8.40,0.01289)(8.50,0.01293)(8.60,0.01288)(8.70,0.01272)(8.80,0.01248)(8.90,0.01218)(9.00,0.01186)(9.10,0.01153)(9.20,0.01122)(9.30,0.01094)(9.40,0.01070)(9.50,0.01051)(9.60,0.01036)(9.70,0.01025)(9.80,0.01018)(9.90,0.01013)(10.00,0.01008)(10.10,0.01002)(10.20,0.00994)(10.30,0.00984)(10.40,0.00972)(10.50,0.00958)(10.60,0.00941)(10.70,0.00922)(10.80,0.00900)(10.90,0.00875)(11.00,0.00847)(11.10,0.00816)(11.20,0.00785)(11.30,0.00753)(11.40,0.00722)(11.50,0.00693)(11.60,0.00667)(11.70,0.00643)(11.80,0.00623)(11.90,0.00604)(12.00,0.00588)(12.10,0.00574)(12.20,0.00561)(12.30,0.00549)(12.40,0.00538)(12.50,0.00527)(12.60,0.00517)(12.70,0.00507)(12.80,0.00497)(12.90,0.00487)(13.00,0.00477)(13.10,0.00468)(13.20,0.00458)(13.30,0.00450)(13.40,0.00441)(13.50,0.00434)(13.60,0.00427)(13.70,0.00420)(13.80,0.00414)(13.90,0.00409)(14.00,0.00404)(14.10,0.00400)(14.20,0.00396)(14.30,0.00392)(14.40,0.00389)(14.50,0.00385)(14.60,0.00382)(14.70,0.00379)(14.80,0.00377)(14.90,0.00374)};
\end{axis}
\end{tikzpicture}\hfill
\begin{tikzpicture}
\begin{axis}[width=0.48\textwidth,height=5.2cm,title={(d) $\mathrm{Re}=1600$, departure from DNS},
  xlabel={$N$}, ylabel={$\max_t|\varepsilon-\varepsilon_{\mathrm{DNS}}|/\max\varepsilon_{\mathrm{DNS}}$}, grid=major, xmode=log, ymin=0, ymax=0.8,
  xtick={32,48,64}, xticklabels={$32^3$,$48^3$,$64^3$},
  tick label style={font=\scriptsize}]
\addplot[blue,thick,mark=*,mark size=1.6pt] coordinates {(32,0.5540)(48,0.3336)(64,0.1670)};
\addplot[red,thick,dashed,mark=square,mark size=1.6pt] coordinates {(32,0.7007)(48,0.6174)(64,0.4150)};
\addplot[blue,mark=diamond*,mark size=2pt] coordinates {(32,0.4627)(48,0.3299)(64,0.1385)};
\addplot[red,dashed,mark=diamond,mark size=2pt] coordinates {(32,0.7026)(48,0.5546)(64,0.4149)};
\end{axis}
\end{tikzpicture}
\caption{Taylor--Green vortex on D3Q27. Top: $\mathrm{Re}=400$, $U_0=0.1$; (a) kinetic
energy on a $32^3$ lattice for the orthogonal (blue solid) and non-orthogonal (red
dashed) formulations with the standard matrix and for the recipe (black dotted),
against the orthogonal result on a $96^3$ lattice (grey); (b) largest
departure from that reference over twenty convective times. Bottom:
$\mathrm{Re}=1600$, $\mathrm{Ma}=0.1$; (c) kinetic-energy dissipation rate on a
$64^3$ lattice against the $512^3$ DNS of Ref.~\cite{dairay2017} (grey);
(d) largest departure of the dissipation rate from the DNS for
$t_c\leq t\leq15t_c$, relative to its peak, on D3Q27 (thick lines) and on D3Q19 (thin
lines with diamonds, Section~\ref{sec:d19stab}); the start-up transient of the
equilibrium initialisation is excluded, and not shown in (c). The recipe diverges at $\mathrm{Re}=1600$ on every lattice
tested and is therefore absent from (c) and (d).}
\label{fig:tgv}
\end{figure}
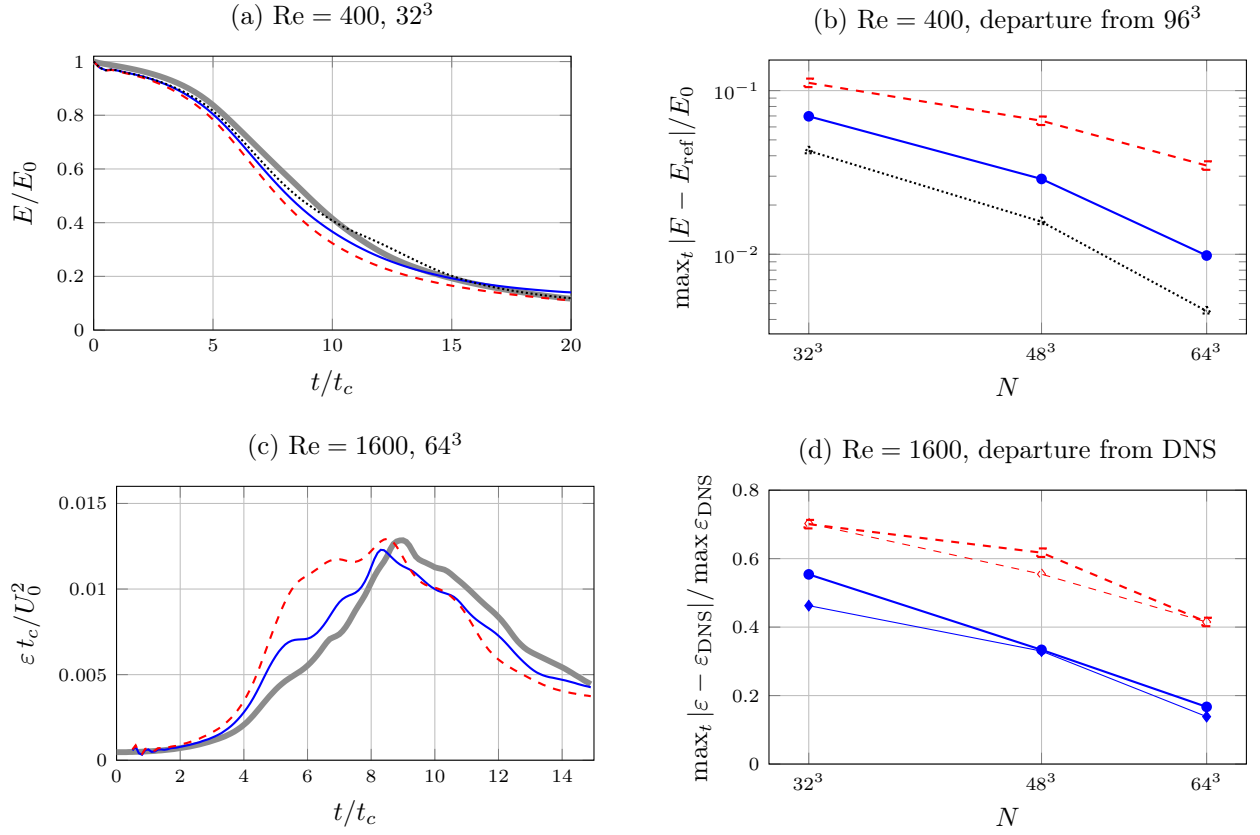

\subsection{The D3Q19 lattice}
\label{sec:d19stab}

Theorem~\ref{thm:geom} predicts a narrower coupling on D3Q19, the most widely used
three-dimensional velocity set, on which central-moment schemes are also in
use~\cite{derosis2020}: only the two normal-stress differences are joined to the fourth
order, through the face diagonals, and the off-diagonal shear moments stay free
(Fig.~\ref{fig:graphs}d). We take the D3Q19 restriction of the basis of
Ref.~\cite{derosis2017}, the nineteen moments $1$, $C_\alpha$, the deviatoric block, the
trace, $C_\alpha C_\beta^2$ and $C_\alpha^2C_\beta^2$ with $\alpha\ne\beta$, with the
Maxwellian central moments as equilibria and the standard matrix of
Section~\ref{sec:d3stab}. Uniform wavevector grids converge slowly on this lattice, because
the modes that limit it occupy narrow bands, so the D3Q19 limits are computed with a
local maximisation of the spectral radius, validated against uniform $48^3$ and $64^3$
grids and against the D3Q27 limits of Fig.~\ref{fig:d3stab} (\ref{sec:wavevector}).

The two formulations coincide at $\omega=1$ and differ at every other rate, first along
the face diagonal (Fig.~\ref{fig:d19stab}). The least stable direction is the body
diagonal, which the lattice does not contain, and there the two limits agree to within
$0.5\%$ at every rate, and to within the bisection tolerance from
$\omega=1.8$ on. The reason is again Theorem~\ref{thm:geom}: the
mode that limits the body diagonal is carried by the momentum, the density, the
off-diagonal shear and the fourth-order moments, with less than 0.02\% of its
content in the normal-stress differences, the only non-conserved moments through which
the commutator acts for this matrix, so the two formulations act on it alike. Along the
axis and the face diagonal the unstable mode of the non-orthogonal formulation holds
35\% and 9\% of its content in the normal-stress differences at $\omega=1.95$, and
there the orthogonal formulation leads, by up to 10\% and 6\%. The recipe, whose two
formulations coincide, is the least stable of the three from $\omega=1.4$ on. In the
advected Taylor--Green vortex the orthogonal formulation leads in the least stable
direction by 3\% at $\omega=1.9$ and 10\% at $1.95$ on $24^3$. These margins fall
to $1\%$ and $3\%$ on $48^3$ and to $-1\%$ at both rates at
three times the run length, towards the linear margin of zero, as the body-diagonal
instability, which the short coarse runs do not reach, takes over
(Table~\ref{tab:d19conv}).

The accuracy follows D3Q27. In the under-resolved Taylor--Green vortex ($32^3$,
$U_0=0.1$) both standard formulations survive up to $\mathrm{Re}=25600$, whereas the
recipe diverges at every Reynolds number tested, from $400$. At $\mathrm{Re}=1600$ the
orthogonal formulation is the closer to the direct numerical simulation on every
lattice (Fig.~\ref{fig:tgv}d), with largest departures of the dissipation rate of
$46\%$, $33\%$ and $14\%$ on $32^3$, $48^3$ and $64^3$ against $70\%$, $55\%$ and $41\%$, and its
dissipation peak on $64^3$ reaches $t=8.8\,t_c$ against $8.3\,t_c$ for the
non-orthogonal one and $8.98\,t_c$ in the reference. A small-amplitude vortex at
$\omega=1.9$ on $32^3$ decays with a viscosity $0.9\%$ above
$c_s^2(1/\omega-\tfrac12)$ in the orthogonal basis and $5.7\%$ in the
non-orthogonal one, against $0.5\%$ and $2.1\%$ on D3Q27, by the mechanism of
Eq.~\eqref{eq:q9}. On D3Q19 the basis thus changes the accuracy as much as on D3Q27,
but the robustness only in the directions in which the normal stresses take part.

\begin{figure}[tbp]
\centering
\begin{tikzpicture}
\begin{axis}[width=0.355\textwidth,height=5cm,title={axis $(1,0,0)$},
  xlabel={$\omega$}, ylabel={$u_{\max}$}, grid=major, xmin=0.95,xmax=1.98, ymin=0.140,ymax=0.454,
  xtick={1,1.2,1.4,1.6,1.8}, tick label style={font=\scriptsize},
  yticklabel style={/pgf/number format/fixed,/pgf/number format/precision=2}]
\addplot[blue,thick,mark=*,mark size=1.1pt] coordinates {(1.000,0.4225)(1.400,0.4225)(1.600,0.4225)(1.700,0.4225)(1.750,0.4213)(1.800,0.4188)(1.850,0.4156)(1.900,0.4109)(1.925,0.4084)(1.950,0.3994)};
\addplot[red,thick,dashed,mark=square,mark size=1.1pt] coordinates {(1.000,0.4225)(1.400,0.4225)(1.600,0.4225)(1.700,0.4225)(1.750,0.4209)(1.800,0.4181)(1.850,0.4128)(1.900,0.3856)(1.925,0.3734)(1.950,0.3619)};
\addplot[black,thick,densely dotted] coordinates {(1.000,0.4225)(1.400,0.4225)(1.600,0.4225)(1.700,0.4225)(1.750,0.4225)(1.800,0.4106)(1.850,0.3575)(1.900,0.3047)(1.925,0.2328)(1.950,0.1613)};
\addplot[blue,only marks,mark=*,mark size=1.7pt] coordinates {(1.9,0.4324)(1.95,0.4324)};
\addplot[red,only marks,mark=square,mark size=1.7pt] coordinates {(1.9,0.4184)(1.95,0.3938)};
\addplot[black,only marks,mark=triangle,mark size=1.9pt] coordinates {(1.9,0.3902)(1.95,0.3516)};
\end{axis}
\end{tikzpicture}\hfill
\begin{tikzpicture}
\begin{axis}[width=0.355\textwidth,height=5cm,title={face diagonal $(1,1,0)$},
  xlabel={$\omega$}, grid=major, xmin=0.95,xmax=1.98, ymin=0.142,ymax=0.582,
  xtick={1,1.2,1.4,1.6,1.8}, tick label style={font=\scriptsize},
  yticklabel style={/pgf/number format/fixed,/pgf/number format/precision=2}]
\addplot[blue,thick,mark=*,mark size=1.1pt] coordinates {(1.000,0.4909)(1.400,0.5072)(1.600,0.5156)(1.700,0.5184)(1.750,0.5191)(1.800,0.4878)(1.850,0.4438)(1.900,0.4050)(1.925,0.3844)(1.950,0.3138)};
\addplot[red,thick,dashed,mark=square,mark size=1.1pt] coordinates {(1.000,0.4909)(1.400,0.4991)(1.600,0.5006)(1.700,0.5003)(1.750,0.5000)(1.800,0.4713)(1.850,0.4216)(1.900,0.3825)(1.925,0.3653)(1.950,0.2988)};
\addplot[black,thick,densely dotted] coordinates {(1.000,0.4909)(1.400,0.4609)(1.600,0.4281)(1.700,0.4019)(1.750,0.3834)(1.800,0.3594)(1.850,0.3250)(1.900,0.2719)(1.925,0.2288)(1.950,0.1722)};
\addplot[blue,only marks,mark=*,mark size=1.7pt] coordinates {(1.9,0.5520)(1.95,0.5379)};
\addplot[red,only marks,mark=square,mark size=1.7pt] coordinates {(1.9,0.4359)(1.95,0.3973)};
\addplot[black,only marks,mark=triangle,mark size=1.9pt] coordinates {(1.9,0.4254)(1.95,0.3938)};
\end{axis}
\end{tikzpicture}\hfill
\begin{tikzpicture}
\begin{axis}[width=0.355\textwidth,height=5cm,title={body diagonal $(1,1,1)$},
  xlabel={$\omega$}, grid=major, xmin=0.95,xmax=1.98, ymin=0.128,ymax=0.474,
  xtick={1,1.2,1.4,1.6,1.8}, tick label style={font=\scriptsize},
  yticklabel style={/pgf/number format/fixed,/pgf/number format/precision=2}]
\addplot[blue,thick,mark=*,mark size=1.1pt] coordinates {(1.000,0.4166)(1.400,0.4256)(1.600,0.4297)(1.700,0.4034)(1.750,0.3666)(1.800,0.3259)(1.850,0.2800)(1.900,0.2269)(1.925,0.1956)(1.950,0.1591)};
\addplot[red,thick,dashed,mark=square,mark size=1.1pt] coordinates {(1.000,0.4166)(1.400,0.4256)(1.600,0.4297)(1.700,0.4056)(1.750,0.3672)(1.800,0.3259)(1.850,0.2800)(1.900,0.2269)(1.925,0.1956)(1.950,0.1594)};
\addplot[black,thick,densely dotted] coordinates {(1.000,0.4166)(1.400,0.3925)(1.600,0.3656)(1.700,0.3441)(1.750,0.3291)(1.800,0.3084)(1.850,0.2778)(1.900,0.2263)(1.925,0.1944)(1.950,0.1516)};
\addplot[blue,only marks,mark=*,mark size=1.7pt] coordinates {(1.9,0.4500)(1.95,0.4359)};
\addplot[red,only marks,mark=square,mark size=1.7pt] coordinates {(1.9,0.4500)(1.95,0.4359)};
\addplot[black,only marks,mark=triangle,mark size=1.9pt] coordinates {(1.9,0.3973)(1.95,0.3832)};
\end{axis}
\end{tikzpicture}
\caption{Linear stability limits on D3Q19 for the standard
multiple-relaxation-time matrix, with shear rate $\omega$ and the trace and every
higher moment at unity (refined wavevector search, Section~\ref{sec:d19stab}):
orthogonal (blue solid), non-orthogonal (red dashed), and the common limit of the two
when every even non-conserved moment relaxes at $\omega$ (black dotted). Symbols:
nonlinear simulations of an advected Taylor--Green vortex ($24^3$, $600$ steps), as in
Fig.~\ref{fig:d3stab}. Note the different vertical scales. Along the body diagonal,
the least stable direction, the two formulations have the same limit to within
$0.5\%$; along the axis and the face diagonal the orthogonal
formulation leads.}
\label{fig:d19stab}
\end{figure}
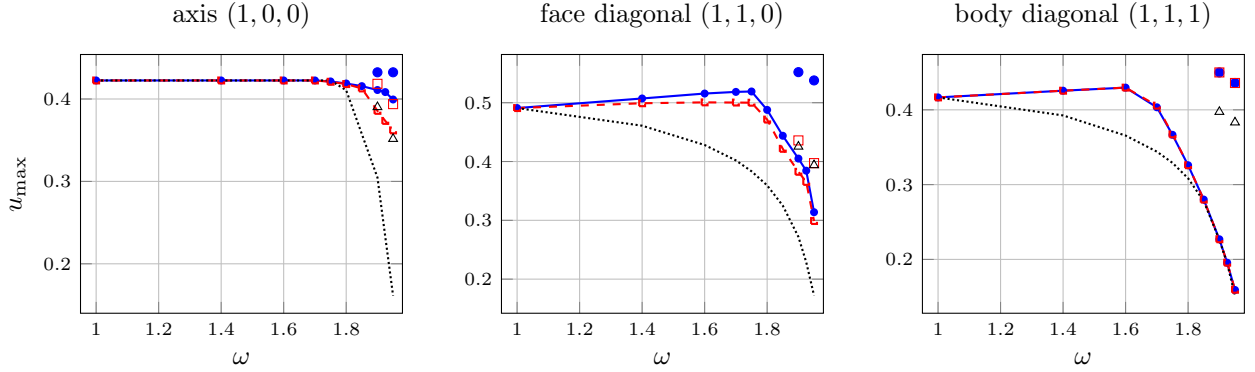

\subsection{The inner product}
\label{sec:ipsummary}

The orthogonal bases above are orthogonal in the lattice-weight inner product, which
on D2Q9, D3Q19 and D3Q27 is that of the rest equilibrium, so that
Proposition~\ref{prop:rest} applies to them. \ref{sec:ip} repeats the
comparison for a one-parameter family of inner products. The consistency result of
Section~\ref{sec:rectstab} holds for all of them, by Eq.~\eqref{eq:A54-var}. The
stability advantage persists near the inner product of the rest equilibrium, where
the optimum is broad, and is lost far from it: with the unweighted inner product the
orthogonal scheme is unstable even at rest on all three lattices of \ref{sec:ip} at
large rates. What stabilises is therefore orthogonality in, or near, the inner product
of the rest equilibrium, not orthogonality as such; on rectangular lattices with a free
sound speed that is not the lattice-weight inner product. The rectangular result of
Section~\ref{sec:rectstab} makes the same point from the other side: the published
non-orthogonal scheme is itself orthogonal in the rest-equilibrium inner product up to
one relaxation entry, whereas restoring that entry in the lattice-weight inner product
costs stability.

\subsection{Conditioning and cost}
\label{sec:cost}

Orthogonality improves the conditioning of the transform, which matters for the
floating-point evaluation of $\Tm^{-1}$ even where the formulations are gauge
equivalent: for $|\bfu|\le0.2$, $\kappa_2(\Tm)$ is $6.2$--$7.5$ against $11.1$--$12.6$
on D2Q9 and $18.0$--$20.9$ against $41.8$--$46.5$ on D3Q27. The non-orthogonal
transform is sparser and built from simpler constants, with $226$ against $482$
nonzeros in the inverse rest-frame transform on D3Q27. By
Corollary~\ref{cor:impl}, however, that sparsity belongs to the basis in which a
scheme is computed, not to the scheme. Compiled collision kernels that share the
central-moment transforms and differ only in the relaxation (\ref{sec:kernels}) show
that, computed with the conjugated matrix, the orthogonal scheme costs
nothing extra on D2Q9 with one shear rate, where the two matrices coincide, $8$--$10\%$
with an independent bulk rate, and at most $2\%$ on D3Q27, where the kernels are
bandwidth bound in memory; explicit transforms to the orthogonal basis cost $20\%$.

\section{Discussion and conclusions}
\label{sec:discussion}

Whether a change of moment basis changes a lattice Boltzmann scheme is decided by
$[\Lam,\Am]$ in the columns of the non-conserved moments, and the answer depends on
the lattice more than on the basis. By Theorem~\ref{thm:geom} the decisive
feature is the velocity set: orthogonalisation reaches the shear moments only through
the face and body diagonals. On D2Q9 and D3Q15 with one shear rate it therefore
changes nothing, so comparisons in that setting measure the equilibrium, the forcing
or the floating-point conditioning, never the basis. On D3Q19 and D3Q27 it changes
the scheme for every relaxation matrix in which the fourth-order moments do not
share the shear rate. On rectangular lattices, which break the symmetry between the
axes, it changes the scheme even with one shear rate. These results complement the rotational-symmetry constraint of
Ref.~\cite{lishan2021}: symmetry forces degeneracies in $\Lam$, and degeneracies in
$\Lam$ are exactly what create gauge freedom in the basis.

The claim that orthogonalisation spuriously couples moments and degrades stability
is therefore conditional. It is right for the published rectangular scheme, which is
the orthogonal scheme of the rest-equilibrium inner product with the one off-diagonal
relaxation entry that an independent bulk rate requires: dropping that entry breaks the
transport coefficients, and restoring it in the lattice-weight inner product costs
stability, in the linear analysis and in a lid-driven cavity alike.
Wherever the two formulations share their hydrodynamics, the orthogonal one, taken in
or near the inner product of the rest equilibrium, was as robust or more robust in
the least stable direction in every case studied, to within $3.6\%$ at the most extreme
shear rate on D2Q9 and $0.5\%$ on D3Q19, and it was the more accurate against
direct numerical simulation. On D3Q19, where orthogonalisation reaches only the normal
stresses, the two were equally robust in the least stable direction, because the mode
that limits it carries almost none of them. These are four lattices, one or two
relaxation matrices on each and a few flow directions, and we make no claim that the
ordering is universal, only that it cannot be inferred from the presence of coupling
and that $[\Lam,\Am]$ decides whether there is anything to measure.

\begin{tcolorbox}[colback=black!3,colframe=black!55,boxrule=0.5pt,arc=1pt,
  left=4pt,right=4pt,top=3pt,bottom=3pt,title={Recommendations},
  fonttitle=\bfseries\small,coltitle=black,colbacktitle=black!10]
\begin{enumerate}
\item Before comparing two formulations, form $\Am=\Tm_2\Tm_1^{-1}$, check that it
  is constant, and inspect $[\Lam,\Am]$ in the non-conserved columns, adding the
  momentum columns if a force acts. If they vanish, the formulations are one scheme,
  and only their implementations can differ.
\item Where they differ, prefer a basis orthogonal in the inner product of the rest
  equilibrium: its rest state is stable for all rates (Proposition~\ref{prop:rest}),
  and in the tests here it was never less robust by more than a few per cent, often
  much more robust, and the more accurate. On square and
  cubic lattices with $c_s^2=1/3$ this is the usual lattice-weight inner product.
\item Compute it through the cheapest transform in its gauge class, with the
  conjugated relaxation matrix $\Am^{-1}\Lam\Am$ (Corollary~\ref{cor:impl}).
\item On rectangular lattices with an independent bulk rate, use the non-orthogonal
  basis, which with the published rates is the rest-equilibrium orthogonal one with the
  off-diagonal entry~\eqref{eq:rect-key} restored; a diagonal relaxation in an
  orthogonal basis is inconsistent there.
\item The recipe, one rate for all mixed moments, makes the formulations one scheme
  and on D2Q9 keeps the bulk viscosity free, but it removes the damping of the
  fourth-order moments as that rate approaches two, and on D3Q19 and D3Q27 it fails
  in under-resolved flow.
\end{enumerate}
\end{tcolorbox}

Several extensions remain. Cumulants lie outside the class treated here, the map from
populations being nonlinear. Relaxation matrices that depend on the local state,
velocity-dependent forces for which the trapezoidal correction is no longer exact,
and boundary schemes written in moment space fall outside the constant-$\Am$
hypothesis and would need a corresponding generalisation; all tests here were
periodic except the lid-driven cavity, whose half-way bounce-back walls act on
populations. The exact finite difference formulation of Ref.~\cite{bellotti2022} is a
natural setting in which to ask whether the equivalence classes identified here
coincide with equality of the associated macroscopic schemes or are strictly finer.

The question, in short, is not whether a moment basis is orthogonal, but whether its
change of coordinates mixes moments across relaxation-rate boundaries, and, if it
does, in which inner product it is orthogonal.

\section*{Code and data availability}
The scripts that produce every result in this paper, the generated collision
kernels and the raw results are deposited in a public repository: \url{https://doi.org/10.5281/zenodo.22979476}.

\section*{Declaration of competing interest}
The author declares no competing interests.


\appendix

\section{The inner product}
\label{sec:ip}

The orthogonal bases of Section~\ref{sec:stability} were constructed with the lattice
weights, which on D2Q9, D3Q19 and D3Q27 are those of the rest equilibrium, so that
Proposition~\ref{prop:rest} applies to them. To see how much of their advantage is
owed to that choice, we repeated the comparison with Gram--Schmidt in a
one-parameter family of product inner products, with one-dimensional weights
$(1-c_w^2,\,c_w^2/2,\,c_w^2/2)$ for the velocities $(0,\pm1)$, rescaled in $y$ on the
rectangular lattice so that both second moments equal $c_w^2$. The value
$c_w^2=c_s^2$ gives the rest equilibrium, and $c_w^2=2c_s^2$ on the square lattices
the unweighted inner product, in which, for instance, the raw-moment basis of
Ref.~\cite{lallemand2000} is orthogonal.

The consistency result of Section~\ref{sec:rectstab} does not depend on the choice.
For any product inner product the offending entry is
given by Eq.~\eqref{eq:A54-var} of Section~\ref{sec:when}.
With the
published matrix the shear viscosity of a wave at $45^\circ$ is $5.7$, $4.8$ and
$5.7$ times too large at $a=0.5$ for the lattice-weight, rest-equilibrium and
unweighted inner products. The same argument shows that the D2Q9 equivalence of
Corollary~\ref{cor:d2q9} and the D3Q27 obstruction of Table~\ref{tab:projections}
hold for every product inner product that treats the axes alike.

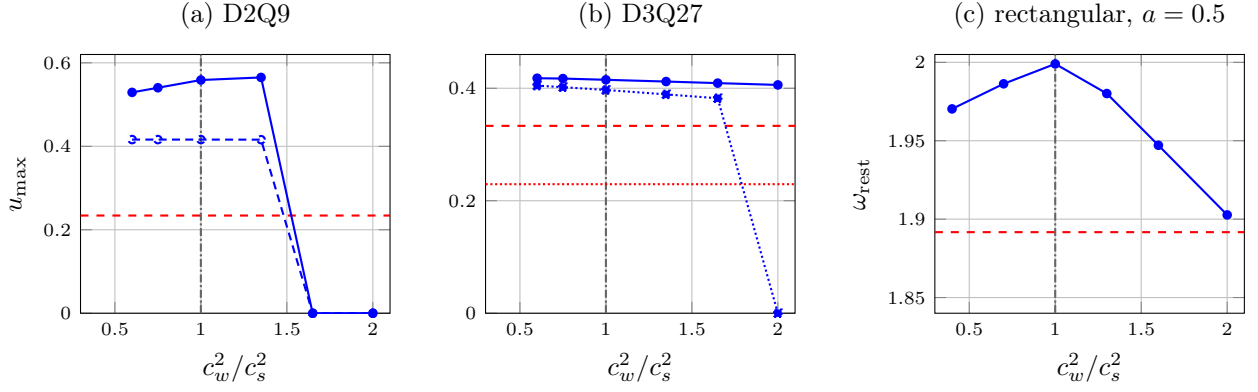
\begin{figure}[t]
\centering
\begin{tikzpicture}
\begin{axis}[width=0.345\textwidth,height=5cm,title={(a) D2Q9},
  xlabel={$c_w^2/c_s^2$}, ylabel={$u_{\max}$}, grid=major, xmin=0.3,xmax=2.1, ymin=0,ymax=0.62,
  tick label style={font=\scriptsize}]
\addplot[black!60,thick,densely dashdotted] coordinates {(1,0)(1,0.62)};
\addplot[red,thick,dashed] coordinates {(0.3,0.2341)(2.1,0.2341)};
\addplot[blue,thick,mark=*,mark size=1.4pt] coordinates {(0.600,0.5292)(0.750,0.5403)(1.000,0.5588)(1.350,0.5650)(1.650,0.0000)(2.000,0.0000)};
\addplot[blue,thick,densely dashed,mark=o,mark size=1.4pt] coordinates {(0.600,0.4161)(0.750,0.4161)(1.000,0.4161)(1.350,0.4159)(1.650,0.0000)(2.000,0.0000)};
\end{axis}
\end{tikzpicture}\hfill
\begin{tikzpicture}
\begin{axis}[width=0.345\textwidth,height=5cm,title={(b) D3Q27},
  xlabel={$c_w^2/c_s^2$}, grid=major, xmin=0.3,xmax=2.1, ymin=0,ymax=0.46,
  tick label style={font=\scriptsize}]
\addplot[black!60,thick,densely dashdotted] coordinates {(1,0)(1,0.46)};
\addplot[red,thick,dashed] coordinates {(0.3,0.3331)(2.1,0.3331)};
\addplot[red,thick,densely dotted] coordinates {(0.3,0.2294)(2.1,0.2294)};
\addplot[blue,thick,mark=*,mark size=1.4pt] coordinates {(0.600,0.4178)(0.750,0.4172)(1.000,0.4150)(1.350,0.4119)(1.650,0.4091)(2.000,0.4059)};
\addplot[blue,thick,densely dotted,mark=square*,mark size=1.4pt] coordinates {(0.600,0.4047)(0.750,0.4019)(1.000,0.3969)(1.350,0.3891)(1.650,0.3822)(2.000,0.0000)};
\end{axis}
\end{tikzpicture}\hfill
\begin{tikzpicture}
\begin{axis}[width=0.345\textwidth,height=5cm,title={(c) rectangular, $a=0.5$},
  xlabel={$c_w^2/c_s^2$}, ylabel={$\omega_{\mathrm{rest}}$}, grid=major, xmin=0.3,xmax=2.1, ymin=1.84,ymax=2.005,
  tick label style={font=\scriptsize}, yticklabel style={/pgf/number format/fixed,/pgf/number format/precision=2}]
\addplot[black!60,thick,densely dashdotted] coordinates {(1,1.84)(1,2.005)};
\addplot[red,thick,dashed] coordinates {(0.3,1.8917)(2.1,1.8917)};
\addplot[blue,thick,mark=*,mark size=1.4pt] coordinates {(0.400,1.9703)(0.700,1.9863)(1.000,1.9990)(1.300,1.9801)(1.600,1.9471)(2.000,1.9027)};
\end{axis}
\end{tikzpicture}
\caption{Dependence of the orthogonal formulation on the inner product in which it is
orthogonalised: product weights with second moment $c_w^2$ in each direction, so that
$c_w^2=c_s^2$ (dash-dotted line) is the rest equilibrium and, on the square lattices,
$c_w^2=2c_s^2$ the unweighted inner product. (a) D2Q9 at $\omega=\omega_b=1.9$ and
$\lambda_4=1$: $u_{\max}$ for flow along the diagonal (filled) and the axis (open).
(b) D3Q27 with the standard matrix: $u_{\max}$ in the least stable of the three
directions at $\omega=1.9$ (circles) and $1.95$ (squares). (c) Rectangular lattice,
$a=0.5$, $c_s^2=0.1$, bulk and shear rates equal: largest shear rate for which the rest
state is stable. Red lines: the non-orthogonal formulation, dashed for the first and
dotted for the second case of each panel. A value of zero means that the rest state is
unstable.}
\label{fig:ip}
\end{figure}

The stability comparison does depend on the choice (Fig.~\ref{fig:ip}). Near the
rest equilibrium the orthogonal formulation keeps its advantage on the three
lattices of Fig.~\ref{fig:ip}, and the optimum is broad: the equilibrium value lies within $2\%$ of the
best found in each panel, at $u_{\max}=0.559$ against $0.565$ on D2Q9 and
$0.397$ against $0.405$ on D3Q27 at $\omega=1.95$. At rest the equilibrium is
optimal where the scheme is a plain multiple-relaxation-time collision: on the
rectangular lattice with matched rates its rest state stays stable to
$\omega=1.999$, against $1.970$ and $1.903$ at $c_w^2=0.4c_s^2$ and $2c_s^2$. Far
from the equilibrium the advantage is lost. The unweighted inner product makes the
orthogonal scheme unstable at rest at $\omega_b=1.8$ and $1.9$ on D2Q9 and at
$\omega=1.95$ on D3Q27, and on the rectangular lattice at $a=0.5$, where the
unweighted and the lattice-weight inner products are both far from the equilibrium
one, it loses the rest state above $\omega=1.628$, against $1.907$ for the lattice
weights and $1.892$ for the non-orthogonal scheme. What stabilises the orthogonal
formulations of this paper is therefore orthogonality in, or near, the inner
product of the rest equilibrium, and not orthogonality as such. On D2Q9 and D3Q27
that is the lattice-weight inner product in routine use; on rectangular lattices
with a free sound speed it is not, and an orthogonalisation that ignores the
equilibrium weights can be less stable than none.

\section{Convergence of the nonlinear stability limits}
\label{sec:convergence}
\label{app:conv}

Tables~\ref{tab:convergence}, \ref{tab:d3conv} and~\ref{tab:d19conv} collect the
resolution and run-length studies of the nonlinear survival limits discussed in
Secs.~\ref{sec:nonlinear}, \ref{sec:d3stab} and~\ref{sec:d19stab}.

\begin{table}[h]
\centering
\caption{Convergence of the nonlinear stability limits $u_{\max}$, orthogonal /
non-orthogonal, at $\omega=1.9$, with the run length scaled with the lattice size
($T\propto N$) and one run of three times the length. The last row of each block
is the von Neumann limit. The two formulations coincide at $\omega_b=1$
throughout, and every simulated limit lies above the corresponding linear one and
approaches it or holds steady as $N$ or $T$ grows.}
\label{tab:convergence}
\small\setlength{\tabcolsep}{4pt}
\begin{tabular}{llccc}
\toprule
 & run & $\omega_b=1.0$ & $\omega_b=1.8$ & $\omega_b=1.9$\\
\midrule
$\theta=0$     & $48^2$, $700$ steps  & $0.4424\,/\,0.4424$ & $0.4248\,/\,0.3486$ & $0.4248\,/\,0.2754$\\
               & $64^2$, $933$ steps  & $0.4365\,/\,0.4365$ & $0.4219\,/\,0.3311$ & $0.4219\,/\,0.2637$\\
               & $96^2$, $1400$ steps & $0.4307\,/\,0.4307$ & $0.4219\,/\,0.3193$ & $0.4219\,/\,0.2520$\\
               & $48^2$, $2100$ steps & $0.4248\,/\,0.4248$ & $0.4219\,/\,0.3135$ & $0.4189\,/\,0.2461$\\
               & linear               & $0.4154\,/\,0.4154$ & $0.4160\,/\,0.2986$ & $0.4162\,/\,0.2367$\\
\midrule
$\theta=\pi/4$ & $48^2$, $700$ steps  & $0.6123\,/\,0.6123$ & $0.6006\,/\,0.5215$ & $0.5947\,/\,0.4658$\\
               & $96^2$, $1400$ steps & $0.6035\,/\,0.6035$ & $0.5889\,/\,0.4863$ & $0.5771\,/\,0.4336$\\
               & linear               & $0.5748\,/\,0.5748$ & $0.5690\,/\,0.4557$ & $0.5591\,/\,0.2341$\\
\bottomrule
\end{tabular}
\end{table}

\begin{table}[h]
\centering
\caption{Survival limits of the advected Taylor--Green vortex on D3Q27,
orthogonal\,/\,non-orthogonal, for flow along the lattice axis and the face
diagonal, with the run length scaled with the lattice size and one run of three
times the length, and the margin of the orthogonal formulation in the least
stable of the two directions. The last row of each block is the linear limit of
Fig.~\ref{fig:d3stab}.}
\label{tab:d3conv}
\small\setlength{\tabcolsep}{5pt}
\begin{tabular}{llccc}
\toprule
 & run & axis & face diagonal & margin\\
\midrule
$\omega=1.9$ & $24^3$, $600$ steps & $0.436\,/\,0.436$ & $0.615\,/\,0.439$ & $+0\%$\\
 & $32^3$, $800$ steps & $0.432\,/\,0.436$ & $0.615\,/\,0.411$ & $+5\%$\\
 & $48^3$, $1200$ steps & $0.429\,/\,0.429$ & $0.608\,/\,0.387$ & $+11\%$\\
 & $24^3$, $1800$ steps & $0.422\,/\,0.425$ & $0.601\,/\,0.369$ & $+14\%$\\
 & linear & $0.415\,/\,0.420$ & $0.574\,/\,0.333$ & $+25\%$\\
\midrule
$\omega=1.95$ & $24^3$, $600$ steps & $0.432\,/\,0.429$ & $0.615\,/\,0.359$ & $+21\%$\\
 & $32^3$, $800$ steps & $0.432\,/\,0.418$ & $0.612\,/\,0.330$ & $+31\%$\\
 & $48^3$, $1200$ steps & $0.425\,/\,0.411$ & $0.605\,/\,0.302$ & $+41\%$\\
 & $24^3$, $1800$ steps & $0.418\,/\,0.408$ & $0.594\,/\,0.274$ & $+53\%$\\
 & linear & $0.397\,/\,0.404$ & $0.537\,/\,0.229$ & $+73\%$\\
\bottomrule
\end{tabular}
\end{table}

\begin{table}[h]
\centering
\caption{Survival limits of the advected Taylor--Green vortex on D3Q19,
orthogonal\,/\,non-orthogonal, as in Table~\ref{tab:d3conv}, for flow along the axis and
the face and body diagonals, and the margin of the orthogonal formulation in the least
stable of the three directions. The last row of each block is the linear limit of
Fig.~\ref{fig:d19stab}.}
\label{tab:d19conv}
\small\setlength{\tabcolsep}{4pt}
\begin{tabular}{llcccc}
\toprule
 & run & axis & face diagonal & body diagonal & margin\\
\midrule
$\omega=1.9$ & $24^3$, $600$ steps & $0.432\,/\,0.418$ & $0.552\,/\,0.436$ & $0.450\,/\,0.450$ & $+3\%$\\
 & $32^3$, $800$ steps & $0.432\,/\,0.408$ & $0.534\,/\,0.422$ & $0.432\,/\,0.432$ & $+6\%$\\
 & $48^3$, $1200$ steps & $0.425\,/\,0.401$ & $0.506\,/\,0.408$ & $0.404\,/\,0.411$ & $+1\%$\\
 & $24^3$, $1800$ steps & $0.422\,/\,0.394$ & $0.485\,/\,0.397$ & $0.387\,/\,0.390$ & $-1\%$\\
 & linear & $0.411\,/\,0.386$ & $0.405\,/\,0.383$ & $0.227\,/\,0.227$ & $0\%$\\
\midrule
$\omega=1.95$ & $24^3$, $600$ steps & $0.432\,/\,0.394$ & $0.538\,/\,0.397$ & $0.436\,/\,0.436$ & $+10\%$\\
 & $32^3$, $800$ steps & $0.429\,/\,0.383$ & $0.510\,/\,0.387$ & $0.415\,/\,0.418$ & $+8\%$\\
 & $48^3$, $1200$ steps & $0.422\,/\,0.376$ & $0.478\,/\,0.373$ & $0.383\,/\,0.390$ & $+3\%$\\
 & $24^3$, $1800$ steps & $0.415\,/\,0.369$ & $0.454\,/\,0.366$ & $0.362\,/\,0.369$ & $-1\%$\\
 & linear & $0.399\,/\,0.362$ & $0.314\,/\,0.299$ & $0.159\,/\,0.159$ & $0\%$\\
\bottomrule
\end{tabular}
\end{table}

\section{Verification}
\label{sec:verification}

The algebraic results were obtained by exact rational and symbolic computation,
with populations and velocity components carried as free symbols.
Corollary~\ref{cor:d2q9} was confirmed at population level: the difference
$f^{\star}_{\mathrm{orth}}-f^{\star}_{\mathrm{non}}$ vanishes identically as a
nine-vector, with a fourth-order equilibrium and a fourth-order Hermite forcing
term retained in full.

The D3Q27 commutators were verified twice, once for the order-graded
orthogonalisation and once for the second-order directions
\begin{equation*}
  \{\Cx\Cy,\ \Cx\Cz,\ \Cy\Cz,\ \Cx^2-\Cy^2,\ \Cx^2+\Cy^2-2\Cz^2\}
\end{equation*}
of the published orthogonal formulation of Ref.~\cite{premnath2011}; the two give identical
commutators, as Proposition~\ref{prop:projection} requires. Beyond these,
$500$ randomly generated admissible orthogonal bases were tested on each lattice,
obtained by random invertible remixing within each relaxation-rate group followed
by re-orthogonalisation. On D2Q9 the commutator vanished in all $500$ cases; on
D3Q27 it was nonzero in all $500$, with a fourth-order-to-deviatoric coupling
present in every one.

Corollary~\ref{cor:physical} was checked by applying $\Delta\Pm$ to random
admissible vectors, with zero mass and momentum. On D3Q27 the result is $0.70$
for the standard multiple-relaxation-time matrix and $8\times10^{-15}$ for the
two-rate central matrix of Corollary~\ref{cor:d3q27}; on D2Q9 it is $0.60$
with a split bulk rate and $1\times10^{-15}$ once the fourth-order moment shares
that rate.

The numerical results were validated as reported in Section~\ref{sec:stability}: the complex-step
Jacobians against the exact equivalence at $\omega_b=1$ to $3.3\times10^{-16}$
and against wavevector refinement to four decimal places by $N=48$; the nonlinear
solver against the analytic Taylor--Green decay rate to $0.10\%$; and the
single-code-path reformulation against the independently assembled orthogonal
transform to $5.6\times10^{-16}$. The nonlinear limits were checked for
convergence in resolution and run length, as reported in
Table~\ref{tab:convergence}, and the D3Q27 linear limits for convergence in the
wavevector grid, as reported below.

The solvers added for the later tests were checked against the earlier ones. The
rectangular implementation reproduces the Gram--Schmidt coefficients
\eqref{eq:A-rect} to machine precision and the transport coefficients of
Ref.~\cite{yahia2021} as described in Section~\ref{sec:rectstab}. With the
entry~\eqref{eq:rect-key} restored, the orthogonal rectangular scheme returns the shear
viscosity within $3\times10^{-5}$ of its nominal value at $0^\circ$, $45^\circ$ and
$90^\circ$ and the sound attenuation of the published scheme to four digits, and in the
rest-equilibrium inner product, where $\Am_{95}$ vanishes, it reproduces the published
post-collision populations to round-off. The
three-dimensional nonlinear solver, which computes central moments axis by axis,
agrees with a direct matrix implementation of the collision to $4\times10^{-16}$
for both bases and both relaxation matrices, and its two-dimensional counterpart,
used for the double shear layer, agrees with the solver of
Section~\ref{sec:nonlinear} to $2\times10^{-16}$.
The decay of a small-amplitude Taylor--Green vortex on $32^3$ returns the shear
viscosity to within $0.4\%$ in both bases at $\omega=1.6$. Proposition~\ref{prop:rest}
was checked on D2Q9 by evaluating $\max_{\bm k}\|\mathsf G(\bm k)\|_w$ for the
random rate assignments of Section~\ref{sec:reststab-theory}: it equals $1$ to
round-off in every draw for the equilibrium-orthogonal basis and reaches $2.1$ for
the non-orthogonal and the unweighted ones. Theorem~\ref{thm:geom} was checked in
exact rational arithmetic with random weights, constant on each velocity shell, on
D2Q9, D3Q15, D3Q19 and D3Q27: the projections onto the deviatoric block of
$\Cx^2\Cy^2$ and $\Cx\Cy\Cz^2$ equal the weights of the face and body diagonals, as the
proof states, and vanish, or the monomial is absent or dependent, exactly where the
theorem says; with the standard weights they reproduce Table~\ref{tab:projections}.
Equation~\eqref{eq:A54-var} agrees with
the Gram--Schmidt coefficient to round-off for three inner products at $a=0.5$, $0.8$
and $2$. The D3Q19 solver, which obtains central moments from raw monomial moments by the
binomial shift, agrees with a direct matrix implementation of the collision to
$7\times10^{-16}$ in both bases; its Gram--Schmidt matrix has
exactly the six edges between the deviatoric block and the fourth order of
Fig.~\ref{fig:graphs}(d); and a small-amplitude Taylor--Green vortex at $\omega=1.6$
returns the shear viscosity to within $0.8\%$ on $32^3$ and $0.2\%$ on $64^3$ in
both bases.

\paragraph{Checks of the stability analysis}
At $\omega_b=1$, where Corollary~\ref{cor:d2q9} predicts exact equivalence, the two
D2Q9 Jacobians agree to $3.3\times10^{-16}$, against $0.397$ at $\omega_b=1.6$; with
$\lambda_4=\omega_b=1.6$ they agree to $8.9\times10^{-16}$, and the von Neumann limits of
the recipe coincide in the two bases at all twelve points computed. The D2Q9 limits
are converged in the wavevector grid: at $\omega_b=1.8$ they take the values $0.4171$,
$0.4161$, $0.4160$ and $0.4161$ (orthogonal) and $0.2986$, $0.2987$, $0.2986$ and
$0.2986$ (non-orthogonal) for $N=32$, $48$, $64$ and $96$, and at the six points of
largest separation a scan from rest in steps of $0.01$ reproduces the bisection limits
to within $2\times10^{-3}$. On D3Q27 at $\omega=1.9$ the Jacobians differ by $0.161$
for the standard matrix and agree to $2.1\times10^{-15}$ for the two-rate central
matrix of Corollary~\ref{cor:d3q27}, and the limits for $N=12$, $16$ and $24$ agree
to within the bisection tolerance of $5\times10^{-4}$. The D3Q19 checks are collected in \ref{sec:wavevector}.

\paragraph{Round-off}
Two independent D2Q9 implementations, each forming and inverting its own transform,
were run from the same initial condition in a shear layer (Fig.~\ref{fig:roundoff}).
At the gauge points the discrepancy begins at zero and accumulates as round-off does
in two different sequences of operations, staying eight to nine orders of magnitude
below the genuine scheme difference at $\omega_b=1.6$ and $\lambda_4=1$, which is
present from the first sample and does not grow. Gauge-equivalent formulations are
thus identical in exact arithmetic but not bitwise in floating point.

\begin{figure}[t]
\centering
\begin{tikzpicture}
\begin{axis}[width=0.72\textwidth,height=5.6cm, ymode=log,
  xlabel={time step}, ylabel={$\max|f_{\mathrm{orth}}-f_{\mathrm{non}}|\,/\,\max|f|$},
  xmin=0,xmax=2000, ymin=1e-16, ymax=1e-4, grid=major,
  xtick={0,500,1000,1500,2000}, /pgf/number format/1000 sep={},
  legend style={font=\small,at={(0.98,0.5)},anchor=east}, legend cell align=left]
\addplot[blue,thick] coordinates {(25,1.998e-15)(50,1.873e-15)(75,1.749e-15)(100,1.873e-15)(125,1.623e-15)(150,1.998e-15)(175,1.998e-15)(200,1.997e-15)(225,1.998e-15)(250,1.998e-15)(275,2.122e-15)(300,1.998e-15)(325,1.748e-15)(350,1.997e-15)(375,1.998e-15)(400,1.998e-15)(425,2.122e-15)(450,2.278e-15)(475,2.185e-15)(500,2.528e-15)(525,2.622e-15)(550,2.622e-15)(575,2.684e-15)(600,2.996e-15)(625,3.121e-15)(650,3.402e-15)(675,3.558e-15)(700,3.714e-15)(725,3.776e-15)(750,3.838e-15)(775,4.151e-15)(800,4.307e-15)(825,4.462e-15)(850,4.619e-15)(875,5.024e-15)(900,5.024e-15)(925,5.305e-15)(950,5.430e-15)(975,5.585e-15)(1000,5.928e-15)(1025,6.272e-15)(1050,5.990e-15)(1075,6.458e-15)(1100,6.427e-15)(1125,6.676e-15)(1150,6.770e-15)(1175,7.020e-15)(1200,7.300e-15)(1225,7.487e-15)(1250,7.830e-15)(1275,7.892e-15)(1300,8.266e-15)(1325,8.204e-15)(1350,8.484e-15)(1375,8.515e-15)(1400,8.827e-15)(1425,9.138e-15)(1450,9.418e-15)(1475,9.668e-15)(1500,9.668e-15)(1525,9.948e-15)(1550,9.979e-15)(1575,1.045e-14)(1600,1.029e-14)(1625,1.063e-14)(1650,1.113e-14)(1675,1.107e-14)(1700,1.119e-14)(1725,1.166e-14)(1750,1.157e-14)(1775,1.147e-14)(1800,1.219e-14)(1825,1.194e-14)(1850,1.203e-14)(1875,1.222e-14)(1900,1.250e-14)(1925,1.253e-14)(1950,1.269e-14)(1975,1.297e-14)(2000,1.303e-14)};
\addlegendentry{$\omega_b=1$ (same scheme)}
\addplot[black,thick,densely dotted] coordinates {(25,2.122e-15)(50,2.372e-15)(75,2.123e-15)(100,2.123e-15)(125,2.372e-15)(150,2.248e-15)(175,2.497e-15)(200,2.372e-15)(225,2.123e-15)(250,2.373e-15)(275,3.245e-15)(300,2.372e-15)(325,2.622e-15)(350,2.497e-15)(375,2.622e-15)(400,2.498e-15)(425,2.746e-15)(450,2.497e-15)(475,2.747e-15)(500,2.902e-15)(525,2.902e-15)(550,3.340e-15)(575,3.308e-15)(600,3.277e-15)(625,3.840e-15)(650,3.745e-15)(675,3.744e-15)(700,3.995e-15)(725,4.213e-15)(750,4.431e-15)(775,4.338e-15)(800,5.118e-15)(825,4.930e-15)(850,4.993e-15)(875,5.212e-15)(900,5.522e-15)(925,5.867e-15)(950,5.773e-15)(975,5.990e-15)(1000,6.022e-15)(1025,6.335e-15)(1050,6.582e-15)(1075,6.614e-15)(1100,6.740e-15)(1125,6.925e-15)(1150,7.643e-15)(1175,7.457e-15)(1200,7.580e-15)(1225,7.611e-15)(1250,8.081e-15)(1275,8.329e-15)(1300,8.390e-15)(1325,8.485e-15)(1350,8.765e-15)(1375,9.075e-15)(1400,8.952e-15)(1425,9.388e-15)(1450,9.386e-15)(1475,9.918e-15)(1500,1.029e-14)(1525,1.007e-14)(1550,1.038e-14)(1575,1.082e-14)(1600,1.094e-14)(1625,1.082e-14)(1650,1.098e-14)(1675,1.122e-14)(1700,1.126e-14)(1725,1.160e-14)(1750,1.150e-14)(1775,1.166e-14)(1800,1.229e-14)(1825,1.222e-14)(1850,1.219e-14)(1875,1.238e-14)(1900,1.275e-14)(1925,1.259e-14)(1950,1.297e-14)(1975,1.278e-14)(2000,1.322e-14)};
\addlegendentry{$\lambda_4=\omega_b=1.6$ (same scheme)}
\addplot[red,thick,dashed] coordinates {(25,6.035e-06)(50,3.540e-06)(75,2.106e-06)(100,5.167e-06)(125,3.476e-06)(150,2.442e-06)(175,5.367e-06)(200,2.139e-06)(225,3.105e-06)(250,5.827e-06)(275,1.020e-06)(300,3.712e-06)(325,5.079e-06)(350,2.150e-06)(375,4.137e-06)(400,4.024e-06)(425,3.556e-06)(450,4.004e-06)(475,2.440e-06)(500,4.309e-06)(525,3.404e-06)(550,1.162e-06)(575,4.589e-06)(600,2.771e-06)(625,1.054e-06)(650,4.428e-06)(675,2.103e-06)(700,1.688e-06)(725,4.099e-06)(750,1.554e-06)(775,2.193e-06)(800,3.792e-06)(825,8.697e-07)(850,2.557e-06)(875,3.434e-06)(900,1.357e-06)(925,2.768e-06)(950,3.028e-06)(975,1.978e-06)(1000,2.834e-06)(1025,2.424e-06)(1050,2.640e-06)(1075,2.688e-06)(1100,1.729e-06)(1125,3.145e-06)(1150,2.409e-06)(1175,9.284e-07)(1200,3.495e-06)(1225,1.956e-06)(1250,1.114e-06)(1275,3.544e-06)(1300,1.586e-06)(1325,1.506e-06)(1350,3.412e-06)(1375,1.245e-06)(1400,1.851e-06)(1425,3.068e-06)(1450,1.072e-06)(1475,2.092e-06)(1500,2.661e-06)(1525,1.584e-06)(1550,2.209e-06)(1575,2.148e-06)(1600,2.049e-06)(1625,2.206e-06)(1650,1.646e-06)(1675,2.369e-06)(1700,2.079e-06)(1725,1.077e-06)(1750,2.625e-06)(1775,1.906e-06)(1800,1.259e-06)(1825,2.709e-06)(1850,1.674e-06)(1875,1.407e-06)(1900,2.685e-06)(1925,1.493e-06)(1950,1.560e-06)(1975,2.488e-06)(2000,1.299e-06)};
\addlegendentry{$\omega_b=1.6$, $\lambda_4=1$}
\end{axis}
\end{tikzpicture}
\caption{Maximum relative difference between the populations of independent
orthogonal and non-orthogonal implementations in a shear layer ($64^2$,
$\omega=1.9$). The two gauge-equivalent pairs grow as round-off does, reaching
$1.3\times10^{-14}$ after $2000$ steps, while the inequivalent pair sits at
$\sim10^{-6}$ from the first sample onwards, eight to nine orders of magnitude
higher.}
\label{fig:roundoff}
\end{figure}
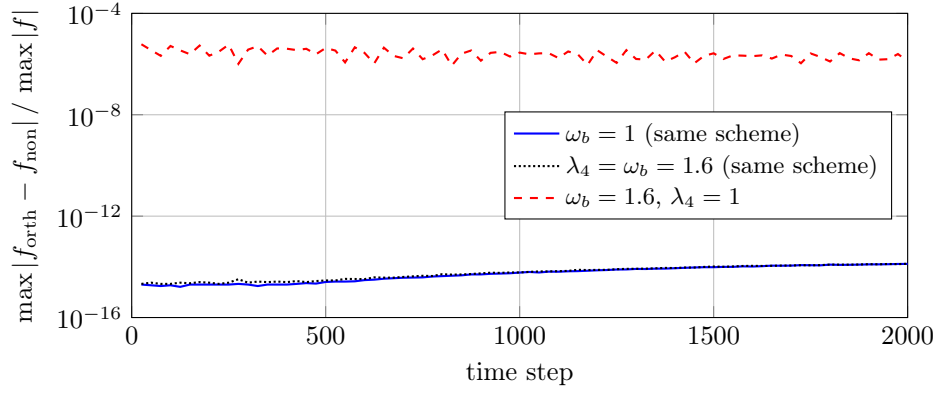

\section{Lid-driven cavity on the rectangular lattice}
\label{sec:cavity}

The stability test of Ref.~\cite{yahia2021} is repeated for the published scheme and the variants
of Section~\ref{sec:rectstab}: a square cavity on a $100\times200$ rectangular lattice with $a=0.5$
and $c_s^2=0.1$, a lid velocity $U=0.2$, half-way bounce-back on every wall with the
moving-lid correction of Ref.~\cite{yahia2021}, $\rho_wc_s^2U/(2a^2)$ on the diagonal links,
which is the momentum-corrected bounce-back with the rest-equilibrium weights, and the
density gradient of the extended equilibria by second-order differences, one-sided at the
walls. The shear rate is raised by bisection, to within $0.002$, until a run of
$100\,000$ steps, $200$ convective times, no longer stays bounded. The vectorised solver
agrees with the per-node implementation of the linear analysis to $3\times10^{-16}$ in all
three bases. Table~\ref{tab:cavity} lists the limits, and Fig.~\ref{fig:cavity} compares the
centreline velocities at $\mathrm{Re}=400$ with a $256^2$ square-lattice reference.

\begin{table}[t]
\centering
\caption{Lid-driven cavity on the rectangular lattice ($a=0.5$, $c_s^2=0.1$, $100\times200$
nodes). Largest shear rate $\omega$ for which a run with lid velocity $U=0.2$ stays
bounded for $100\,000$ steps, the corresponding $\mathrm{Re}=UN_x/\nu$, the linear
rest-state limit of Section~\ref{sec:rectstab}, and the largest departure of the
centreline velocities from a $256^2$ square-lattice reference at $\mathrm{Re}=400$, relative
to the lid velocity. Ref.~\cite{yahia2021} reports $\mathrm{Re}=6733$ for the published scheme in
the same test.}
\label{tab:cavity}
\small\setlength{\tabcolsep}{5pt}
\begin{tabular}{lcccc}
\toprule
scheme & $\omega_{\max}$ & $\mathrm{Re}_{\max}$ & rest-state limit & departure at $\mathrm{Re}=400$\\
\midrule
published scheme & $1.886$ & $6590$ & $1.871$ & $0.7\%$\\
orthogonal counterpart & $1.589$ & $1547$ & $1.716$ & $16.8\%$\\
orthogonal, entry~\eqref{eq:rect-key} restored & $1.247$ & $662$ & $1.223$ & diverges\\
bulk rate $=$ shear rate, non-orthogonal & $1.706$ & $2323$ & $1.892$ & $0.7\%$\\
bulk rate $=$ shear rate, orthogonal & $1.859$ & $5276$ & $1.907$ & $0.7\%$\\
recipe & $1.842$ & $4669$ & $1.999$ & $0.7\%$\\
\bottomrule
\end{tabular}
\end{table}

\begin{figure}[t]
\centering
\begin{tikzpicture}
\begin{axis}[width=0.46\textwidth,height=6cm,xlabel={$u/U$ at $x=1/2$},ylabel={$y$},grid=major,
  tick label style={font=\scriptsize}]
\addplot[black!45,line width=2.2pt] coordinates {(-0.0032,0.0020)(-0.0389,0.0254)(-0.0720,0.0488)(-0.1036,0.0723)(-0.1351,0.0957)(-0.1669,0.1191)(-0.1991,0.1426)(-0.2310,0.1660)(-0.2611,0.1895)(-0.2873,0.2129)(-0.3076,0.2363)(-0.3202,0.2598)(-0.3241,0.2832)(-0.3191,0.3066)(-0.3061,0.3301)(-0.2865,0.3535)(-0.2622,0.3770)(-0.2351,0.4004)(-0.2066,0.4238)(-0.1778,0.4473)(-0.1493,0.4707)(-0.1214,0.4941)(-0.0940,0.5176)(-0.0670,0.5410)(-0.0401,0.5645)(-0.0132,0.5879)(0.0140,0.6113)(0.0415,0.6348)(0.0693,0.6582)(0.0975,0.6816)(0.1258,0.7051)(0.1540,0.7285)(0.1818,0.7520)(0.2089,0.7754)(0.2349,0.7988)(0.2595,0.8223)(0.2829,0.8457)(0.3068,0.8691)(0.3368,0.8926)(0.3860,0.9160)(0.4771,0.9395)(0.6348,0.9629)(0.8585,0.9863)};
\addplot[red,thick,dashed] coordinates {(-0.0039,0.0025)(-0.0337,0.0225)(-0.0615,0.0425)(-0.0881,0.0625)(-0.1143,0.0825)(-0.1407,0.1025)(-0.1674,0.1225)(-0.1945,0.1425)(-0.2214,0.1625)(-0.2473,0.1825)(-0.2709,0.2025)(-0.2911,0.2225)(-0.3066,0.2425)(-0.3165,0.2625)(-0.3200,0.2825)(-0.3172,0.3025)(-0.3084,0.3225)(-0.2944,0.3425)(-0.2762,0.3625)(-0.2551,0.3825)(-0.2321,0.4025)(-0.2080,0.4225)(-0.1836,0.4425)(-0.1593,0.4625)(-0.1354,0.4825)(-0.1119,0.5025)(-0.0887,0.5225)(-0.0657,0.5425)(-0.0429,0.5625)(-0.0200,0.5825)(0.0030,0.6025)(0.0262,0.6225)(0.0496,0.6425)(0.0733,0.6625)(0.0972,0.6825)(0.1211,0.7025)(0.1450,0.7225)(0.1687,0.7425)(0.1920,0.7625)(0.2146,0.7825)(0.2363,0.8025)(0.2569,0.8225)(0.2768,0.8425)(0.2967,0.8625)(0.3196,0.8825)(0.3515,0.9025)(0.4038,0.9225)(0.4921,0.9425)(0.6307,0.9625)(0.8187,0.9825)};
\addplot[blue,thick,densely dashdotted] coordinates {(-0.0040,0.0025)(-0.0327,0.0225)(-0.0556,0.0425)(-0.0743,0.0625)(-0.0895,0.0825)(-0.1022,0.1025)(-0.1128,0.1225)(-0.1218,0.1425)(-0.1296,0.1625)(-0.1363,0.1825)(-0.1422,0.2025)(-0.1472,0.2225)(-0.1513,0.2425)(-0.1544,0.2625)(-0.1566,0.2825)(-0.1576,0.3025)(-0.1573,0.3225)(-0.1557,0.3425)(-0.1527,0.3625)(-0.1484,0.3825)(-0.1426,0.4025)(-0.1355,0.4225)(-0.1271,0.4425)(-0.1176,0.4625)(-0.1072,0.4825)(-0.0958,0.5025)(-0.0839,0.5225)(-0.0714,0.5425)(-0.0585,0.5625)(-0.0455,0.5825)(-0.0324,0.6025)(-0.0193,0.6225)(-0.0063,0.6425)(0.0066,0.6625)(0.0192,0.6825)(0.0315,0.7025)(0.0437,0.7225)(0.0556,0.7425)(0.0673,0.7625)(0.0791,0.7825)(0.0911,0.8025)(0.1039,0.8225)(0.1187,0.8425)(0.1376,0.8625)(0.1646,0.8825)(0.2069,0.9025)(0.2757,0.9225)(0.3860,0.9425)(0.5522,0.9625)(0.7761,0.9825)};
\addplot[blue,thick] coordinates {(-0.0039,0.0025)(-0.0337,0.0225)(-0.0615,0.0425)(-0.0881,0.0625)(-0.1143,0.0825)(-0.1407,0.1025)(-0.1674,0.1225)(-0.1945,0.1425)(-0.2214,0.1625)(-0.2473,0.1825)(-0.2710,0.2025)(-0.2912,0.2225)(-0.3067,0.2425)(-0.3165,0.2625)(-0.3200,0.2825)(-0.3172,0.3025)(-0.3084,0.3225)(-0.2944,0.3425)(-0.2762,0.3625)(-0.2551,0.3825)(-0.2321,0.4025)(-0.2080,0.4225)(-0.1836,0.4425)(-0.1594,0.4625)(-0.1354,0.4825)(-0.1119,0.5025)(-0.0887,0.5225)(-0.0657,0.5425)(-0.0429,0.5625)(-0.0200,0.5825)(0.0030,0.6025)(0.0262,0.6225)(0.0496,0.6425)(0.0733,0.6625)(0.0972,0.6825)(0.1211,0.7025)(0.1450,0.7225)(0.1687,0.7425)(0.1920,0.7625)(0.2146,0.7825)(0.2363,0.8025)(0.2569,0.8225)(0.2767,0.8425)(0.2967,0.8625)(0.3196,0.8825)(0.3515,0.9025)(0.4037,0.9225)(0.4920,0.9425)(0.6307,0.9625)(0.8186,0.9825)};
\addplot[red!60!black,thick,densely dotted] coordinates {(-0.0039,0.0025)(-0.0337,0.0225)(-0.0615,0.0425)(-0.0881,0.0625)(-0.1143,0.0825)(-0.1407,0.1025)(-0.1674,0.1225)(-0.1945,0.1425)(-0.2214,0.1625)(-0.2473,0.1825)(-0.2709,0.2025)(-0.2911,0.2225)(-0.3066,0.2425)(-0.3165,0.2625)(-0.3200,0.2825)(-0.3172,0.3025)(-0.3084,0.3225)(-0.2944,0.3425)(-0.2762,0.3625)(-0.2551,0.3825)(-0.2321,0.4025)(-0.2080,0.4225)(-0.1836,0.4425)(-0.1594,0.4625)(-0.1354,0.4825)(-0.1119,0.5025)(-0.0887,0.5225)(-0.0657,0.5425)(-0.0429,0.5625)(-0.0200,0.5825)(0.0030,0.6025)(0.0262,0.6225)(0.0496,0.6425)(0.0733,0.6625)(0.0972,0.6825)(0.1211,0.7025)(0.1450,0.7225)(0.1687,0.7425)(0.1920,0.7625)(0.2146,0.7825)(0.2362,0.8025)(0.2569,0.8225)(0.2767,0.8425)(0.2967,0.8625)(0.3195,0.8825)(0.3515,0.9025)(0.4037,0.9225)(0.4920,0.9425)(0.6307,0.9625)(0.8186,0.9825)};
\end{axis}
\end{tikzpicture}\hfill
\begin{tikzpicture}
\begin{axis}[width=0.46\textwidth,height=6cm,xlabel={$x$},ylabel={$v/U$ at $y=1/2$},grid=major,
  tick label style={font=\scriptsize}]
\addplot[black!45,line width=2.2pt] coordinates {(0.0020,0.0081)(0.0254,0.0919)(0.0488,0.1538)(0.0723,0.1980)(0.0957,0.2295)(0.1191,0.2526)(0.1426,0.2702)(0.1660,0.2836)(0.1895,0.2931)(0.2129,0.2982)(0.2363,0.2987)(0.2598,0.2940)(0.2832,0.2843)(0.3066,0.2697)(0.3301,0.2508)(0.3535,0.2283)(0.3770,0.2029)(0.4004,0.1757)(0.4238,0.1472)(0.4473,0.1181)(0.4707,0.0888)(0.4941,0.0597)(0.5176,0.0308)(0.5410,0.0022)(0.5645,-0.0262)(0.5879,-0.0544)(0.6113,-0.0827)(0.6348,-0.1115)(0.6582,-0.1411)(0.6816,-0.1721)(0.7051,-0.2055)(0.7285,-0.2419)(0.7520,-0.2821)(0.7754,-0.3259)(0.7988,-0.3711)(0.8223,-0.4125)(0.8457,-0.4413)(0.8691,-0.4463)(0.8926,-0.4173)(0.9160,-0.3516)(0.9395,-0.2567)(0.9629,-0.1494)(0.9863,-0.0492)};
\addplot[red,thick,dashed] coordinates {(0.0050,0.0200)(0.0250,0.0891)(0.0450,0.1425)(0.0650,0.1828)(0.0850,0.2129)(0.1050,0.2357)(0.1250,0.2534)(0.1450,0.2675)(0.1650,0.2787)(0.1850,0.2872)(0.2050,0.2926)(0.2250,0.2949)(0.2450,0.2936)(0.2650,0.2887)(0.2850,0.2801)(0.3050,0.2681)(0.3250,0.2528)(0.3450,0.2348)(0.3650,0.2146)(0.3850,0.1925)(0.4050,0.1692)(0.4250,0.1451)(0.4450,0.1205)(0.4650,0.0957)(0.4850,0.0709)(0.5050,0.0463)(0.5250,0.0219)(0.5450,-0.0023)(0.5650,-0.0264)(0.5850,-0.0504)(0.6050,-0.0744)(0.6250,-0.0987)(0.6450,-0.1235)(0.6650,-0.1491)(0.6850,-0.1759)(0.7050,-0.2045)(0.7250,-0.2353)(0.7450,-0.2688)(0.7650,-0.3049)(0.7850,-0.3428)(0.8050,-0.3803)(0.8250,-0.4133)(0.8450,-0.4360)(0.8650,-0.4416)(0.8850,-0.4237)(0.9050,-0.3794)(0.9250,-0.3110)(0.9450,-0.2262)(0.9650,-0.1363)(0.9850,-0.0529)};
\addplot[blue,thick,densely dashdotted] coordinates {(0.0050,0.0110)(0.0250,0.0488)(0.0450,0.0782)(0.0650,0.1005)(0.0850,0.1170)(0.1050,0.1290)(0.1250,0.1377)(0.1450,0.1437)(0.1650,0.1479)(0.1850,0.1506)(0.2050,0.1523)(0.2250,0.1530)(0.2450,0.1529)(0.2650,0.1520)(0.2850,0.1503)(0.3050,0.1478)(0.3250,0.1443)(0.3450,0.1400)(0.3650,0.1347)(0.3850,0.1285)(0.4050,0.1212)(0.4250,0.1130)(0.4450,0.1037)(0.4650,0.0934)(0.4850,0.0822)(0.5050,0.0699)(0.5250,0.0567)(0.5450,0.0424)(0.5650,0.0271)(0.5850,0.0107)(0.6050,-0.0068)(0.6250,-0.0255)(0.6450,-0.0455)(0.6650,-0.0669)(0.6850,-0.0897)(0.7050,-0.1140)(0.7250,-0.1398)(0.7450,-0.1667)(0.7650,-0.1943)(0.7850,-0.2216)(0.8050,-0.2471)(0.8250,-0.2689)(0.8450,-0.2842)(0.8650,-0.2897)(0.8850,-0.2822)(0.9050,-0.2594)(0.9250,-0.2207)(0.9450,-0.1681)(0.9650,-0.1067)(0.9850,-0.0437)};
\addplot[blue,thick] coordinates {(0.0050,0.0200)(0.0250,0.0891)(0.0450,0.1425)(0.0650,0.1828)(0.0850,0.2129)(0.1050,0.2357)(0.1250,0.2534)(0.1450,0.2675)(0.1650,0.2787)(0.1850,0.2872)(0.2050,0.2926)(0.2250,0.2949)(0.2450,0.2936)(0.2650,0.2887)(0.2850,0.2801)(0.3050,0.2680)(0.3250,0.2528)(0.3450,0.2348)(0.3650,0.2146)(0.3850,0.1925)(0.4050,0.1692)(0.4250,0.1451)(0.4450,0.1205)(0.4650,0.0957)(0.4850,0.0709)(0.5050,0.0463)(0.5250,0.0219)(0.5450,-0.0023)(0.5650,-0.0264)(0.5850,-0.0504)(0.6050,-0.0744)(0.6250,-0.0987)(0.6450,-0.1235)(0.6650,-0.1491)(0.6850,-0.1759)(0.7050,-0.2045)(0.7250,-0.2353)(0.7450,-0.2688)(0.7650,-0.3049)(0.7850,-0.3428)(0.8050,-0.3803)(0.8250,-0.4133)(0.8450,-0.4361)(0.8650,-0.4416)(0.8850,-0.4237)(0.9050,-0.3794)(0.9250,-0.3110)(0.9450,-0.2262)(0.9650,-0.1363)(0.9850,-0.0529)};
\addplot[red!60!black,thick,densely dotted] coordinates {(0.0050,0.0200)(0.0250,0.0891)(0.0450,0.1425)(0.0650,0.1827)(0.0850,0.2128)(0.1050,0.2356)(0.1250,0.2534)(0.1450,0.2675)(0.1650,0.2787)(0.1850,0.2872)(0.2050,0.2926)(0.2250,0.2948)(0.2450,0.2936)(0.2650,0.2886)(0.2850,0.2801)(0.3050,0.2680)(0.3250,0.2528)(0.3450,0.2348)(0.3650,0.2146)(0.3850,0.1925)(0.4050,0.1692)(0.4250,0.1451)(0.4450,0.1205)(0.4650,0.0957)(0.4850,0.0709)(0.5050,0.0463)(0.5250,0.0219)(0.5450,-0.0023)(0.5650,-0.0264)(0.5850,-0.0504)(0.6050,-0.0744)(0.6250,-0.0987)(0.6450,-0.1235)(0.6650,-0.1491)(0.6850,-0.1759)(0.7050,-0.2045)(0.7250,-0.2353)(0.7450,-0.2688)(0.7650,-0.3049)(0.7850,-0.3428)(0.8050,-0.3803)(0.8250,-0.4133)(0.8450,-0.4360)(0.8650,-0.4416)(0.8850,-0.4237)(0.9050,-0.3794)(0.9250,-0.3110)(0.9450,-0.2262)(0.9650,-0.1363)(0.9850,-0.0529)};
\end{axis}
\end{tikzpicture}
\caption{Centreline velocities in the lid-driven cavity at $\mathrm{Re}=400$ on the
rectangular lattice ($a=0.5$, $c_s^2=0.1$, $100\times200$ nodes, $U=0.1$) after $40$
convective times, against a $256^2$ square-lattice reference (grey): published scheme (red
dashed), orthogonal counterpart (blue dash-dotted), and with the bulk rate equal to the shear
rate the non-orthogonal (dark red dotted) and orthogonal (blue solid) formulations. The orthogonal scheme with the entry~\eqref{eq:rect-key}
restored diverges at this Reynolds number, as its rest-state limit predicts.}
\label{fig:cavity}
\end{figure}
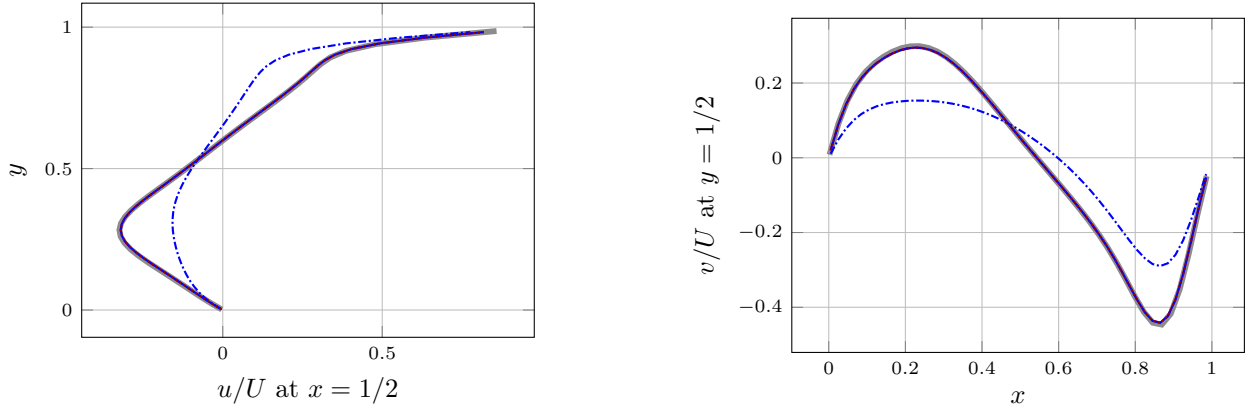

\section{Wavevector search on D3Q19}
\label{sec:wavevector}

On D3Q19 the modes that limit stability along the diagonals occupy narrow bands of
wavevectors, and uniform grids converge slowly. At $\omega=1.9$ the orthogonal limit
along the body diagonal is $0.286$, $0.265$ and $0.246$ on $12^3$, $16^3$ and $24^3$
grids, and along the face diagonal $0.443$, $0.420$ and $0.397$, whereas along the axis
it is $0.412$ on all three. The limits of Section~\ref{sec:d19stab} are therefore
computed with a local maximisation of the spectral radius: from each of the $64$
largest values of $\rho(\mathsf G)$ on a $16^3$ grid, a pattern search evaluates
$\rho$ on the $3\times3\times3$ neighbourhood of the current wavevector, moves to the
largest value and halves the step, fourteen times. The search detected instability at
least as early as uniform $48^3$ and $64^3$ grids at every point checked. Along the
body diagonal at $\omega=1.9$, for instance, it finds $\rho-1=1.2\times10^{-5}$ at
$u=0.23$ for the orthogonal formulation, where both uniform grids return a stable
scheme, and $6.3\times10^{-5}$ at $u=0.24$ for the non-orthogonal one, against
$5.3\times10^{-5}$ and $4.0\times10^{-5}$ on the $48^3$ and $64^3$ grids. The most
unstable wavevectors, near $(-0.14,-\pi,-0.14)$ at $u=0.235$, lie between the nodes of
every uniform grid. The same search leaves the D3Q27 limits of Fig.~\ref{fig:d3stab}
unchanged: at each of the twelve limits at $\omega=1.9$ and $1.95$ it finds the scheme
stable $0.002$ below the limit and unstable $0.002$ above it. On D3Q19 at $\omega=1.9$
the Jacobians of the two formulations differ by $0.18$ for the standard matrix and
agree to $2\times10^{-15}$ for the recipe.

\section{Cost of the collision kernels}
\label{sec:kernels}

The kernels of Table~\ref{tab:cost} are compiled implementations of the collision
that share the axis-by-axis central-moment transforms and differ only in the
relaxation step: diagonal in the non-orthogonal basis, with the conjugated matrix
$\Am^{-1}\Lam\Am$ of Corollary~\ref{cor:impl}, or through explicit transforms to and
from the orthogonal basis. They agree with the reference implementation to
$6\times10^{-16}$. Computed with the conjugated matrix, the orthogonal scheme costs
nothing extra on D2Q9 with one shear rate, where the two matrices coincide,
$8$--$10\%$ with an independent bulk rate, and on D3Q27 at most $2\%$ in cache and
nothing measurable from memory, where all kernels are bandwidth bound; explicit
transforms to the orthogonal basis cost $20\%$ in cache.

\begin{table}[tbp]
\centering
\caption{Collision kernels on one core of an Apple M1 (clang, \texttt{-O3}):
nonzero coefficients of the relaxation step and throughput in million lattice
updates per second, median of nine runs averaged over two sessions, with the
populations resident in cache or streamed from main memory. The kernels share the
central-moment transforms and differ only in the relaxation: diagonal in the
non-orthogonal basis (non.), with the conjugated matrix $\Am^{-1}\Lam\Am$ (folded),
or through explicit transforms to and from the orthogonal basis (explicit). The two
D2Q9 kernels with $\omega_b=1$ and $11$ coefficients are the same code, and their
difference indicates the run-to-run variation.}
\label{tab:cost}
\begin{tabular}{llccc}
\toprule
 & relaxation & nnz & cache & memory\\
\midrule
D2Q9, $\omega_b=1$   & non.     & $11$  & $142$  & $103$\\
                     & folded   & $11$  & $140$  & $105$\\
                     & explicit & $34$  & $119$  & $94$ \\
D2Q9, $\omega_b=1.6$ & non.     & $11$  & $139$  & $106$\\
                     & folded   & $16$  & $125$  & $98$ \\
                     & explicit & $34$  & $119$  & $95$ \\
D3Q27                & non.     & $33$  & $16.7$ & $10.2$\\
                     & folded   & $45$  & $16.4$ & $10.5$\\
                     & explicit & $136$ & $13.4$ & $10.1$\\
\bottomrule
\end{tabular}
\end{table}

\bibliographystyle{elsarticle-num}
\bibliography{bibliography}

\end{document}